\documentclass[11pt,english]{article}
\pdfoutput=1
\usepackage{mathptmx}

\usepackage[T1]{fontenc}
\usepackage[latin9]{inputenc}
\usepackage{geometry}
\usepackage{verbatim}
\usepackage{float}
\usepackage{amsmath}
\usepackage{amsthm}
\usepackage{amssymb}
\usepackage{graphicx}
\usepackage{rotfloat}
\usepackage{setspace}
\usepackage[authoryear]{natbib}
\PassOptionsToPackage{normalem}{ulem}
\usepackage{ulem}
\usepackage[unicode=true,pdfusetitle,
bookmarks=true,bookmarksnumbered=false,bookmarksopen=false,
 breaklinks=false,pdfborder={0 0 0},pdfborderstyle={},backref=false,colorlinks=false]
 {hyperref}
\usepackage[dot]{bibtopic}

\makeatletter
\theoremstyle{plain}
\newtheorem{assumption}{\protect\assumptionname}
\theoremstyle{plain}
\newtheorem{prop}{\protect\propositionname}
\theoremstyle{plain}
\newtheorem{lem}{\protect\lemmaname}
\theoremstyle{plain}
\newtheorem{cor}{\protect\corollaryname}
\theoremstyle{plain}
\newtheorem{remark}{Remark}

\@ifundefined{date}{}{\date{}}
\usepackage{multicol}
\usepackage{setspace}
\usepackage{amsbsy}
\usepackage{chngcntr}

\newcommand\independent{\protect\mathpalette{\protect\independenT}{\perp}}
\def\independenT#1#2{\mathrel{\rlap{$#1#2$}\mkern2mu{#1#2}}}

\makeatother

\usepackage{babel}
\providecommand{\assumptionname}{Assumption}
\providecommand{\lemmaname}{Lemma}
\providecommand{\propositionname}{Proposition}
\providecommand{\corollaryname}{Corollary}

\usepackage{footmisc}

\begin{document}
\pagestyle{plain}

\thispagestyle{empty}

\setcounter{footnote}{0}

\renewcommand{\thefootnote}{\fnsymbol{footnote}}

\newcounter{daggerfootnote}
\newcommand*{\daggerfootnote}[1]{%
    \setcounter{daggerfootnote}{\value{footnote}}%
    \renewcommand*{\thefootnote}{\fnsymbol{footnote}}%
    \footnote[2]{#1}%
    \setcounter{footnote}{\value{daggerfootnote}}%
    \renewcommand*{\thefootnote}{\arabic{footnote}}%
    }

\begin{titlepage}

\noindent \begin{center}
~\\
\vspace{35bp}
{\LARGE{}Further Education During Unemployment}\footnote[1]{We are grateful for insightful comments from the editor and four anonymous reviewers. We also thank Burt Barnow, Damon Clark, Rajeev Darolia, Sue Dynarski, Nathan Grawe, Gordon Hanson, Philipp Kircher, David S. Lee, Mike Lovenheim, Jordan Matsudaira, Doug Miller, Olivia Mitchell, Ronni Pavan, Ceci Rouse, Haiyuan Wan, Abbie Wozniak, and participants at AASLE, APPAM, Binghamton, CEU, CIRANO-CIREQ, Duke, IZA ed workshop, Michigan State, Northwestern, NTA, Princeton, SOLE, UC Davis, UVA, and Zurich for suggestions and discussions. Lexin Cai, Amanda Eng, Rebecca Jackson, Hyewon Kim, Suejin Lee, Kathryn McGinnis, and Katherine Wen provided excellent research assistance. We are indebted to Lisa Neilson and the staff members at the Center for Human Resource Research at Ohio State University, the Ohio Department of Job and Family Services, and the Ohio Department of Higher Education for providing the data and answering our many questions. We thank Jeff Smith for generously sharing the National JTPA Study data. Financial support from the Cornell Institute of Social Sciences is gratefully acknowledged. All errors and opinions are our own.}\footnote[2]{The Ohio Longitudinal Data Archive is a project of the Ohio Education Research Center (\url{oerc.osu.edu}) and provides researchers with centralized access to administrative data. The OLDA is managed by The Ohio State University's Center for Human Resource Research (\url{chrr.osu.edu}) in collaboration with Ohio's state workforce and education agencies (\url{olda.ohio.gov}), with those agencies providing oversight and funding. For information on OLDA sponsors, see \url{https://chrr.osu.edu/projects/ohio-longitudinal-data-archive}.}
\par\end{center}\vspace{20bp}

\renewcommand{\thefootnote}{\arabic{footnote}}

\setcounter{footnote}{0}

\begin{multicols}{2} 

\begin{singlespace}
\begin{center}
{\large{}Pauline Leung}\footnote{{\footnotesize{}Email: pleung@cornell.edu}}{\large{}}\\
{\large{}Cornell University}{\large\par}
\par\end{center}
\end{singlespace}

\columnbreak

\begin{singlespace}
\begin{center}
{\large{}Zhuan Pei}\footnote{{\footnotesize{}Email: zhuan.pei@cornell.edu}}{\large{}}\\
{\large{}Cornell University}{\large\par}
\par\end{center}
\end{singlespace}

\end{multicols}\vspace{-20bp}

\ 

\begin{center}
{\large{}November 2026}{\large\par}
\par\end{center}
\begin{abstract}
\begin{singlespace} \vspace{-10bp}

{\normalsize{}Evidence on the effectiveness of retraining unemployed U.S. workers primarily comes from evaluations of training programs, a narrow avenue for skill acquisition. We use administrative records from Ohio to estimate the effects of retraining, broadly defined as postsecondary enrollment, via matching. Our simple method bridges two strands of the dynamic treatment effect literature that differ in whether the counterfactual allows later treatment or rules it out throughout the treatment window. Enrollees experience earnings gains of six percent three to four years after enrolling, following depressed earnings during the first two years. The gains are driven by industry-switchers, particularly to healthcare.}{\normalsize\par}

\bigskip{}



\end{singlespace}
\end{abstract}
\thispagestyle{empty}

\end{titlepage}

\pagebreak{}

\renewcommand{\thefootnote}{\arabic{footnote}}

\newcommand{\multfootsep}{\textsuperscript{,}}

\setcounter{footnote}{0}

\setlength{\abovedisplayskip}{3pt}

\setlength{\belowdisplayskip}{3pt}

\setlength{\abovedisplayshortskip}{3pt}

\setlength{\belowdisplayshortskip}{3pt}

\section{Introduction\label{sec:Intro}}

The U.S. labor market has become increasingly polarized in recent decades, as high- and low-skilled jobs grow at the expense of middle-skilled jobs that traditionally employ workers with moderate levels of education (\citealp{AutorEtAl2006,Autoretal2008}; \citealp{AutorDorn2013}; \citealp{Autor2015}). These patterns were exacerbated during economic downturns, leading to disproportionately high unemployment among workers without a college degree (\citealp{Katz2010}; \citealp{Hoynesetal2012}; \citealp{JaimovichSiu2020}). A long line of research shows that job displacement, especially during economic downturns, is associated with large and persistent earnings losses, adverse health outcomes, and negative impacts on the children of the unemployed (\citealp{JacobsonEtAl1993}; \citealp{CouchPlaczek2010}; \citealp{Krolikowski2018}; \citealp{SullivanandVonWachter2009}; \citealp{DavisvonWachter2011}; \citealp{Oreopoulosetal2008}; \citealp{StevensandSchaller2011}). To mitigate these social and economic costs, economists and policymakers have advocated for new skill acquisition through further education (\citealp{Holzer2015}). At the peak of the Great Recession, for example, the U.S. Departments of Labor and Education created the website opportunity.gov and encouraged state governments to contact unemployment insurance (UI) claimants and inform them of resources (e.g., federal financial aid) and institutions (e.g., community colleges) for reskilling. More recently, following bipartisan proposals, the One Big Beautiful Bill Act created Workforce Pell Grants for students in qualifying short-term job training programs, effective July 2026.

Much of our knowledge about the effects of further education for unemployed workers in the U.S. comes from evaluations of government-sponsored training programs (e.g., the Workforce Investment Act program, or WIA, now replaced by the Workforce Innovation and Opportunity Act, or WIOA), even though these programs only constitute one narrow avenue through which unemployed workers upgrade their skills. In reality, many more workers enroll directly in a local postsecondary institution such as a community college. For instance, in the fall of 2017, 4.5 million nontraditional (i.e., 25 years old or above) undergraduate students were enrolled nationwide, compared to 1.6 million participants in the largest U.S. training program in 2016-17 (WIOA Adult and Dislocated Worker programs), of whom only a minority received training services (\citealp{NCES_Digest2018}; \citealp{WIA_Databook_2016}). While a large literature analyzes the effects of community college education (see \citealp{KaneRouse1999} and \citealp{BelfieldBailey2011} for reviews), few studies focus on unemployed workers, a policy-relevant group that differs from other community college attendees. Unemployed workers tend to be older and more experienced, but they have different opportunity costs, face different labor market barriers, and therefore may see different returns to further education. The only studies that directly examine unemployed workers enrolled in community colleges are \citet{Jacobson_etal2005_JE,Jacobson_etal2005_ILR}, which focus on long-tenured Washington state workers laid off in the early 1990s. This leaves a twenty-year research void on this important topic.

Our study seeks to fill this void. We estimate the labor market effects of retraining among unemployed workers, where retraining is broadly defined as enrollment in a postsecondary institution.\footnote{We use the word ``retrain'' in accordance with its meaning in the \href{https://dictionary.cambridge.org/us/dictionary/english/retrain}{Cambridge Dictionary}: to learn new skills so you can do a different job, or to teach someone a new skill so that they can do a different job. Workers can ``retrain'' irrespective of their educational background.} We link together high-quality administrative data from the state of Ohio, which include UI claims, quarterly wage records, course enrollment and credential data from all in-state public higher education institutions (including community colleges and technical centers), and WIA records. By following unemployed workers who filed a UI claim between 2004 and 2011, we observe that indeed the majority of retraining does not occur within the context of a narrowly defined training program: in our data, nearly 88,000 workers enroll in public postsecondary institutions following a layoff, compared with 27,000 workers who retrain through WIA.

We also tackle methodological issues along the way of our empirical inquiry. To estimate the effects of retraining, we use a matching method that compares the labor market outcomes of unemployed workers who pursue further education (enrollees) versus observably similar workers who do not (matched non-enrollees) within two years after layoff. Because workers enroll at different times, a standard matching estimand will only identify the effect of enrolling now versus not enrolling now but possibly later (the ``treatment-now-versus-no-treatment-now'' effect; e.g., \citealp{Sianesi2004}). It does not, in general, identify the effect we are interested in---that of enrolling versus not enrolling at any point in the two-year window (the ``treatment-versus-no-treatment'' effect). We show that, with a testable additional assumption regarding the selection into training that is implied by models in \citet{Heckman1978}, \citet{AshenfelterandCard1985}, and \citet{HeckmanandRobb1985JOE,HeckmanandRobb1985book}, we can identify a lower bound of the treatment-versus-no-treatment effect with a simple modification of the standard estimand. The estimated lower bound appears to be tight in our empirical context.

Our matching specification is informed by the large literature that uses selection-on-observables designs to evaluate training programs in both the U.S. and international contexts (see \citealp{McCalletal2016_Handbook} for a comprehensive review). Moreover, to support our specification, we have conducted our own validation analysis in the spirit of \citet{LaLonde1986}, using data from the National Job Training Partnership Act Study (NJS) (details can be found in our previous working paper \citealp{LeungandPei2020}). This analysis, which builds on the influential work by \citet{Heckman_etal1997,Heckmanetal1998}, \citet{Heckmanetal1998_ECMA}, and \citet{HeckmanSmith1999}, evaluates the ability of various models and specifications (including those based on machine-learning) to recover a causal effect. We find that when we have a sample of workers recently attached to the labor market and incorporate detailed earnings histories linearly into the covariate set, conventional (logit-based) propensity score matching performs well, indicating the plausibility of the underlying conditional independence assumption. 

We graphically present the average earnings trajectories of enrollees and matched non-enrollees in the five years before and four years after enrollment. The trajectories reveal little difference in earnings pre-enrollment, followed by temporarily depressed earnings of enrollees while they are in school (the ``lock-in'' effect), and sustained positive effects thereafter. Overall, we estimate that the lower bound earnings effect among enrollees is \$348 per quarter, or about six percent, in the third and fourth years after enrolling. A decomposition of this earnings gain reveals that retraining affects earnings mostly at the extensive margin. While the magnitudes of enrollment effects are heterogeneous across various subgroups, we consistently observe earnings gains four years after enrolling. Following an early subset of workers for a longer period, we find that the retraining effect persists and widens to 13 percent at the end of a ten-year horizon.

Another advantage of our study relative to existing training program evaluations is our ability to look into the ``black box'' of retraining. That is, we observe the courses taken and credentials received by enrollees in our sample, which allows us to explore the types of training underlying our estimates. A simple accounting exercise suggests that the enrollment effects are driven by workers who train and subsequently find employment in new industries post-layoff, particularly the healthcare sector.

This paper makes the following contributions. First, it bridges two largely separate strands of empirical literature on training programs and on community colleges by studying the policy-relevant unemployed worker population that intersects with both. As mentioned above, the U.S. training literature focuses mainly on evaluating government-sponsored programs, which finds mixed results.\footnote{For WIA, recent experimental and non-experimental evaluations find zero to long-lasting negative effects of training for dislocated workers (\citealp{Heinrichetal2008}; \citealp{McConnelletal2016}; \citealp{Fortsonetal2017}; \citealp{Anderssonetal2013}). For the Trade Adjustment Assistance (TAA) program, which provides training to workers affected by trade, one non-experimental evaluation finds initially large negative effects that fade to zero over a four-year period, while another study utilizing quasi-random variation on TAA petition approvals finds positive effects, though the two studies present estimates of different quantities that are not directly comparable (\citealp{Schochet2012}; \citealp{Hyman2018}).} The community college literature primarily estimates the earnings gain associated with specific credentials, finding that associate degrees yield earnings gains of about 18 to 26 percent relative to no degree, mixed effects of other credentials, and positive effects of healthcare-related programs (see review by \citealp{Belfield_Bailey2017a}). Since many community college students do not earn a credential (``non-completers''), these estimates are not directly comparable to our enrollment effect. In the one study that reports both credential and per-credit effects---\citet{Jepsen_etal2014}---we calculate that enrollment raises earnings by roughly 16 percent for women and 6 percent for men.\footnote{We weight their estimates for each credential type by credential shares and, for non-completers, multiply their per-credit estimate by 20 credits---the average among non-completers from \citet{Belfield_Bailey2017a}, since \citet{Jepsen_etal2014} do not report it.} As noted above, within the community college literature, \citet{Jacobson_etal2005_JE,Jacobson_etal2005_ILR} are the only studies that look specifically at unemployed workers. Their preferred regression model suggests that enrollment increases earnings between six and eight percent (their main earnings effect finding of nine to thirteen percent is for one year of full-time enrollment, and we scale it based on the average course load in their data), but their estimates are sensitive to the model used.\footnote{Most of the community college research using similar administrative earnings data relies on a fixed effects model. In our setting, we find differential pre-trends between enrollees and non-enrollees that would render estimates from fixed effects models biased, similar to \citet{Jacobson_etal2005_JE}. We also find that fixed effects specifications yield biased estimates in our validation exercise (see \citealp{LeungandPei2020}).}

Second, we contribute methodologically to the dynamic treatment effect literature by connecting the studies that estimate the treatment-now-versus-no-treatment-now and treatment-versus-no-treatment effects of training. In particular, by modifying the treatment-now-versus-no-treatment-now estimand per \citet{Sianesi2004}, we can use ``static'' propensity score matching to bound the treatment-versus-no-treatment effects on the treated that are typically identified with dynamic estimands as in \citet{Lechner2009JBES} and \citet{LechnerandMiquel2009}. Our simple estimator can be implemented with off-the-shelf software commands and avoids the inferential challenges their dynamic counterparts encounter. 

Finally, this paper sheds light on the effects of retraining for a recent period, which includes the Great Recession and covers a wide range of economic conditions, and in a Rust Belt state characterized by movement away from declining manufacturing industries. Recency of data is important: labor market trends such as the rise in automation and trade in past decades (\citealp{Autoretal2013}; \citealp{Autoretal2014}; \citealp{AcemogluRestrepo2020}) may have impacted training effects particularly in former manufacturing centers, which makes our estimates more informative for current policy-making relative to those by \citet{Jacobson_etal2005_JE,Jacobson_etal2005_ILR} from Washington state in the early 1990s. The economic boom and bust in our sample period allow us to speak to the literature examining the dependence of educational returns on labor market conditions. Consistent with previous studies (\citealp{LechnerWunsch2009}; \citealp{Kahn2010}; \citealp{Oreopoulosetal2012}), we find retraining leads to larger average earnings gains for those enrolled during the Great Recession, who sought jobs afterwards in a thawing labor market.

\section{Institutional Background\label{sec:Institutional-Background}}

\subsection{Unemployment Insurance}

In this paper, we identify unemployed workers as those who claim unemployment insurance. To be eligible for UI, workers must have lost a job through no fault of their own and have sufficient earnings and work weeks prior to job loss. In Ohio, workers must have worked at least 20 weeks and have an average weekly wage of about \$200 in the one-year period that begins five calendar quarters prior to job loss. As a result, our study population consists of UI claimants previously attached to the labor force.

While workers generally need to actively search for jobs and be available to work in order to continue receiving benefit payments, they can pursue ``approved training'' opportunities without losing UI eligibility. States vary in their definitions of approved training, though it generally includes vocationally-oriented or basic education training. According to \citet{NASWA2010}, the Ohio unemployment agency automatically approves all training through workforce programs and has 7,000 courses listed as approved. It also approves academic courses that do not lead to a specific occupation on a case-by-case basis.

\subsection{Postsecondary Institutions}

Our study focuses on the impact of classroom training, which can take place in different settings. First, workers may choose to attend community colleges and enroll in courses that may lead to an associate degree or sub-associate credentials such as certificates. Alternatively, workers may enroll at technical centers. Technical centers typically offer occupation-specific programs that may lead to a state license or other credentials. Examples include a state license for practical nursing or professional certification in welding.

The cost of attendance varies by institution and program. According to the Integrated Postsecondary Education Data System (IPEDS), the average tuition and fees across Ohio institutions were approximately \$6,400 per year in 2010. However, since unemployed workers are often financially constrained, they are likely to be eligible for and rely on several forms of financial assistance. First, workers may be eligible for federal financial aid such as Pell Grants, subsidized loans, and tuition tax credits.\footnote{Although the amount of federal aid typically depends on income from the year prior during our sample period, UI claimants may qualify for simplified needs tests or automatic zero expected family contribution starting in 2009.} Second, they may obtain training funding from workforce programs like WIA or TAA. In WIA, eligible participants may receive an Individual Training Account (ITA) voucher that can be used towards approved training (fewer than 20 percent of WIA participants receive training services, while the others only receive ``core'' or ``intensive'' services such as assistance in job search, placement, employment and career planning). The TAA program provides tuition assistance to workers affected by import competition. To understand the relative sizes of the various sources of financial assistance, \citet{BarnowSmith2016} report that in 2014, Pell Grants for those who pursued vocational education totaled \$8.2 billion. In contrast, the expenditures for the WIA Dislocated Worker program and TAA were only \$1.2 and \$0.3 billion, respectively.

\section{Data, Analysis Sample, and Descriptive Statistics\label{sec:Data}}

Our analysis primarily draws on several administrative data sources from Ohio: 1) UI claim records, 2) student records from public postsecondary institutions, including community colleges and technical centers, 3) quarterly wage records, and 4) WIA participant records.

For our analysis of labor market effects, we study workers who file an eligible UI claim between 2004 and the third quarter of 2011. To focus on workers who seek further education after unemployment, we exclude those who enroll at any point within two years prior to layoff.\footnote{This restriction eliminates roughly 147,000 claims. The excluded claimants are younger with lower tenure and lower pre-layoff earnings.} Our analysis sample contains 1.9 million claims, coming from 1.3 million unique individuals (see Appendix \ref{sec:Data Appendix} for details on sample construction and data elements). The UI records contain the claim date, demographics (gender, race, number of dependents, age, and zip code), and prior job information (industry and occupation).\footnote{The number of dependents is recorded because the maximum UI benefit amount is higher when a claimant has more dependents. However, only workers with high enough prior earnings receive the maximum UI amount. Since many workers' UI benefits do not change with dependents, this measure likely understates the true number of dependents.}

Our schooling data cover all public postsecondary institutions in Ohio. The Higher Education Information system (HEI) records contain enrollment information for 37 public two- and four-year institutions, though we focus on those who first enroll in a two-year institution in this analysis. In the HEI data, we observe terms enrolled, courses taken, and degrees or other credentials obtained (i.e., graduate or professional, bachelor's, associate, or less than two-year awards). We also have student records from 53 publicly funded technical centers through the Ohio Technical Centers (OTC) database. In the OTC data, the courses offered range from one-day courses to certificate programs that last several years. We observe the dates of enrollment in courses as well as any credentials obtained. Despite the expansive coverage of the HEI and OTC data, we do not know whether a worker enrolls in a private institution. While we acknowledge this to be a limitation of our paper as our ``non-enrollees'' may in fact enroll in a program we do not observe, many studies in the training literature also suffer from similar issues. For example, a survey of trainees in the WIA Gold Standard Evaluation reveals that 28 percent of the training was not funded by WIA (\citealp{Fortsonetal2017}).\footnote{Specifically, \citet{Fortsonetal2017} find that while 43 percent of workers assigned to the ``full WIA'' arm self-reported to have trained, only 31 percent in the arm received WIA funding. This implies that $(43-31)/43=28$ percent of the trainees sought training outside WIA.} In comparison, the share of observed non-enrollees in our data enrolling in a private institution is likely to be smaller. According to IPEDS, among students 25 or older at two-year or lower institutions in the fall of 2007, less than 10 percent were enrolled in private institutions, of which 94 percent were in for-profit institutions. The unobserved private school enrollment is likely to lead us to understate the effect of training if there is private enrollment in our matched comparison sample---\citet{CelliniTurner2019} show that attending for-profit institutions leads to a positive, but statistically insignificant, earnings gain. We provide more details on the potential bias from unobserved private enrollment in Appendix \ref{subsec:courses_appendix}.

We construct our main outcome variables using quarterly earnings data from the state's UI system (we discuss how out-of-state earnings may impact our estimates in Appendix \ref{sec:Data Appendix}). In addition to earnings, we also observe, through the third quarter of 2017, number of weeks worked from 2003 and industry for each private sector employer from 1995 (unlike the UI claim records, the quarterly wage data contain no information on occupation). The long earnings history allows us to observe at least three years of pre-layoff earnings. We also use these data to construct measures of pre-layoff job tenure and outcomes like industry switching.

Finally, we observe whether a worker in our analysis sample is in the WIA Standardized Record Data, which cover participants of WIA Adult, Youth, and Dislocated Worker programs. We use information on the dates of WIA training to identify the subset of enrollees observed in the HEI and OTC data who received WIA training services. Next we state our definition of an enrollee and provide descriptive statistics on enrollee demographics, the timing of enrollment, and enrollment characteristics.

\textbf{Definition of an enrollee} We define an enrollee as a worker who enrolls in a community college or a technical center within two years of layoff.\footnote{Following a reviewer suggestion, we remove those who initially enroll in a four-year institution from the analysis entirely (i.e., they are not in treatment or comparison groups).}\multfootsep\footnote{\label{fn:enrollment-quarter}Since we observe enrollment term rather than date of enrollment in the HEI data, we approximate enrollment terms winter, spring, summer, and fall to the first, second, third, and fourth calendar quarters, respectively. Since enrollment typically occurs in the fall and spring terms, and those terms are likely to begin earlier than the corresponding fourth and second calendar quarters, we are likely to report an enrollment start date that is on average later than when workers actually begin schooling. By 2013, all Ohio public colleges had switched to the semester system, which eliminates the corresponding winter quarter.} The two-year window is motivated by the fact that UI benefits were available for up to 99 weeks during the Great Recession, which our analysis period covers. We choose to have a consistent definition of treatment by using the same two-year window throughout the entire sample, even though UI is only available for 26 weeks under normal economic conditions. One may be concerned as to whether workers receiving 26 weeks of UI are still unemployed two years after layoff, but as we show in Section \ref{sec:Results}, enrollees on average have continuously depressed earnings between layoff and enrollment regardless of how soon they enroll.

\textbf{Who enrolls? }Panel A of Table \ref{tab:descriptives} presents descriptive statistics for our UI claimant sample by enrollment status. Of the 1.9 million claims in our data, 71,745, or 4 percent, are followed by enrollment in a public postsecondary (two-year or less) institution (see \citealp{Minayaetal2023} for related statistics). While women make up only 34 percent of UI claimants, they are better represented among enrollees, at 44 percent. Compared with 13 percent among non-enrollees, Blacks make up 18 percent of the enrollee population. In terms of prior job characteristics, enrollees are less likely to have worked in manufacturing and transportation, have lower job tenure, are younger, and have lower prior earnings (all earnings are expressed in 2012 dollars).

\textbf{When do workers enroll, and what are their enrollment characteristics?} Table \ref{tab:enrollment_chars} shows that workers take time to enroll---on average 3.7 quarters after layoff---and the mean enrollment length is 4.5 terms. Enrollees in our UI claimant sample mostly attend community colleges (87 percent). 90 percent of enrolled workers take at least one occupational course, and the average proportion of occupational courses is 60 percent, where courses are classified as occupational based on their Classification of Instructional Programs (CIP) code following the taxonomy by the National Center for Education Statistics. The overall credential receipt rate within four years of enrollment (including sub-baccalaureate awards, licenses, and industry credentials) is 26 percent. The vast majority of degrees or credentials obtained are associate degrees or lower: sub-associate credentials account for 58 percent of total credential receipts, associate degrees for 40 percent, and bachelor's and graduate degrees make up the remaining 2 percent (while we focus on workers who begin enrollment in a community college or technical center, some eventually go on to four-year institutions).

\section{Identification, Estimation, and Empirical Implementation\label{sec:Empirical-Strategy}}

The empirical challenge in estimating labor market effects of enrollment stems from the differences in the characteristics of workers who do and do not enroll that relate to their future earnings potential. Following a long line of research in training program evaluation as reviewed by \citet{McCalletal2016_Handbook}, we adopt a selection-on-observables research design to measure the causal effects of retraining. Because enrollment timing relative to layoff varies across workers, we rely on a dynamic framework for partial identification, with our main proposition providing a treatment effect lower bound. We discuss the underlying assumptions and define the treatment effect parameters of interest in Section \ref{subsec:Identification} below.

Before we proceed, we highlight several rationales for why matching may ``work'' in our context, given the (often justified) skepticism towards it. First, our sample construction helps to mitigate the ``Ashenfelter's dip'' problem, a major challenge in training program evaluation. Ashenfelter's dip refers to the phenomenon in most evaluations of U.S. training programs wherein trainees experience (on average) an earnings dip prior to training, while non-trainees do not. This is likely because the decision to retrain is often a reaction to transitory shocks (\citealp{Ashenfelter1978}). A prime example of such shocks is the loss of employment. By starting from the set of recently unemployed workers in Ohio, we shut down job loss as a channel that triggers training, as both enrollees and non-enrollees in our sample have experienced it.

Second, in our matching specification we take advantage of the information available on past labor market histories from our rich administrative data, which makes the selection-on-observables assumption plausible as recent studies have argued. As \citet{Anderssonetal2013} succinctly state, ``{[}m{]}otivated workers, and high ability workers, should do persistently well in the labor market; if so, conditioning on earlier labor market outcomes will remove any selection bias that results from motivation and ability also helping to determine training receipt.'' \citet{Anderssonetal2013} also find that adding firm fixed effects does not change causal estimates relative to specifications which incorporate detailed labor market histories. Similarly, in extensive empirical Monte Carlo simulations based on German administrative data, \citet{LechnerWunsch2013} find that including variables on firm characteristics, industry- and occupation-specific experience, health, program compliance, desired job characteristics, and detailed regional information does not further reduce bias relative to only using basic demographics and labor market histories. Finally, \citet{Caliendoetal2017} find that controlling for typically unobserved non-cognitive traits adds little beyond past labor market histories.

Third, our matching specification is guided by a validation exercise in the spirit of \citet{LaLonde1986}, \citet{Heckman_etal1997,Heckmanetal1998}, and \citet{Heckmanetal1998_ECMA}. As described in detail in a previous version of this paper (\citealp{LeungandPei2020}), we use data from the NJS to assess whether more flexible machine-learning-based specifications offer better performance than a conventional (logit) propensity score model where the terms enter linearly. We find that conventional propensity score matching methods perform competitively even without the use of machine learning algorithms, and are able to recover a causal effect  in subsamples that more closely resemble our Ohio study population---namely workers with previous labor market attachment, for whom past earnings are likely to be predictive of future prospects.

Fourth, the specification informed by our validation study achieves high match quality. In Section \ref{subsec:Estimation-Strategy-and}, we show overlapping support in the propensity score distributions (except for outliers that amount to one percent of the enrollee sample) and covariate balance across enrollees and matched non-enrollees.

Finally, we believe matching to be better suited than other empirical methods for our setting. In Appendix \ref{sec:Alternative-Identification-Strat}, we explore the possibility of adopting alternative research designs---which may identify different causal parameters than matching---including fixed effects models, the use of a distance instrument, and a regression discontinuity design based on layoff timing. We discuss why we cannot use them for our analysis.

Despite these reasons in favor of using matching for our analysis, doubts may still linger over the validity of the selection-on-observables assumption. Chief among them is the possibility that even within a matched pair with the same labor market history, the non-enrollee chooses not to enroll because she expects to be recalled to her previous employer or has a job offer in hand to start in the future (\citealp{Sianesi2004}; \citealp{FredrikssonJohansson2008}).\footnote{Our UI claims data do not contain an indicator of whether a claimant has a known recall date.} As we show next, our main identification result provides a lower bound on a dynamic treatment effect parameter, and the presence of such ``forward-looking'' workers may further bias our estimate downward, against finding an enrollment effect. 

\subsection{Parameters of Interest, Identification and Estimation\label{subsec:Identification}}

The main assumption underlying our method is the conditional independence assumption (CIA), or ``unconfoundedness''. For expositional purposes, we begin with the simple \emph{static} case where the enrollment decision takes place at one point in time for all laid-off workers. Let $D$ denote whether a worker enrolls in school, with $D=1$ if the worker enrolls and $D=0$ if she does not. Let $Y(\cdot)$ denote the potential post-enrollment earnings of the worker: $Y(1)$ is the worker's potential future earnings if she enrolls, and $Y(0)$ is her potential future earnings if she does not enroll. The observed outcome is $Y=Y(1)D+Y(0)(1-D)$.

The static unconfoundedness assumption standard in the matching literature is 
\begin{equation}
Y(0)\independent D|\mathbf{X}\label{eq:static_unconfound}
\end{equation}
where $\mathbf{X}$ is a vector of observed covariates realized at or prior to enrollment. Together with a common support condition---the propensity score $p(\mathbf{X})\equiv\Pr(D=1|\mathbf{X})$ is less than one (almost) everywhere on the support of $\mathbf{X}$---CIA implies the identification of the treatment effect on the treated (TOT) parameter via propensity score matching:
\begin{align}
E\left[Y|D=1\right]-E\left[E[Y|D=0,p(\mathbf{X})]|D=1\right] & =\underbrace{E[Y(1)-Y(0)|D=1]}_{\text{TOT}}.\label{eq:static_identification}
\end{align}

In our study, the ``treatment'' (enrollment) is allowed to occur over a period of two years, a complication that necessitates a \emph{dynamic} variant of the framework above. An important consideration in the dynamic setting is the treatment effect parameter of interest. Several earlier studies (\citealp{Sianesi2004}; \citealp{FredrikssonJohansson2008}; \citealp{Biewen_etal2014}) rely on the dynamic counterpart of the conditional independence assumption and use an estimand similar to that in equation (\ref{eq:static_identification}) to identify the treatment effect of enrolling in one period versus not enrolling in that period but possibly later (the treatment-now-versus-no-treatment-now effect). This is different from the dynamic version of the TOT parameter we are interested in, which is the effect of enrolling versus not enrolling in the two-year window (the treatment-versus-no-treatment effect). This treatment-versus-no-treatment parameter gets us closer to the effect of enrolling versus not enrolling within the two-year window or thereafter, which, as \citet{McCalletal2016_Handbook} point out, can be more easily incorporated into a cost-benefit analysis than the treatment-now-versus-no-treatment-now parameter.\footnote{The treatment-now-versus-no-treatment-now effect may be the more relevant parameter in other contexts such as the earnings impact of job displacement (\citealp{Krolikowski2018}).} 

In this paper, we modify the method used by studies that estimate the treatment-now-versus-no-treatment-now effect. This modified method preserves the simplicity of \emph{static} propensity score matching and can partially identify the \emph{dynamic} TOT parameter under an additional testable assumption. In the remainder of this subsection, we state our assumptions and identification results. We discuss the connection between our method and alternatives in the dynamic treatment effect literature in the next subsection.

For ease of exposition, we focus on a two-period case here and leave the more general multi-period case to Appendix \ref{sec:identification_appendix} (we have an enrollment window of eight quarters in our empirical setting). In the two-period case, $D$ equals one if a worker starts training either in the first or second period post-layoff. We use the binary variable $D_{1}$ to denote whether a worker begins enrollment in period 1, and $D_{2}$ to denote whether she begins enrollment in period 2. It follows that $D=D_{1}+D_{2}$ with two implications. First, $D=0$ implies that $D_{1}=0$ and $D_{2}=0$ (a worker who never enrolls does not start training in either period). Second, at most one of $D_{1}$ and $D_{2}$ can equal one, meaning that $D_{1}=1$ implies $D_{2}=0$ (a worker who already enrolls in period 1 is no longer ``at risk'' for starting enrollment in period 2), and vice versa (a worker who starts enrollment in period 2 cannot have enrolled in period 1).

We use $\mathbf{X}^{1}$ and $\mathbf{X}^{2}$ to denote the vectors of conditioning covariates available at the beginning of the first and second period and before the realization of $D_{1}$ and $D_{2}$, respectively. The superscript notation, consistent with \citet{AbbringandHeckman2007}, indicates that the covariates incorporate cumulative information up to the beginning of a time period before the training decision for that period. More concretely, we think of $\mathbf{X}^{1}$ as incorporating the relevant covariates available at the beginning of period 1 and before the realization of $D_{1}$; $\mathbf{X}^{2}$ includes all the variables in $\mathbf{X}^{1}$, and also additional variables realized after $D_{1}$ but prior to $D_{2}$.\footnote{The conditioning set notation in \citet{Lechner2009JBES} and \citet{LechnerandMiquel2009} is similar: our $\mathbf{X}^{1}$ and $\mathbf{X}^{2}$ correspond to their $\text{\ensuremath{\underline{X}_{0}}}$ and $\text{\ensuremath{\underline{X}_{1}}}$.} Denoting earnings of period $t$ by $Y_{t}$, which realizes at the end of the period, the additional variables in $\mathbf{X}^{2}$ but not $\mathbf{X}^{1}$ typically include $Y_{1}$. In summary, the sequence of variable realizations is:
\[
\mathbf{X}^{1}\text{ (realized at layoff)}\to D_{1}\to Y_{1}\to\mathbf{X}^{2}\to D_{2}\to Y_{2}\to Y_{3}\to Y_{4}\to\dots
\]
Note that because training can only take place during period 1 or 2 in the two-period case, it is not necessary to include the conditioning set $\mathbf{X}^{t}$ and treatment $D_{t}$ when $t>2$ in the sequence above.

Finally, we use $Y_t(1)$ and $Y_t(0)$ to denote the potential earnings at the end of period $t$ corresponding to the treatment $D$. While it is common in the dynamic treatment effect literature to define potential outcomes corresponding to $D_s$ or just $s$ (e.g., \citealp{AbbringandHeckman2007}; \citealp{LechnerandMiquel2009}; \citealp{Fitzenbergeretal2023}), we define them with respect to the single binary $D$ for simplicity. We choose this notation because our primary interest is the treatment effect of enrolling ($D$) rather than the treatment effect of beginning enrollment in a particular period ($D_s$), and because it maintains consistency with the static case.

Our main identification result uses the following observation rule: for a worker who begins enrollment in period $s$ ($s=1,2$), the observed post-treatment ($t\geqslant s$) outcome is $Y_t=Y_t(1)$; for a worker who enrolls in neither period, $Y_t=Y_t(0)$ for all $t$. Missing from the observation rule is the relation between observed and potential \emph{pre}-treatment outcomes among the treated population. While it is tempting to write $Y_t=Y_t(0)$ for a period-$s$ enrollee when $t<s$, doing so requires an additional assumption of no anticipation. No-anticipation is not needed for our main identification result, but we do invoke it to derive testable implications of our assumptions---more details are in Appendix \ref{subsec:Discussion-of-Assumption2}.

Here are the assumptions for our main identification result. The first is a dynamic CIA assumption
\begin{assumption}
\label{assu:S2_seq_unconfound}a) $Y_t(0)\independent D_{1}|\mathbf{X}^{1}$ for $t\geqslant1$ and b) $Y_t(0)\independent D_{2}|\mathbf{X}^{2},D_{1}=0$ for $t\geqslant2$.
\end{assumption}
\noindent Assumption \ref{assu:S2_seq_unconfound} and variations thereof are standard in the literature---they are referred to as ``sequential randomization'' (e.g., \citealp{Robins1997}) or ``weak dynamic conditional independence'' (e.g., \citealp{Lechner2009JBES}). In our context, it says that within the ``at-risk'' set of workers who have not yet enrolled and are therefore able to begin enrollment at the start of period $s$ ($s=1,2$), the potential future outcome $Y_t(0)$ is independent of a worker's enrollment decision conditional on the information available. Following an argument analogous to the static case, Assumption \ref{assu:S2_seq_unconfound} allows for the identification of the training effect for period-2 enrollees when conditioning on $\mathbf{X}^{2}$ and $D_{1}=0$ (i.e., the effect of $D_{2}$). By simply conditioning on $\mathbf{X}^{1}$, conditional independence also allows for the identification of the effect of enrolling in period 1 versus not enrolling in period 1 but possibly in period 2 (i.e., the effect of $D_{1}$) with the estimand:
\begin{equation}
E\left[Y_t|D_{1}=1\right]-E\left[E[Y_t|D_{1}=0,\mathbf{X}^{1}]|D_{1}=1\right]\label{eq:now-v-later}
\end{equation}

\noindent for $t\geqslant1$. However, it is more complex to identify the TOT parameter of training versus no-training (i.e., the effect of $D$) for period-1 enrollees. We invoke an additional assumption to bound it from below:
\begin{assumption}
\label{assu:S2_potential_outcome}$E[Y_t(0)|D_{1}=0,D_{2}=1,\mathbf{X}^{1}]\leqslant E[Y_t(0)|D_{1}=0,D_{2}=0,\mathbf{X}^{1}]\text{ for \ensuremath{t\geqslant1}}.$
\end{assumption}
\noindent Assumption \ref{assu:S2_potential_outcome} says that among workers who have the same observable characteristics at the beginning of period 1, those who begin enrollment in period 2 have lower average potential future earnings absent training than their never-enrolled counterparts. It reflects the idea that workers who select into training tend to have lower opportunity costs in doing so. As we discuss in Appendix \ref{subsec:Discussion-of-Assumption2}, under additional conditions that are not substantively more restrictive, this assumption is consistent with classic models of selection into training by \citet{Heckman1978}, \citet{AshenfelterandCard1985}, and \citet{HeckmanandRobb1985JOE,HeckmanandRobb1985book}. The frameworks in these papers also lead to testable implications in pre-treatment data: to the extent that the earnings process has positive serial dependence, the observed period-1 earnings $Y_1$ among period-1 non-enrollees can proxy for their future potential outcome $Y_t(0)$ (as mentioned above, this test requires an additional no-anticipation assumption, but we point out in Appendix \ref{subsec:Discussion-of-Assumption2} that the earnings process specifications in the seminal studies referenced above implicitly maintain this assumption). Thus, we can test Assumption \ref{assu:S2_potential_outcome} by comparing the $Y_{1}$ of period-2 enrollees and never-enrolled workers with similar $\mathbf{X}^{1}$. The test can be implemented via propensity score matching, and we present evidence in support of Assumption \ref{assu:S2_potential_outcome} in Section \ref{subsec:Overall-Returns}.

Assumption \ref{assu:S2_potential_outcome} allows us to bound the TOT effect for period-1 enrollees from below with a simple modification of the treatment-now-versus-no-treatment-now estimand in (\ref{eq:now-v-later}):
\begin{equation}
E\left[Y_t|D_{1}=1\right]-E\left[E[Y_t|D=0,\mathbf{X}^{1}]|D_{1}=1\right].\label{eq:lb_X1}
\end{equation}
The only difference between (\ref{eq:now-v-later}) and (\ref{eq:lb_X1}) is the comparison group: the comparison group in (\ref{eq:now-v-later}) consists of workers who did not enroll in period 1 ($D_{1}=0)$, while the comparison group in (\ref{eq:lb_X1}) consists of workers who did not enroll in either period ($D=0$). Intuitively, it is impossible to identify the TOT with the treatment-now-versus-no-treatment-now estimand (\ref{eq:now-v-later}) because we do not observe $Y_t(0)$ of the later-enrollees (i.e., workers in its comparison group who enroll in period 2). By going from (\ref{eq:now-v-later}) to (\ref{eq:lb_X1}), we replace the later-enrollees by non-enrollees with similar $\mathbf{X}^{1}$, for whom $Y_t(0)$ is the observed $Y_t$. Assumption \ref{assu:S2_potential_outcome} says that these non-enrollees have a higher average $Y_t(0)$, allowing (\ref{eq:lb_X1}) to provide a lower bound for the TOT parameter.

Directly implementing estimand (\ref{eq:lb_X1}) is subject to the curse of dimensionality when $\mathbf{X}^{1}$ contains many covariates. The use of a propensity score is the usual remedy for overcoming this challenge. However, a strong CIA assumption for the estimand (i.e., $Y_t(0)\independent D_{2}|D_{1}=0,\mathbf{X}^{1}$ for $t\geqslant 1$) would imply that Assumption 2 holds with equality, which is potentially inconsistent with classic models of training selection as discussed above. (Note that this CIA assumption restricted to $t\geqslant2$ is different from Assumption \ref{assu:S2_seq_unconfound}(b); the former conditions on $\mathbf{X}^{1}$, the information available just prior to the realization of $D_1$, whereas the latter conditions on $\mathbf{X}^{2}$, the information available just prior to the realization of $D_2$.) As we are unwilling to make such a strong independence assumption, we do not have a standard propensity score theorem at our disposal. Fortunately, as we show in Appendix \ref{subsec:Details-on-Bounds}, propensity score matching preserves inequality.\footnote{By similar reasoning, propensity score matching also preserves inequality in the simple static setting. Specifically, if selection into treatment is negative conditioning on covariates: $E[Y_t(0)|D=0,\mathbf{X}]\geqslant E[Y_t(0)|D=1,\mathbf{X}]$, then selection into treatment is also negative conditioning on the propensity score: $E[Y_t(0)|D=0,p(\mathbf{X})]\geqslant E[Y_t(0)|D=1,p(\mathbf{X})]$. This is a useful result, as it can help sign the population bias in propensity score matching.}

Formally, define the propensity scores $p_{1}(\mathbf{X}^{1})\equiv\Pr(D_{1}=1|D_{2}=0,\mathbf{X}^{1})$ and $p_{2}(\mathbf{X}^{2})\equiv\Pr(D_{2}=1|D_{1}=0,\mathbf{X}^{2})$, and our main identification result is
\begin{prop}
\label{prop:Prop_S2_lb}Under Assumptions \ref{assu:S2_seq_unconfound} and \ref{assu:S2_potential_outcome} and provided that $p_{1}(\mathbf{X}^{1}),p_{2}(\mathbf{X}^{2})<1$,

(a):
\begin{equation}
E\left[Y_t|D_{1}=1\right]-E\left[E[Y_t|D=0,p_{1}(\mathbf{X}^{1})]|D_{1}=1\right]\leqslant\underbrace{E[Y_t(1)-Y_t(0)|D_{1}=1]}_{\text{TOT}\text{ for }D_{1}=1\text{: }\delta_{1t}}\text{ for \ensuremath{t\geqslant1}}\label{eq:ID_ps_T1_lb}
\end{equation}
\begin{equation}
E\left[Y_t|D_{2}=1\right]-E\left[E[Y_t|D=0,p_{2}(\mathbf{X}^{2})]|D_{2}=1\right]=\underbrace{E[Y_t(1)-Y_t(0)|D_{2}=1]}_{\text{TOT}\text{ for }D_{2}=1\text{: }\delta_{2t}}\text{ for \ensuremath{t\geqslant2}};\label{eq:ID_ps_T2_lb}
\end{equation}
(b):
\[
\sum_{s=1}^{2}\left\{ E\left[Y_t|D_{s}=1\right]-E\left[E[Y_t|D=0,p_{s}(\mathbf{X}^{s})]|D_{s}=1\right]\right\} \Pr(D_{s}=1|D=1)\leqslant\underbrace{E[Y_t(1)-Y_t(0)|D=1]}_{\text{Overall TOT: \ensuremath{\delta_t}}}\text{ for \ensuremath{t\geqslant2}}.
\]
\end{prop}
All proofs are in Appendix \ref{sec:identification_appendix}. Part (a) of Proposition \ref{prop:Prop_S2_lb} consists of identification results for $\delta_{1t}$ and $\delta_{2t}$, the TOTs at time $t$ among period-1 and period-2 enrollees, respectively. Part (b) aggregates across the two enrollee populations and provides a lower bound for $\delta_t$, the overall TOT parameter at time $t$. For expositional ease, $t$, our time index here, refers to time relative to layoff. But we can also define $\delta$ with an alternative time index, such as time relative to the beginning of enrollment. Concretely, to obtain the overall TOT $\tau$ quarters after enrollment, we take an average of $\delta_{1\tau}$ and $\delta_{2(\tau+1)}$ weighted by the shares of period-1 and period-2 enrollees. Since enrollment is our treatment variable of interest, we use this alternative time indexing in most of our empirical analyses, which is consistent with the convention in event studies. 

The identification results in Proposition \ref{prop:Prop_S2_lb}(a) lead to standard propensity score matching estimators for the lower bound of $\delta_{1t}$ and for the value of $\delta_{2t}$. We can compute their corresponding asymptotic variances by following \citet{Abadie_Imbens2016}. To estimate the lower bound of the overall TOT $\delta_t$, we simply take an average of the two estimators weighted by the shares of period-1 and period-2 enrollees and compute its asymptotic variance accordingly (details in Appendix \ref{subsec:Details-on-Bounds}). 

A natural question that arises is how conservative the lower bounds from Proposition \ref{prop:Prop_S2_lb} are. It is easy to see that the tightness of our bound for $\delta_{1t}$ depends crucially on the share of later-enrollees, i.e., $\Pr(D_{2}=1|D_{1}=0,p_{1}(\mathbf{X}^{1}))$. When this share is zero, the lower bound estimand point identifies $\delta_{1t}$. Intuitively, if no one enrolls in period 2, the treatment-now-versus-no-treatment-now effect for $D_{1}=1$ becomes the treatment effect of enrolling in period 1 versus not enrolling in either period. The bounds are likely to stay informative when this share of later-enrollees is small, which is indeed the case in our empirical context.

More concretely, we can empirically assess the tightness of the bound in two ways. First, we can also construct an upper bound of $\delta_{1t}$ (and therefore $\delta_t$) via propensity score matching. The construction uses the non-negativity of earnings, allowing us to bound the $Y_t(0)$ of later-enrollees by zero from below. We can easily estimate this upper bound and apply standard inference procedures (details in Appendix \ref{subsec:Details-on-Bounds}). The second way is to recognize that $\delta_{1t}$ is actually point identified under Assumption \ref{assu:S2_seq_unconfound}, per results by \citet{Lechner2009JBES} and \citet{LechnerandMiquel2009}. While the associated estimator from Lechner's identification result is more complex to implement and comes with inferential challenges, we can use it to generate a point estimate and compare to our estimated lower bound. In the next section, we discuss the connection of our method to Lechner's and the broader dynamic treatment effect literature.

\subsection{Relation to the Dynamic Treatment Effect Literature\label{subsec:Relations-to-DTE-lit}}

This is certainly not the first paper to consider identification of treatment effects in a dynamic context. As mentioned above, previous studies (e.g., \citealp{Sianesi2004}) have estimated treatment-now-versus-no-treatment-now effects. But many other studies aim to estimate alternative causal parameters. We review relevant research in this latter category below with a particular focus on Lechner's work, providing more context for our method.\footnote{A recent study by \citet{vandenBergVikstrom2019} analyzes dynamic treatment effects on earnings in a setting where training eligibility hinges on a worker remaining unemployed. The closest connection of \citet{vandenBergVikstrom2019} to our paper is their treatment effect parameter: they are also interested in the TOT effect of training versus not training. Since workers in our setting are not subject to the eligibility criterion they consider, our methodology is more closely related to the work by Lechner.}

In a series of influential studies (e.g., \citealp{Robins1986,Robins1997,GillandRobins2001}), James Robins extends the static potential outcomes framework to consider identification of dynamic treatment effects. Specifically, Robins studies identification of potential outcomes under alternative treatment sequences. These treatment sequences are related to but distinct from our treatment variable defined in Section \ref{subsec:Identification}: whereas our treatment variable is whether a worker begins training, each element in Robins's sequence corresponds to whether a worker receives training in a given period. Under a sequential-randomization assumption and a no-anticipation assumption, the distributions of potential outcomes under alternative treatment sequences are identified in the Robins framework. In Appendix \ref{subsec:Robins_details}, we formally discuss Robins's assumptions and results by adapting the excellent summary in \citet{AbbringandHeckman2007} to our two-period setting.

Two important studies by \citet{Lechner2009JBES} and \citet{LechnerandMiquel2009} extend the work by Robins. They focus on the identification of average effects such as the TOTs and propose estimators based on sequential propensity score matching (\citealp{LechnerandMiquel2009}) and sequential inverse probability weighting (\citealp{Lechner2009JBES}). The estimators offer an advantage over those by Robins as they require no functional form assumptions for potential outcomes.

A remarkable implication of the elegant results by \citet{Lechner2009JBES} and \citet{LechnerandMiquel2009} is that $\delta_{1t}$, the TOT for period-1 enrollees from Section \ref{subsec:Identification}, is point identified under dynamic CIA (Assumption \ref{assu:S2_seq_unconfound}). Consequently, the aggregate TOT parameter $\delta_t$ is also identified under dynamic CIA. Using our notation, the estimand that identifies $\delta_{1t}$ for $t\geqslant2$ takes the form:\footnote{Like Robins, \citet{Lechner2009JBES} and \citet{LechnerandMiquel2009} define potential outcomes for various treatment sequences. Using their potential outcome notation (reviewed in Appendix \ref{subsec:Robins_details}), the second term in estimand (\ref{eq:Lechner}) identifies the mean potential outcome $E[Y_t^{00}|D_{1}=1]$ in the counterfactual scenario where period-1 enrollees do not enroll in either period.}
\begin{equation}
E[Y_t|D_{1}=1]-E\left[E[E[Y_t|p_{2}(\mathbf{X}^{2}),\tilde{p}_{1}(\mathbf{X}^{1}),D=0]|\tilde{p}_{1}(\mathbf{X}^{1}),D_{1}=0]|D_{1}=1\right].\label{eq:Lechner}
\end{equation}
In (\ref{eq:Lechner}), the propensity score for $D_{1}=1$ is defined as $\tilde{p}_{1}(\mathbf{X}^{1})\equiv\Pr(D_{1}=1|\mathbf{X}^{1})$, which differs from $p_{1}(\mathbf{X}^{1})$ by not conditioning on $D_{2}=0$, i.e., period-2 enrollees are in the at-risk set when computing $\tilde{p}_{1}(\mathbf{X}^{1})$.

We now provide an intuitive account of Lechner's identification result and how it relates to ours. For ease of understanding, we focus on identification results with estimands that directly match on covariates---that is, replace the propensity scores $p_{1}(\mathbf{X}^{1})$, $\tilde{p}_{1}(\mathbf{X}^{1})$, and $p_{2}(\mathbf{X}^{2})$ in Proposition \ref{prop:Prop_S2_lb} and (\ref{eq:Lechner}) by the covariates $\mathbf{X}^{1}$, $\mathbf{X}^{1}$, and $\mathbf{X}^{2}$, respectively. Constructing the Lechner estimand  entails two steps in a two-period setting. The first step is equivalent to the treatment-now-versus-no-treatment-now matching of \citet{Sianesi2004}, for which the matched control set consists of period-1 non-enrollees ($D_{1}=0$) with similar $\mathbf{X}^{1}$. This matched control set can be broken down into two groups of workers: i. those who enroll later ($D_{2}=1$) and ii. the ``never-enrollees'' ($D_{2}=0$). The second step updates the matched control set: it replaces group i by their never-enrolled counterparts with similar $\mathbf{X}^{2}$ (since $\mathbf{X}^{2}$ contains all the covariates in $\mathbf{X}^{1}$, matching on $\mathbf{X}^{2}$ is equivalent to matching on both $\mathbf{X}^{1}$ and $\mathbf{X}^{2}$).

To make sense of these two steps, first note that Assumption \ref{assu:S2_seq_unconfound}(a) ensures that the average $Y_t(0)$ in the matched control set in step one is the average $Y_t(0)$ for the $D_{1}=1$ population. The problem is that $Y_t(0)$ is not observed for group i, so we need to replace it with observable quantities. This is what step two accomplishes. Under Assumption \ref{assu:S2_seq_unconfound}(b), the average $Y_t$ among group i's replacements is equal to the average $Y_t(0)$ of group i, which implies that the average $Y_t$ of the updated matched control set at the end of step two identifies $E[Y_t(0)|D_{1}=1]$.

To compare our method with Lechner's, it is useful to recast our estimand construction also as a two-step process. The first step is identical to Lechner's, resulting in a matched control set consisting of groups i and ii. In the second step, we also replace group i by their never-enrolled counterparts. But unlike Lechner, these replacements share similar $\mathbf{X}^{1}$ with group i, rather than $\mathbf{X}^{2}$. Under our Assumption \ref{assu:S2_potential_outcome}, the observed mean $Y_t$ among these replacements is an upper bound for the mean $Y_t(0)$ of group i. Because we only use covariates $\mathbf{X}^{1}$, we can merge the two steps and directly look for matches with similar $\mathbf{X}^{1}$ among the never-enrollees.\footnote{We also see a parallel between our approach and that of \citet{Lee2009} as both invoke a monotonicity assumption  for partial identification. \citet{Lee2009} addresses sample selection (missing wage data for those not working) by assuming a nonnegative treatment effect on the selection variable (employment) and trimming the tails of the outcome distribution based on the excess proportion of employed individuals in the treatment group. We address dynamic treatment selection (missing counterfactual outcome) by maintaining Assumption \ref{assu:S2_potential_outcome} and substituting the unobserved potential outcomes of later-enrollees with observed outcomes of never-enrollees to construct a lower bound.}

The point identification of the Lechner estimand is quite appealing, but its implementation is complex, especially with many periods. Specifically, an $S$-period setting requires $S$ matching steps and involves the estimation of $S$ propensity scores. For the two-period case, implementing the estimand (\ref{eq:Lechner}) involves matching twice: first matching on the estimated propensity score $\tilde{p}_{1}(\mathbf{X}^{1})$, and then matching on the estimated propensity score vector $\left(p_{2}(\mathbf{X}^{2}),\tilde{p}_{1}(\mathbf{X}^{1})\right)$. We have eight periods for our Ohio analysis, as we allow for an enrollment window of eight quarters post-layoff. As such, we need to perform matching eight times, and each time, the corresponding propensity score vector increases in length by one. \citet{LechnerandMiquel2009} sensibly use Mahalanobis matching when matching on multiple propensity scores in their four-period analysis, but concerns with the curse of dimensionality become more relevant with more periods, thereby limiting the general applicability of the estimator.

More importantly, the complexity of the \citet{LechnerandMiquel2009} estimator creates inferential challenges. As with many nearest-neighbor propensity score matching papers at the time, sampling variation in estimating the propensity score was ignored when computing analytical standard errors. With the advent of \citet{Abadie_Imbens2016}, it became standard to analytically account for the uncertainty in propensity score estimation for nearest-neighbor matching.\footnote{The inferential challenge addressed by \citet{Abadie_Imbens2008,Abadie_Imbens2016} stems from the fact that nearest-neighbor matching estimators are ``non-smooth functionals of the distribution of the matching variables.'' ``Smooth'' propensity-score-based estimators, such as kernel-based matching (e.g., \citealp{Heckman_etal1997,Heckmanetal1998,Heckmanetal1998_ECMA,HeckmanSmith1999,Biewen_etal2014}) and inverse propensity weighting (e.g., \citealp{Lechner2009JBES,Busso_etal2014}), do not face the same challenge, and using bootstrap to account for the uncertainty in propensity score estimation still provides valid inference.} Thus, should we choose to rely on \citet{LechnerandMiquel2009}, extending \citet{Abadie_Imbens2016} to the complex case of Mahalanobis matching on a vector of estimated propensity scores seems warranted, but doing so would distract from the substantive analysis of this paper. Further investigation on inference is also needed for the inverse propensity weighting (IPW) estimator of \citet{Lechner2009JBES}: as we discuss in Appendix \ref{subsec:Lechner_IPW}, the original study tries five variance estimators but does not settle on a preferred choice.

For these reasons, we rely on the partial identification result in Proposition \ref{prop:Prop_S2_lb} (and the corresponding estimator) to generate our main estimates. Its advantage is simplicity, as the estimation and inference procedures can be implemented using off-the-shelf Stata commands. Its disadvantage is the loss of point identification, but we do not view it as a threat to the empirical substance of our analysis. As mentioned above, we can also construct an upper bound, and together, the two bounds indicate an informative range for the TOT effect. Moreover, we can produce the TOT point estimate by following \citet{LechnerandMiquel2009}, which is only slightly above our estimated lower bound. We report these estimates in Section \ref{subsec:Overall-Returns}.

Finally, while this section focuses on research that assumes selection on observables, another strand of the dynamic treatment effect literature explicitly models unobserved heterogeneity. Examples include \citet{AbbringandvandenBerg2003,AbbringandvandenBerg2004}, \citet{HeckmanandNavarro2007}, \citet{Baetal2017}, \citet{Han2021}, and \citet{Fitzenbergeretal2023}. We provide an overview of these studies in Appendix \ref{subsec:unobserved_heterogeneity}.

\subsection{Matching Specification\label{subsec:Matching-Specification}}

This section describes our matching specification. As discussed above, the specification is informed by our own validation exercise, described in \citet{LeungandPei2020}. In particular, we find that in another training context, where detailed earnings histories exist for a sample of recently employed workers, modeling the propensity score as a simple logit where the terms enter linearly performs well against more flexible models. However, the validation exercise sample is quite small relative to the Ohio sample. Since we have a large sample, we use both exact and propensity score matching to ensure that our matched comparison group is as similar as possible to the enrollees.

As described in Section \ref{sec:Data}, we define enrollees as those who start school within eight quarters after filing a UI claim. Proposition \ref{prop:Prop_S2_lb} (more precisely, its many-period generalization, Proposition \ref{prop:Sgeneral_ps_lb} in Appendix \ref{subsec:S-period-Generalization}) suggests that we can estimate the TOT lower bound separately for those who enroll in the first through eighth quarter after layoff and aggregate them to obtain a lower bound for the overall TOT. The Proposition also states that, for all eight enrollee cohorts, we should construct the corresponding matched comparison group by drawing from the pool of workers who do not enroll within the two-year post-layoff period.

We perform exact matching along three dimensions. First, we require that enrollees be exactly matched to non-enrollees laid off in the same quarter. This is motivated by the fact that economic conditions and policies varied widely over our study period. Workers laid off in 2004 may differ from those starting unemployment at the peak of the Great Recession and face dramatically different labor market landscapes.\footnote{Although \citet{Heckmanetal1998_ECMA} stress the importance of matching workers from the same local labor market, evidence from \citet{Michalopoulousetal2004} and \citet{Mueser_etal2007} suggests that this is less important when comparison groups are drawn from a single state. Since our data come from (moderately sized) Ohio, temporal rather than geographic variation captures most of the variation in labor market conditions within our sample. That said, we also include the unemployment rate in the month and county of layoff in the propensity score model, which accounts for geographic differences.} Furthermore, policies enacted to help workers overcome challenges during the Great Recession, such as UI extensions and information campaigns about resources for retraining, may have influenced workers' decisions to enroll (\citealp{BarrTurner2015,BarrTurner2016}). By comparing workers that were laid off around the same time, we attempt to control for the influence of time-varying labor market conditions and policies.

Second, following the training evaluation literature, we require that workers be exactly matched on gender, as decisions to enroll may differ between men and women. As argued by \citet{HeckmanSmith1999}, men's decisions to enroll may be more heavily influenced by economic prospects, while women's decisions may depend more on family responsibilities. The different motivations for enrolling may translate into different training effects across gender.

Finally, we exactly match enrollees and non-enrollees based on whether they were working in manufacturing at layoff. The manufacturing sector is of particular interest to policymakers, as it has been in rapid decline over the past decades (\citealp{Autoretal2013}), particularly in Rust Belt states like Ohio.

Within the (exactly matched) layoff quarter, gender, and manufacturing cells for each of the eight enrollment timing cohorts, we estimate a separate propensity score model. As we demonstrate in the validation study in \citet{LeungandPei2020}, it is crucial to include pre-enrollment earnings, and we use three years of pre-layoff quarterly earnings as (linear) inputs into the model.\footnote{While the literature also favors including indicators for zero earnings in each quarter, we do not include these to minimize the number of perfect prediction and overlap problems within the exact-matching-cells. However, as discussed in Section \ref{subsec:Overall-Returns} below, including these indicators does not meaningfully change our estimates.} In addition to pre-layoff earnings, we also include quarterly earnings between layoff and enrollment per our identification result (we illustrate the importance of incorporating these earnings with empirical evidence in Section \ref{subsec:Overall-Returns}). Finally, we include in our propensity score model demographic and prior job characteristics to the extent that our data allow it.\footnote{Out of the 980 propensity score models, 27 do not converge to a solution with the full set of covariates. In these cases, we estimate the model by eliminating one covariate at a time in the following order: quarterly earnings (starting with most recent), unemployment rate, tenure dummies (starting with longest), dependent dummy, race dummies (unknown, American Indian/Asian/Hawaiian/Hispanic, and Black), age dummies (starting with oldest), and sector dummies (transportation, retail, accommodation and food, healthcare, administrative support, wholesale, and construction).} Specifically, we include race indicators (White, Black, other, or unknown), pre-layoff sector indicators (construction, wholesale trade, administrative support and waste management, healthcare and social assistance, accommodation and food services, retail trade, transportation, and other), job tenure categories (less than one year, one to six years, and more than six years), age indicators (age below 19, each year from age 19 to 59, and older than 59), whether a worker has a dependent at the time of UI claim, and county unemployment rate during the month of layoff.

We use the estimated propensity score to pair each enrollee with her nearest neighbor from the comparison sample, where each comparison worker may be matched to more than one enrollee (i.e., matching with replacement). Choosing a larger number of neighbors may further reduce variance in the estimated treatment effect, but the quality of the match may deteriorate as we allow for larger propensity score differences between the enrollee and comparison workers. Since we have a large enough sample size to attain precise estimates and are more concerned with bias, we use only one neighbor in our main specification. We do, however, explore the sensitivity of our results with respect to the number of neighbors in Section \ref{subsec:Overall-Returns}.

After each logit propensity score regression, we estimate the (cell-specific) TOT with the average difference in outcomes between each enrollee and its match, and the standard error is computed following \citet{Abadie_Imbens2016}. We then aggregate the TOTs across the enrollment-timing $\times$ layoff-quarter $\times$ gender $\times$ sector cells to obtain the overall treatment effect, where the weights are proportional to the number of enrollees in each cell.

Given that we match workers along the dimensions discussed above, the question of residual differences between enrollees and matched non-enrollees remains. We argue that these remaining factors that compel workers to enroll are unlikely to be related to future earnings potential, conditional on the covariates used for matching. The empirical literature documents several potential sources of variation. First, research has shown that distance to a community college affects whether a worker ultimately enrolls (\citealp{Card1993}). While distance is not a strong instrument in our setting (Appendix \ref{sec:Alternative-Identification-Strat}), it is indeed negatively associated with enrollment. Second, information and nudges appear to influence students' decision-making. For example, in our context, \citet{BarrTurner2016} show that a letter targeted to UI claimants informing them of financial aid resources affects enrollment. Relatedly, complexity in UI rules (such as what counts as ``approved training'' that allows workers to continue receiving benefits) may introduce additional variation in whether workers decide to pursue training opportunities. As mentioned in Section \ref{sec:Institutional-Background}, although there are 7,000 courses that Ohio approves automatically, academic courses were approved on a case-by-case basis in the absence of an official policy (\citealp{NASWA2010}). Third, in the context of training programs, there may be variation in who can access training resources based on requirements set by local job centers (\citealp{Fortsonetal2017}). Finally, there has been evidence of capacity constraints in community colleges and programs within community colleges (e.g., \citealp{Grosz2020}), which may generate exogenous variation in who can enroll.

\subsection{Matching Design Quality Check: Overlapping Support and Covariate Balance\label{subsec:Estimation-Strategy-and}}

As emphasized by \citet{Smith_Todd2005}, we need to assess the validity of the overlapping support assumption. We first point out that some observations indeed appear to violate this assumption, but they are a tiny fraction of the sample. Specifically, when we use the trimming threshold from the algorithm by \citet[367--368]{Imbensandrubin2015}, we find that only 842 observations out of the nearly 72,000 enrollees are dropped---841 observations are dropped due to perfect prediction within exact matching cells, and one observation has a propensity score that is too close to 1.\footnote{The remaining difference in the enrollee sample size between Panels A and B of Table \ref{tab:descriptives} is due to the elimination of the winter quarter of 2013 in the HEI data mentioned in Section \ref{sec:Data}. We do not estimate the propensity score and enrollment effect for (OTC only) enrollees who start in that quarter, which eliminates an additional 24 enrollees.}

We show the overlap of the enrollee and non-enrollee distributions in Appendix Figure \ref{fig:overlap}. Because many of the propensity scores are close to zero, for ease of visual inspection, we overlay the histograms of the log odds ratio ($lor$) of the two groups---$lor$ is a monotone transformation of the propensity score (i.e. $lor\equiv\log\frac{p(x)}{1-p(x)}$). In the top row, we plot the estimated log odds ratio distributions for workers who enroll one, four and eight quarters post-layoff, and overlay the corresponding distributions for the non-enrollees. We see that the $lor$ distributions have little support on the positive range, indicating that the propensity score is below 0.5 for the vast majority of observations, and not surprisingly enrollees tend to have a higher propensity to pursue further education. We show the \emph{frequency} plots of $lor$ for the two groups in the bottom row, which are more relevant for assessing overlapping support. Because the number of observations in the non-enrollee group is far larger than that in the enrollee group, there appears to be sufficient overlap even in the higher range of the $lor$.

To compare covariate values across samples, we follow the literature (e.g., \citealp{Imbens2015}) and report in Panel A of Table \ref{tab:descriptives} the normalized differences between enrollees and non-enrollees for each covariate. That is, for each covariate $X$, we report $(\bar{X}_{E}-\bar{X}_{N})/\sqrt{(S_{X,E}^{2}+S_{X,N}^{2})/2}$, where $\bar{X}_{E}$ and $\bar{X}_{N}$ are the respective sample means of $X$ among enrollees and non-enrollees, and $S_{X,E}$ and $S_{X,N}$ are the corresponding sample standard deviations. The denominator can be interpreted as an average standard deviation (ASD), allowing the normalized difference to be interpreted in percentage terms of the ASD. The covariate that exhibits by far the largest difference is age: enrollees are younger by 58 percent of ASD. The contrasts are less stark for the other 21 covariates: nine have a normalized difference below 5 percent of ASD, two between 5 and 10 percent, eight between 10 and 20 percent, and two just above the 20 percent threshold of what \citet{RosenbaumRubin1985} consider a large difference (\citealp{Imbens2015} suggests a higher rule of thumb threshold of 30 percent). In Panel B of Table \ref{tab:descriptives_matched}, we show the average characteristics of the \emph{matched} enrollee and non-enrollee samples. Indeed, covariates are balanced across the two groups: normalized differences are very small, and all $t$-tests fail to reject equality despite the large sample size. We proceed to present our main empirical results on the effects of  retraining in the next section, where we also provide additional graphical evidence in support of balance in pre-enrollment earnings trajectories.

\section{Empirical Results: Effects of Further Education During Unemployment\label{sec:Results}}

\subsection{Overall Effects\label{subsec:Overall-Returns}}

We begin by graphically presenting the average earnings of the full sample of enrollees and their matched non-enrollees. In Figure \ref{fig:matched_overall}, the solid line shows the average earnings of enrollees over time, from 20 quarters before until 16 quarters after enrollment begins. The dashed line shows the earnings averaged across each enrollee's nearest neighbor in the non-enrollee sample (see Appendix Figure \ref{fig:unmatched} for a comparison of earnings trajectories for enrollees and unmatched non-enrollees). Prior to enrollment, enrollees and their closest comparison workers have similar earnings trajectories, increasing in the period from five years prior to enrollment to approximately two years prior, before dropping to 50 percent of the peak by the time of enrollment. The seemingly slow decline in earnings is due to the fact that we are averaging the earnings of workers who enroll at different times post-layoff, and not because of a drawn-out earnings reduction process for all workers. Because a correctly specified propensity score is a balancing score (\citealp{RosenbaumRubin1983}), the close alignment of the pre-enrollment trajectories lends confidence to the specification of our propensity score model (balance is not mechanically guaranteed just because the earnings enter the logit propensity score model---see Figure 2 of \citealp{LeungandPei2020}). After enrollment, the two lines begin to diverge. Enrollees have lower earnings for approximately two years before surpassing their comparison group. This ``lock-in'' effect, which may come about because enrollees are more constrained in their ability to search for jobs and work while in school, is consistent with the finding in \citet{Heinrich_etal2013} for the WIA Dislocated Worker program.\footnote{Although the mean (median) enrollment duration is 4.5 (4) quarters in the sample, the average duration between first and last enrollment quarters is closer to 6 quarters due to the fact that workers are not always ``continuously'' enrolled. Furthermore, it may take some time post-training for earnings to recover: \citet{Fortsonetal2017} note that there were about two quarters between when WIA trainees completed training and began post-training employment.} The gains appear to grow after the lock-in period while the earnings of non-enrollees flatten. In the third and fourth years post-enrollment, enrollees earn \$348 more per quarter than non-enrollees as reported in Table \ref{tab:att_estimates}, a gain of about six percent. It is notable that even at more than four years after layoff, both enrollees and non-enrollees do not catch up (on average) to their pre-layoff earnings.

As discussed in Section \ref{subsec:Identification}, these estimates are lower bounds of the enrollment effect, under the assumption that those who enroll in later periods have lower counterfactual earnings than those who never enroll (Assumption \ref{assu:S2_potential_outcome} and its many-period generalization Assumption \ref{assu:Sgeneral_potential_outcome} in Appendix \ref{subsec:S-period-Generalization}). We first present evidence in support of this assumption using the test proposed in Section \ref{subsec:Identification} and Appendix \ref{subsec:Discussion-of-Assumption2}. Specifically, for an enrollment quarter $s=1,\dots,7$ and a later-enrollee cohort enrolling $l$ quarters later ($1\leqslant l \leqslant 8-s$), the test compares their earnings during each of the $l$ quarters (\emph{before} the later-enrollees enroll) against non-enrollees with similar characteristics at the beginning of quarter $s$. We should find lower earnings among later-enrollees under Assumptions \ref{assu:S2_potential_outcome} and \ref{assu:Sgeneral_potential_outcome} if these interim earnings positively correlate with future potential earnings $Y_t(0)$. We implement the test using propensity score matching and present the earnings differences between later-enrollees and matched non-enrollees for all seven values of $l$ (we aggregate across $s$ for concise presentation) in Appendix Table \ref{tab:lower_bound_test}. Indeed, all earnings differences are negative, lending credibility to the assumption underlying the lower bound result.

Second, we show that the lower bound is informative in our setting. As discussed in Section \ref{subsec:Relations-to-DTE-lit}, results from \citet{Lechner2009JBES} and \citet{LechnerandMiquel2009} imply the point identification of the TOT. We implement the sequential propensity score matching estimator of \citet{LechnerandMiquel2009} and find that enrollees earn \$351 more than non-enrollees in the third and fourth years post-enrollment (also as discussed in Section \ref{subsec:Relations-to-DTE-lit}, developing an inference procedure for the \citealp{LechnerandMiquel2009} estimator by extending the \citealp{Abadie_Imbens2016} standard errors is outside the scope of this paper). Our lower bound estimate of \$348 is remarkably close and its 95 percent confidence interval comfortably contains \$351. We also construct an upper bound using Proposition \ref{prop:Sgeneral_ps_ub} in Appendix \ref{subsec:S-period-Generalization}. Its estimate is \$497, or nine and a half percent; together, the two bounds pinpoint an informative range for the earnings effects. The tightness of the bounds results from the small number of later-enrollees relative to non-enrollees, so replacing the later-enrollees' potential earnings with other values will not have a major impact. The fact that the point-identified estimate is much closer to our lower bound than upper bound suggests that the counterfactual earnings of the later-enrollees are better approximated by the earnings of their non-enrollee replacements than by zero.

Because our lower bound is very close to the point estimate of the TOT effect, we refer to the lower bound simply as ``the effect'' in the remainder of this paper for brevity. Next, we conduct several sensitivity checks on our average effect estimates. While the set of non-enrollees in Figure \ref{fig:matched_overall} is selected by using 12 quarters of pre-layoff earnings in the propensity score formulation, Panels A and B of Appendix Figure \ref{fig:partialmatch} show that earnings patterns are similar under alternative matching specifications that include fewer quarters of earnings (one pre-layoff quarter and four pre-layoff quarters, respectively, plus earnings between layoff and enrollment). However, when we match without earnings between layoff and enrollment in Panel C, the patterns are quite different: the enrollees in this panel have much lower earnings than their matched comparison group heading into enrollment.\footnote{The estimates of enrollment effects in the third and fourth year after enrollment are \$378 (seven percent), \$366 (seven percent), and -\$219 (-4 percent) for Appendix Figure \ref{fig:partialmatch} Panels A, B, and C, respectively.} Appendix Figure \ref{fig:partialmatch} highlights the importance of controlling for the most recent information just before enrollment. In contrast, including dummies for zero earnings (on top of the full set of quarterly earnings) in the propensity score model does not appear to affect estimates substantially---the effect in the third and fourth year post-enrollment remains at six percent (\$339 per quarter, Panel A of Appendix Figure \ref{fig:zeroearn_noborder}). We also probe the sensitivity of our results to dropping counties that border another state, as workers from these counties are more likely to be employed in another state post-layoff (which we would not be able to observe) and find that it does not change our estimates much---in the third and fourth year post-enrollment, enrollees earn \$324 more per quarter (six percent, Panel B of Appendix Figure \ref{fig:zeroearn_noborder}). Additionally, we  explore the sensitivity of our findings to using more neighbors. We show in Panel A of Appendix Figure \ref{fig:neighbors_analysis} the estimated treatment effect against the number of neighbors used.\footnote{For this exercise, we only use the cell with the largest number of claims (male non-manufacturing workers laid off in the first quarter of 2009). We focus on this subsample to ease the significant computational burden of running the matching analysis 25 times with the full sample.} The estimated treatment effect tends to shift upward as we increase the number of neighbors, and the variance is reduced by up to 66 percent as we show in Panel B of the same figure. An estimate using the full sample with five neighbors follows this general pattern, but since our sample size allows for precise estimates, we use the lowest bias one-neighbor specification. Lastly, we apply IPW and local linear matching with block bootstrap at the county level to assess the sensitivity of our main estimates with respect to alternative classes of estimators and cluster-robust inference. With IPW, the impacts on earnings are -\$393 (standard error 16) during the first two years and \$341 (standard error 19) during years three and four post-enrollment.\footnote{The standard errors for IPW without clustering are 15 and 18, respectively. Their similarity implies that treating observations as i.i.d. across individual UI claims as we did when producing our main estimates is a reasonable approximation.} With local linear matching, the corresponding impact estimates are -\$364 (standard error 19) and \$375 (standard error 21), respectively. Compared to the effects reported in Table \ref{tab:att_estimates}--- -\$381 (standard error 21) and \$348 (standard error 23)---the largest deviation of the alternative point estimates is under eight percent, giving us confidence in our main estimates. While the alternative standard errors are somewhat smaller, all estimates are highly statistically significant. 

Finally, the two panels of Appendix Figure \ref{fig:ext_margin} show that the probability of having positive earnings and weeks worked per quarter, respectively, mirror the earnings patterns. A natural question that arises is whether or not the gain in employment can explain all of the gains in earnings, or whether enrollment increases both employment and wage rates. One way to answer this question is to note that in Panel B of Appendix Figure \ref{fig:ext_margin}, non-enrollees work for 7.3 weeks and earn \$5,364 per quarter on average three to four years after enrollment, implying a weekly wage of approximately \$735. Enrollees, on the other hand, are employed 7.9 weeks and have a weekly wage of approximately \$723. Although the differences in weekly wage rates are not causal effects, this indicates that increased wage rates are not driving the enrollment effects.\footnote{Since Appendix Figure \ref{fig:ext_margin} shows a small gap in weeks worked between enrollees and non-enrollees in the pre-period, the weeks ``effect'' may be overstated. To see how this affects our conclusion, consider the following decomposition of the earnings effect: $\Delta\text{Earnings =\ensuremath{\Delta\text{Weeks}}\ensuremath{\times}\ensuremath{\frac{\text{Earnings}_{N}}{\text{Weeks}_{N}}}+\ensuremath{\Delta\frac{\text{Earnings}}{\text{Weeks}}\times\ensuremath{\text{Weeks}_{E}}}}$ where Earnings and Weeks are average quarterly earnings and average weeks worked in a quarter, respectively, for enrollees ($E$) and non-enrollees ($N$), and $\Delta$ denotes the difference between enrollees and non-enrollees. If we adjust the weeks ``effect'' downward by 0.22 weeks (the gap in the pre-period), the first term of this decomposition still accounts for 80 percent of the total earnings effect.} Another way to see this is to examine the earnings distributions of enrollees and matched non-enrollees, shown in Appendix Figure \ref{fig:distributional_effects}. While the top panels of the figure document similar earnings distributions before enrollment, the bottom panels show that the distributions diverge eight quarters after enrollment, and the gap widens further by quarter 16, where the gains are concentrated at the extensive margin. This finding parallels those of two recent studies. Leveraging quasi-random assignments of Danish UI recipients to caseworkers, \citet{HumlumMunchRasmussen2025} find that classroom training raises workers' earnings mainly by helping them find jobs, rather than by increasing their hourly wages or hours conditional on employment. In a randomized experiment, \citet{Belotetal2025} document that the positive effect of providing automated occupational advice to the long-term unemployed operates mainly through the extensive margin: the intervention increased unemployed workers' probability of reaching a modest earnings threshold and their probability of finding a (stable) job but had statistically insignificant effects on cumulative earnings overall. We conclude that training mainly affects employment and likely has minimal effect on wage rates four years after enrolling.

\subsection{Effects By Subgroup\label{subsec:subgroup_returns}}

We now present the enrollment effects by subgroup. For subgroups that have exactly matched participants (i.e., enrollment timing, layoff quarter, gender, and sector), we estimate the enrollment effect by restricting to enrollees in the subgroup and examine the difference in outcomes compared to their matched non-enrollees. For subgroups that are not exactly matched (age, tenure, and race groups), we first restrict the sample to the subgroup and then re-match enrollees to non-enrollees using the estimated propensity score described in Section \ref{subsec:Matching-Specification}.\footnote{We cannot simply compare enrollees within a certain subgroup to their original matched comparison groups because this compares the outcomes of enrollees within a subgroup to non-enrollees that are potentially not in the subgroup. A simple way to see this is to consider a randomized experiment: if a population is randomly assigned to a treatment or control group with a coin flip, the propensity to be treated is 50 percent for the entire population. If we wanted to estimate the treatment effect for women only, we cannot simply compare the outcomes of treated women with the entire control population, even though their propensity scores are the same at 0.5.} That is, we do not implement another matching procedure with re-estimated propensity scores conditional on the subgroup, which is computationally expensive and could bring in workers not included in the set of matched non-enrollees in Section \ref{subsec:Overall-Returns}, resulting in inconsistencies across samples.

A potential challenge is that the nearest neighbor for an enrollee in the subgroup analysis may differ from that in the full sample analysis and may be a worse match. But as we show in Appendix Table \ref{tab:subgroup_balance}, our method still achieves balance in most of the 28 subgroups, as measured by estimates of TOT on earnings one to two years and three to four years pre-enrollment, respectively. After adjusting for multiple hypotheses testing with the Holm method, pre-enrollment earnings from three subgroups (workers under 40, workers with job tenures greater than six years, and Black workers) remain statistically unbalanced at the 5 percent level.\footnote{The Holm method is also known as sequential Bonferroni. Hypotheses are tested sequentially in order of significance, reducing the number of hypotheses tested in each step by one (\citealp{Holm1979}). It provides more power to reject false hypotheses than standard Bonferroni without imposing additional assumptions.}\multfootsep\footnote{Note that, because the nearest neighbor may change from the full sample to a particular subgroup, balance in the full sample and one subgroup does not imply balance in the complement of that subgroup.} But even for these three groups, the imbalance goes away in the most recent year immediately preceding enrollment as seen in Appendix Figures \ref{fig:subgroups_age}, \ref{fig:subgroups_tenure}, and \ref{fig:subgroups_race}. This is reassuring as the most recent pre-enrollment earnings are possibly most predictive of future potential outcomes. To alleviate any remaining concerns with this imbalance, which manifests in parallel trajectories between enrollees and matched non-enrollees before earnings start to decline, we also report estimates using difference-in-differences matching in Table \ref{tab:subgroup_ests}. These estimates are obtained by first differencing future earnings with (symmetrically timed) pre-enrollment earnings within person, and then comparing the resulting differences between enrollees and matched non-enrollees.\footnote{The differencing step does not materially affect estimates. This is consistent with \citet{ChabeFerret2017}, who finds that when conditioning on many periods of pre-treatment earnings as we do here, the bias of matching with or without differencing is similar.} The reported standard errors in Table \ref{tab:subgroup_ests} are from testing whether the average pairwise difference between enrollees and matched non-enrollees is different from zero, and we abstract away from the sampling variation in forming these matched pairs. This abstraction is inconsequential for the full sample where we also have available the standard errors that account for the sampling variation in forming matched pairs per \citet{Abadie_Imbens2016}: compared with their counterparts in Table \ref{tab:att_estimates}, the standard errors in columns (1) and (6) of the first row in Table \ref{tab:subgroup_ests} are only slightly larger. 

For each of these subgroups, the effects may vary due to differences in ``types'' and intensity of schooling (e.g., types of courses taken, duration of enrollment, and whether a credential was obtained), differences in enrollee composition, or other factors such as labor market conditions. While we cannot tease out the exact reasons for why we find different effects for some subgroups, we discuss in turn likely explanations.

\paragraph{By Enrollment Timing}

As discussed in Section \ref{sec:Data}, workers typically do not enroll immediately after layoff, and some take longer to go back to school than others. Later enrollees may differ from earlier enrollees in many ways, including in terms of employment opportunities. Therefore, although we show in Appendix Figure \ref{fig:by_enroll_timing} that impacts of enrollment differ by the quarter in which workers go back to school within two years after layoff, these differences should not be interpreted as causal nor informative about optimal enrollment timing. Each panel, which corresponds to enrolling in a certain quarter (one through eight) since job loss, shows a sharp drop in earnings in the pre-enrollment period. However, it is clear that the groups differ in the extent to which their earnings recover post-enrollment, with larger earnings gains for those who enroll later relative to layoff. As reported in Table \ref{tab:subgroup_ests}, we find that workers who enroll later have a smaller ``lock-in'' effect in the first two years after enrolling, and larger gains from enrollment (up to 11 percent) in the third and fourth years.

There are several potential explanations for these patterns. First, it is possible that the larger lock-in effects are due to a higher ``treatment dosage'' for those who enroll earlier. Those who enroll one quarter after layoff have a mean enrollment duration of 4.7 quarters, while those who enroll eight quarters after layoff are only enrolled on average for 4.2 quarters (and the pattern is monotonic for those who enroll in the second through seventh quarters). It is also true that these longer durations of enrollment correspond to higher rates of credential receipt: for those who enroll within the first quarter post-layoff, about 27 percent receive a credential, while for those who enroll in the eighth quarter after layoff, the credentialing rate is about 20 percent (the pattern is roughly monotonic in between, though it peaks for those who enroll in the second quarter). However, we also see that earlier enrollees have lower post-enrollment gains, which seems inconsistent with the fact that they have more ``intensive'' schooling (though it is possible that we simply have not looked at a long-enough post-period to observe the full gains).

Another explanation, however, is consistent with the discussion in \citet{Heckmanetal1999_HoLEv3} that training is a form of job search, and enrollees tend to be workers who are not yet reemployed. Correspondingly, those who enroll later have relatively long spells of unemployment. In other words, as discussed above, earnings patterns seem to reflect selection in who enrolls earlier or later. Appendix Figure \ref{fig:by_enroll_timing} shows that later enrollees tend to be those who have earnings that have been depressed for a while before enrolling, and their counterfactual earnings without school tend to be flatter than earlier enrollees. The lack of lock-in for later enrollees likely reflects their low opportunity cost of schooling. Therefore, while we find larger enrollment effects for those who enroll later, we note that this result does not imply that it would be more advantageous for workers to delay enrollment. 

\paragraph{By Layoff Year}

Since our study period covers the Great Recession, there is considerable variation in the labor market conditions at the times of enrollment and reemployment. Prior studies have shown that training programs are more effective when the unemployment rate is high at program entry and low when training ends (\citealp{LechnerWunsch2009}; \citealp{Kluve2010}; \citealp{CardKluveWeber2015}). In particular, poor labor market conditions at program entry are associated with smaller lock-in effects due to lower opportunity costs of training, as well as larger gains at the end of training. Outside of the training literature, there is also evidence that graduating college during an economic downturn is associated with persistent earnings losses (\citealp{Kahn2010}; \citealp{Oreopoulosetal2012}).

In Appendix Figure \ref{fig:by_year}, we present the earnings of enrollees and matched non-enrollees separately by year of layoff. Consistent with the literature, we find that the effects of training are larger among workers laid off when unemployment rates peaked during the Great Recession: Table \ref{tab:subgroup_ests} shows that workers laid off in 2010 and early 2011 have earnings gains of 8 to 14 percent in the third and fourth years post-enrollment, while the cohorts laid off prior to the recession in 2006 to 2007 (who would be entering the labor market as unemployment was rising rapidly) have gains that are indistinguishable from zero (the effects for those laid off in 2008 and 2009 fall somewhere in between these extremes).

To explore the extent to which compositional differences explain the effects in different years, we conduct an exercise similar to those by \citet{LechnerWunsch2009} and \citet{HeinrichMueser2014}, where we reweight each year's estimates so that the enrollee composition matches those laid off in 2004 based on observable characteristics.\footnote{Specifically, for enrollees laid off in year $m$, we estimate the counterfactual average earnings impact (conditioning on $D=1$ is omitted for brevity):
\begin{eqnarray*}
 &  & \int E[Y_t(1)-Y_t(0)|\mathbf{X},\text{year}=m]\mathrm{d} F(\mathbf{X}|\text{year}=2004)\\
 & = & \int E[Y_t(1)-Y_t(0)|\mathbf{X},\text{year}=m]\underbrace{\frac{\Pr(\text{year}=2004|\mathbf{X})}{\Pr(\text{year}=m|\mathbf{X})}\frac{\Pr(\text{year}=m)}{\Pr(\text{year}=2004)}}_\text{weight}\mathrm{d} F(\mathbf{X}|\text{year}=m).
\end{eqnarray*}
To estimate the weights, we restrict the sample to only enrollees laid off in 2004 and year $m$. We predict whether an observation is from the 2004 cohort via logit using observable characteristics (gender, industry, age, race, presence of dependents, pre-layoff wages, and quarter of layoff), and the predictions give us the first term of the weights. For the second term of the weights, we simply need the proportion of observations
in each year. We then calculate the reweighted treatment effect for enrollees laid off in year $m$ via the usual nearest neighbor matching, but with each treatment-control pair multiplied by its estimated weight.} We present the results in Appendix Table \ref{tab:recession_reweight} and find that the changing composition of the layoff cohorts over time plays a minimal role in the larger treatment effects of 2010 and 2011, though reweighting does appear to moderately dampen the effect for 2009.\footnote{The small discrepancies in the number of observations between Appendix Table \ref{tab:recession_reweight} and Table \ref{tab:subgroup_ests} are due to the fact that we require ``overlap'' between the 2004 enrollees and other years' enrollees. Practically, this means that when enrollees have characteristics that perfectly predict their layoff year (i.e., they are so different from the 2004 enrollees that they cannot be used for reweighting), we drop them from the analysis.} The finding that composition does not explain much of the effect variation over time echoes \citet{LechnerWunsch2009}. It is also broadly similar to \citet{HeinrichMueser2014}, who find that training programs targeting displaced workers were more beneficial in 2008-2009 relative to 2007 (though composition explains more of the temporal variation in that context than ours).

\paragraph{By Gender}

Appendix Figure \ref{fig:by_gender} shows the enrollment effects by gender. In contrast to \citet{Jacobson_etal2005_JE}, we find that the gains to enrollment are larger for men in the four years post-enrollment, though we note that women spend more time in school, at an average of 5.1 quarters over four years versus 4.1 quarters for men. This means that women have a longer ``lock-in'' period relative to men, and a longer follow-up period may reveal larger effects for women. In the third and fourth years after enrollment, we find that the earnings gains of enrollment are ten percent for men and one percent for women. In Section \ref{subsec:What-Does-Schooling}, we discuss the different types of courses and credentials obtained by women and men, and the subsequent differences in industries of employment.

\paragraph{By Manufacturing versus Non-manufacturing}

Due to the decline of the manufacturing sector and the large earnings losses suffered by unemployed workers therein, the effects of retraining among this group have attracted more policy attention (\citealp{CouchPlaczek2010}). Appendix Figure \ref{fig:by_industry} shows separately how further schooling impacts manufacturing and non-manufacturing workers. Consistent with previous studies on the earnings impact of displacement, manufacturing workers, who make up 29 percent of our sample, have higher average pre-layoff earnings than non-manufacturing workers and experience a larger drop in earnings post-layoff. In the third and fourth years post-enrollment, manufacturing workers see a two percent earnings gain on average, smaller than the nine percent for non-manufacturing workers (Table \ref{tab:subgroup_ests}). However, the long lock-in period and the fact that manufacturing enrollees are 14 percentage points more likely than matched non-enrollees to switch industries suggest that the effect may increase further in the future. Since switchers forgo industry-specific skills, it may take longer for the human capital investment to pay off (we discuss industry switching in more detail in Section \ref{subsec:What-Does-Schooling}). Furthermore, we see in Panel A of Appendix Figure \ref{fig:by_industry} that the earnings trajectory of the non-enrollees is quite flat for those who previously worked in manufacturing, allowing for the possibility of higher longer-run enrollment effects.

\paragraph{Other Subgroups}

The final rows of Table \ref{tab:subgroup_ests}, and the accompanying earnings trajectory plots in Appendix Figures \ref{fig:subgroups_age}-\ref{fig:subgroups_race}, show the enrollment effects for several other subgroups of interest. We find statistically significant positive effects across workers of different ages and some racial groups. There do not seem to be large differences in enrollment effects for older (age 40 and above) workers relative to younger workers and the effects are of a similar magnitude as in \citet{Jacobson_etal2005_ILR}. Although \citet{Jacobson_etal2005_JE,Jacobson_etal2005_ILR} focus on long-tenured workers who have worked three or more years before layoff, we find that the shortest-tenured group with less than one year of tenure enjoys the largest gains at ten percent. Finally, while there is no consensus on how the effects to community college vary by race (\citealp{BelfieldBailey2011}) and estimates are not available in \citet{Jacobson_etal2005_JE,Jacobson_etal2005_ILR}, we find larger effects for Whites than for Blacks.\footnote{Effects may vary by race for two reasons. In a human capital framework with skill complementarity, returns to education depend on the quantity and quality of prior schooling (e.g., \citealp{CardKrueger1992}; \citealp{CunhaHeckman2007}). Therefore, to the extent that prior levels of schooling or schooling quality vary by race, we might expect different training effects for Blacks and Whites (\citealp{CardKrueger1992}; \citealp{NealJohnson1996}). Discrimination in the labor market may also play a role. For example, \citet{BertrandMullainathan2004} find that Black job applicants have lower returns (as measured by callbacks) to resume quality than White job applicants.}

Lastly, we consider those with previous college experience. We focus on those who are age 25 or younger to capture prime college attendance years (i.e., ages 18-25), and who were laid off in 2007 or later so that we can observe recent college experience (i.e., within the prior seven years).\footnote{While one may be concerned with the validity of conditional independence for younger workers with a limited earnings history, we note that our sample only contains workers who have had substantial enough earnings to be eligible for UI.} These restrictions eliminate 94 percent of the sample, though within this subsample, the proportion of those who enroll after layoff is higher than in the full sample (nine percent versus four percent).\footnote{Although it would also be interesting to examine those with previous college degrees, we do not observe many prior degrees in our data. Of the already small subgroup of those with prior college, only 13 percent are observed to have received a degree.} Among these enrollees, 70 percent had no prior college. We re-estimate the propensity scores for each subsample (enrollees who have or do not have prior college) with adjustment for the much smaller samples. Instead of exactly matching on claim quarter, gender, and sector, we simply require enrollees be matched with non-enrollees who were laid off within the same year and are of the same gender, and add claim quarter and sector dummies to the propensity score model. Appendix Figure \ref{fig:subgroups_prioredu} shows that those without prior college experience see larger effects of about 12 percent while those with recent prior college experience see a marginally statistically significant gain (at the ten percent level) of about three percent in the third and fourth years after enrolling. One interpretation of this finding is decreasing marginal benefits: among young unemployed workers, additional training is less beneficial if they already had some college in the past. However, we caution that this is a narrow subgroup of workers and the results may not generalize to older or higher tenure workers who may have been out of school for a longer period of time.

\subsection{Effects By School Type and Quality\label{subsec:school type_returns}}

We now turn to how earnings gains vary by the type of school attended. Our main enrollee sample consists of those whose first observed institution is a community college or technical center. We now separately analyze the enrollment effects for each. First, we note that technical center enrollees differ from community college enrollees in terms of both pre-layoff and enrollment characteristics. As shown in Appendix Table \ref{tab:cc_otc_chars}, technical centers tend to enroll many more former manufacturing workers (50 percent of technical center enrollees versus 25 percent of community college enrollees), and those who have higher tenures and pre-layoff earnings. Technical center enrollees also tend to be enrolled for slightly longer periods and are much more likely to obtain a (sub-associate) credential.

To estimate the effects of community college (technical center) attendance, we restrict the sample to enrollees whose first institution is a community college (technical center), and non-enrollees from the main analysis sample. We then implement the same matching specification as our main analysis. Appendix Figure \ref{fig:schooltypes} shows the earnings of enrollees in the two types of institutions and their matched non-enrollees. We find that the patterns of earnings are quite different for community college and technical center attendees. The ``lock-in'' effect for technical center enrollees is much more substantial than that of community college enrollees, of about -\$990 (22 percent) versus -\$283 (six percent) for the first two years post-enrollment. This may reflect the relatively more ``intensive'' nature of the courses offered at technical centers; the fact that many more enrollees gain a credential from a technical center (in approximately the same duration of enrollment) suggests higher coursework intensity. In the third and fourth years post-enrollment, the earnings gain for technical center enrollees is also larger, at \$516 (ten percent) versus \$336 (six percent).

Since the majority of our sample are community college enrollees, it is worth examining whether earnings effects vary by community college quality.  Although quality is multidimensional, we focus on three measures: instructional spending per student (following \citealp{Stange2012}), proportion of students earning a credential (``completing'') within eight years after entry, and median long-run (ten-year) earnings of those who enroll in an institution. The first measure is calculated using information from IPEDS and the latter two are based on the College Scorecard developed by the U.S. Department of Education. We break the sample of community college enrollees into two groups by whether their college is high or low (above or below median within our sample)  in each quality measure and match them to non-enrollees.\footnote{In the enrollee sample, the correlation between enrolling in a school with above-median instructional spending and one with above-median completion rates is -0.17. The correlation between enrolling in a school with above-median instructional spending and one with above-median earnings is -0.25. Finally, the correlation between enrolling in a school with above-median completion rates and one with above-median earnings is 0.30.} Appendix Figure \ref{fig:schoolquality} shows that school quality matters for earnings effects. For those enrolled in schools with high instructional expenditures, the third and fourth year effect is \$419 (seven percent), while the corresponding effect for those attending schools with low expenditures is \$300 (six percent). Similar patterns emerge using the other quality measures: for those attending institutions with high completion rates and earnings, the third and fourth year earnings effect is \$471 (nine percent) and \$443 (eight percent), respectively, versus \$379 (seven percent) and \$291 (five percent), respectively, for institutions with low completion and earnings.

\subsection{Enrollment Effects For Workforce Investment Act Enrollees\label{subsec:WIA}}

As discussed in the introduction, much of the evidence on retraining in the U.S. comes from evaluations of training programs, such as those funded by WIA. While our study population consists of a broader set of workers who seek retraining, this section focuses on the subset of enrollees in our sample who also received training through WIA to better connect to the literature on government-sponsored training programs. 

To identify the subgroup of enrollees who also received WIA training, we rely on administrative WIA data, which contain training dates. We define a worker as a ``WIA enrollee'' if they are an enrollee (i.e., enrolled in a community college or technical center within two years of layoff) \emph{and} they are observed to have participated in WIA training within two years of layoff. We identify 8,771 such enrollees within our main sample of almost 71,000 enrollees. We observe an additional 18,070 individuals in our analysis sample who do not enroll in a public two-year postsecondary institution but who are observed to have WIA training, which occurs because WIA training may take place at employers, private/online schools, or other institutions not captured in our data. In \citet{Fortsonetal2017}, 24 percent of those who trained under the ``full WIA'' arm of the nationally representative WIA Gold Standard Evaluation enrolled in two-year community colleges; the analogous number in our setting is 19 percent. 

Appendix Table \ref{tab:wia_chars} shows that WIA enrollees differ in characteristics from our overall enrollee sample. More than half of WIA enrollees are formerly manufacturing workers, compared with 29 percent for the main analysis sample. They are also older and have higher earnings prior to layoff. In terms of schooling characteristics, WIA enrollees are more likely to enroll in technical centers (45 percent versus 15 percent for the main sample), are enrolled longer (5.2 versus 4.5 quarters), and are more likely to obtain a credential (58 percent versus 26 percent).

Using this sample of WIA enrollees and the non-enrollees from our main analysis sample, we implement the matching estimator discussed in Section \ref{subsec:Matching-Specification} with the same adjustment for the small number of WIA enrollees as in the subsample analysis on previous college experience in Section \ref{subsec:subgroup_returns}.\footnote{We drop 459 enrollee observations that fail our overlap check.} Figure \ref{fig:wia_subgroup} shows the earnings of WIA enrollees and their matched non-enrollees.\footnote{Individuals who do not enroll in a public postsecondary institution but who are observed to have WIA training are in the control group of non-enrollees, though results are identical in percentage terms when we exclude them from the control group.} Compared with Figure \ref{fig:matched_overall} for all enrollees, WIA enrollees have higher average pre-layoff earnings and experience a larger drop at layoff. After enrollment, WIA enrollees endure a larger lock-in effect (-33 percent relative to non-enrollees in the first two years post-enrollment), catch up to matched non-enrollees more slowly, and have smaller effects (four percent; statistically significant at the five percent level) in years three and four. However, as discussed previously for manufacturing workers, the trajectories indicate that effects may increase further into the future.

Our results are broadly consistent with a randomized evaluation of WIA by \citet{Fortsonetal2017}. Their experiment randomly assigned workers to one of three arms: a ``core'' arm offered only core services; a ``core-and-intensive'' arm offered both core and intensive services (non-training services that include informational tools, job search assistance, and placement services); and a ``full WIA'' arm offered the full array of WIA services (including training services if eligible).  The effect of WIA-funded training is captured by the comparison between the core-and-intensive arm and the full WIA arm. \citet{Fortsonetal2017} find no statistically significant difference in earnings between these two arms 12 quarters after randomization, the last time period observed, though the difference in actual training receipt between them is only about nine percentage points (and their estimate is a local average treatment effect that is not directly comparable to our estimates of the treatment-on-the-treated effect). However, the point estimates of \citet{Fortsonetal2017} document similar dynamics as Figure \ref{fig:wia_subgroup}, in that the earnings of the full WIA arm catch up to the core-and-intensive arm at about two years after randomization.

\subsection{Long-Run Effects}

While our data do not span a long enough period to allow for the estimation of longer run enrollment effects for the entire analysis sample, we can examine whether the earnings gains are likely to persist by following a subsample of early enrollees further out. For this analysis, we restrict our attention to workers who enroll no later than 2007Q3 (15,833 enrollees) and compare their earnings to their matched non-enrollees for ten years post-enrollment. Since all workers in this subsample were laid off before the Great Recession, they differ from our main analysis sample in that their earnings (and short-run enrollment effects) are somewhat depressed in the four-year follow-up period as shown in Table \ref{tab:subgroup_ests}. We see this pattern again in Figure \ref{fig:long_run_effects}, where the twin vertical dashed lines denote the two follow-up periods for our main effect estimates. Despite the dip in earnings that begins around the second to third year after enrollment, the effects appear to increase over time as the labor market rebounds, ultimately resulting in a 13 percent gain in the tenth year after enrollment (point estimate \$706, standard error 73). We find that enrollees are five percentage points more likely to be employed at the end of the ten-year follow-up period, and conditional on employment, enrollees earn six percent more (the extensive margin explains 58 percent of the overall earnings gain). Our estimates are higher than \citet{Jacobson_etal2005_JE}'s preferred extrapolated long-run effects of retraining from the 1990s (six to eight percent as mentioned in Section \ref{sec:Intro}), which they suggest may be downwardly biased due to differential pre-trends. Our estimates are qualitatively consistent with the long-run positive effects of TAA for high-quality (long-duration) training by \citet{Hyman2018}, although it is hard to draw quantitative comparisons because \citet{Hyman2018} reports intent-to-treat effects and the TAA treatment includes both benefit payments and training.

\subsection{What Does Schooling Do?\label{subsec:What-Does-Schooling}}

In this section, we explore the mechanisms underlying our main estimates. We first present results from simple decomposition analyses showing that the positive earnings effects are primarily driven by increased employment in certain industries, especially healthcare. We then analyze the course and credential data and document that the ``excess'' employment in these industries can be largely accounted for by the number of workers taking related courses.

We start our analysis by exploring the extent to which enrollees are more likely to leave their pre-layoff (two-digit) industry. We focus on two-digit industries rather than finer categories to highlight the contrast between pre-layoff and post-layoff jobs. The right (left) side of Figure \ref{fig:ind_switch-prob} Panel A shows the probability that an enrollee (matched non-enrollee) is employed in their pre-layoff industry, not employed, or employed in a different industry, over the four-year follow-up period. Enrollees are less likely than non-enrollees to be employed immediately after enrollment, and Figure \ref{fig:ind_switch-prob} shows that this difference is almost entirely explained by the differential probability of re-employment within the pre-layoff industry, as the two groups are about equally likely to switch industries in quarter zero. Over time, the probability of working in a different industry increases more quickly for enrollees: by the end of the four-year follow-up period, they are nine percentage points more likely to work in a different industry. On net, enrollees are seven percentage points more likely to be employed in quarter 16 after enrollment.

Panel B of Figure \ref{fig:ind_switch} decomposes the average earnings of enrollees and non-enrollees in each quarter of the follow-up period from Figure \ref{fig:matched_overall} into components contributed by industry stayers and industry switchers (following \citealp{Autoretal2014}). Specifically, each group's quarterly earnings are decomposed as
\begin{eqnarray}
\text{\ensuremath{E[\text{Earnings}]}} & = & \underbrace{\Pr(\text{Empl. in Same Ind.)\ensuremath{\cdot E[\text{\text{Earnings}| Empl. in Same Ind.}]}}}_{(\text{i})}+\nonumber \\
 &  & \underbrace{\Pr(\text{Empl. in Diff. Ind.)\ensuremath{\cdot E[\text{\text{Earnings}| Empl. in Diff. Ind.}]}}}_{(\text{ii})}\label{eq:ind_switch}
\end{eqnarray}
where ``Empl. in Same Ind.'' and ``Empl. in Diff. Ind.'' denote employment in the pre-layoff and non-pre-layoff industry, respectively. The gray lines in the figure represent components (i) and (ii) in equation (\ref{eq:ind_switch}) for enrollees (solid) and non-enrollees (dashed), and each pair sums up to their respective black lines. Consistent with Panel A, the earnings difference in quarter zero is driven by a higher contribution from industry stayers (component i) among non-enrollees. Over time, however, industry switchers among enrollees (component ii) out-contribute their non-enrollee counterparts. By quarter 16, industry switching entirely explains the overall earnings gain: the difference in component (ii) between enrollees and non-enrollees is 103 percent of the overall earnings gain, while the difference in component (i) is -3 percent.\footnote{We further show in Appendix \ref{sec:akm_firm_effects} that the effect is not solely driven by industry switchers systematically finding employment in different types of firms (i.e., high-wage firms).}

Appendix Figure \ref{fig:destination_ind} shows which industries drive the enrollment effects by plotting the number of enrollees and matched non-enrollees employed in each sector (or not employed at all) in quarter 16. Among women (Panel A), enrollees are much more likely to work in healthcare. Among men (Panel B), enrollees are also more likely to work in healthcare (though to a lesser extent) and construction. Both male and female enrollees are less likely to work in manufacturing. To connect these findings to Figure \ref{fig:ind_switch}, Appendix Figure \ref{fig:destination_ind_switching} shows that industry switchers make up the bulk of the increased employment in healthcare. This pattern of industry switching driving the effects of enrollment is consistent with \citet{CarruthersSanford2018}, who study the effects of attending sub-associate level institutions in Tennessee. This industry-mobility finding is also in a similar spirit to \citet{HumlumMunchRasmussen2025} and \citet{Belotetal2025}, where the employment gains are driven by occupation mobility (as mentioned in Section \ref{sec:Data}, our wage records do not contain occupation; as a result, we cannot study occupation mobility).

Given that increased employment in healthcare and construction explains most of the gains associated with schooling, we now provide evidence that the courses taken and credentials received by individuals in these sectors are related to their work. We link courses or credential subjects to industries using a mapping of academic subjects to occupations from the National Center for Education Statistics and the joint distribution of occupations and industries from the National Employment Matrix by the Bureau of Labor Statistics. Specifically, we define an academic subject as associated with an industry if the subject prepares a worker for an occupation where more than a quarter of the occupation is employed in a single two-digit industry (see Appendix \ref{subsec:courses_appendix} for details). For community colleges, where individuals typically take multiple courses, we associate an enrollee's coursework with the modal industry across her courses.

Figure \ref{figrelated_course} shows the number of enrollees in each post-layoff industry who have taken courses that are linked to that industry as described above. We see a large fraction of women (Panel A) employed in healthcare having taken related courses. The analogous plot for men (Panel B) shows that meaningful fractions of enrollees employed in construction, healthcare, and manufacturing have taken related courses. Strikingly, for both men and women, the number of healthcare or construction course-takers mirrors the enrollee-matched-non-enrollee employment difference in the sector. Turning to the analogous results for industry-related credentials, Figure \ref{fig:related_credential} shows that for both men and women, the ``excess'' healthcare employment matches those with health-related credentials, but the same is not true for construction.\footnote{We find that 73 percent 
 of those who have taken healthcare courses and are employed in healthcare receive a credential, compared to the overall credential rate of 26 percent. The credential data suggest that it is not a specific occupation within healthcare that is driving the result: of the approximately 3,500 enrollees employed in healthcare with credentials, 26 percent received a credential in Licensed Practical Nursing, 14 percent in Nursing Assistance, and 12 percent in Registered Nursing.} The finding that the healthcare sector is a key driver of employment growth among those who retrain is consistent with empirical evidence from other studies. \citet{Grosz2020} shows that one specific health program (associate's degree in nursing) leads to substantial earnings gains relative to not enrolling at all. Others find evidence that earnings gains associated with credentials in health appear larger than those for non-health subjects (e.g.,
\citealp{Bohnetal2016}; \citealp{Stevensetal2019}). The lack of credentials in construction suggests that the mechanism through which schooling improves labor market prospects may differ by industry. While training may increase access to jobs that require licenses or certification in healthcare, the mechanism is more likely to be industry-specific skill-building in construction.\footnote{We do observe that industry switchers are more likely to have obtained a credential generally. 32 percent of switchers have a credential while only 22 percent of the non-employed and 19 percent of those who stayed in the same industry have credentials.}

\subsection{Enrollment Effects and Completers\label{subsec:returns_degree}}

An interesting question that remains is whether the enrollment effects differ by the nature of the schooling spell (i.e., whether someone ultimately obtains a credential, or ``completes''). One way to answer this question would be to make assumptions about the earnings processes and selection into completion. This is the prevailing method in the literature for estimating credential effects, and is used, for example, by  \citet{Jepsen_etal2014} and \citet{Stevensetal2019}. These papers estimate the effects of different types of community college degrees by assuming that conditional on an enrollee fixed effect, there is no more selection into completion of different types of credentials. We have run the same fixed effects regression, but we find that, in our sample, there are large differences in earnings patterns (particularly in the transitory components) among completer and non-completer enrollees that would violate the identifying assumptions of a fixed effects model. This may not be surprising given that we also find differential trends in earnings leading up to enrollment among all UI claimants when we control only for individual fixed effects and displacement effects, as discussed in Appendix \ref{sec:Alternative-Identification-Strat}.

In a previous working paper (\citealp{LeungandPei2023}), we tried out an alternative strategy where we formally defined the enrollment-effect-for-completers parameter and constructed bounds for it. In that framework, we borrowed from the principal stratification literature (e.g., \citealp{Frangakis2002}) and treated completer as a pre-determined type that was only revealed after a worker enrolled. This assumption might be implausible if training completion was strongly influenced by shocks after training began. Regardless, this approach led to uninformative bounds, which could be largely attributed to the relatively low fraction of completers.

While the bounds were not informative, we provide suggestive evidence that the post-enrollment earnings growth is concentrated among completers. We decompose enrollee earnings into parts contributed by completers and non-completers and present the results in Appendix Figure \ref{fig:degree_decomp}.\footnote{Specifically, this graph plots components (i) and (ii) in the following decomposition for enrollees: $\text{\ensuremath{E[\text{Earnings}]}}=\underbrace{\Pr(\text{Completer})\ensuremath{\cdot E[\text{\text{Earnings}| Completer}]}}_{(\text{i})}+\underbrace{\Pr(\text{Non-completer)\ensuremath{\cdot E[\text{\text{Earnings}| Non-completer}]}}}_{(\text{ii})}$.} We see completers contribute disproportionately to earnings growth. Between quarter 0 and 16, enrollees gained \$2,566 in quarterly earnings over the four years, with completers (who make up 26 percent of enrollees) contributing \$982 and non-completers (who make up 74 percent of enrollees) contributing \$1,585. In other words, completers' gain is more than three quarters larger than that of non-completers, and if the completer proportion were increased to 50 percent, the average enrollee earnings gain would increase to \$2,959. 

\section{Cost-Benefit Analysis of Further Education\label{sec:IRR}}

In light of the estimated average earnings impacts, can we justify the investment in further education? In this section, we provide back-of-the-envelope calculations on the private and social returns to retraining an additional worker. The private return compares the net present value of the stream of \emph{after}-tax earnings impacts of retraining against the out-of-pocket education expenses an enrollee has to pay upfront. The social return compares the net present value of the stream of \emph{pre}-tax earnings impacts of retraining against the overall cost of enrolling an additional unemployed worker.

While we rely on our estimates of earnings impacts to calculate the returns, we need further assumptions regarding other inputs, such as years of work life remaining and tax rates. First, we assume that an average enrollee has 30 years remaining in her work life just prior to retraining. Given that an average enrollee in our sample is about 35 years old at layoff (Table \ref{tab:descriptives}) and that it takes slightly less than a year from layoff to enrollment (Table \ref{tab:enrollment_chars}), our assumption implies that the enrollee will stop working at around 66, which is consistent with \citet{Jacobson_etal2005_ILR} and just under the normal retirement age of 67 faced by most of our enrolled workers. Second, we use the ten-year post-enrollment real earnings impacts graphed in Figure \ref{fig:long_run_effects} for workers who began schooling before the fall quarter of 2007 and assume that the earnings gains from year 11 until year 30 will stay at the tenth year level. The earnings impacts increase monotonically from $-$\$2,337 in the first year to \$2,915 in the ninth year before dipping slightly to \$2,824 in the tenth year (the tenth year impact does not differ statistically significantly from the ninth year impact). Our returns measures will be biased downward if the gains continue to rise beyond the tenth year, but they will overstate the benefit of further education if the gains eventually fall. We return to the latter possibility at the end of this section. Third, we use a real interest rate of two percent to discount future earnings when calculating the net present value. This is based on the fact that the daily Treasury real long-term rates, as calculated from the yields of outstanding long-term Treasury Inflation-Protected Securities (TIPS), averaged 1.98 percent between 2009 and 2010, the two years during which enrollment in our sample peaked. Our real interest rate is lower than the four percent used by \citet{Jacobson_etal2005_ILR} due to the different time period we examine.\footnote{\citet{Jacobson_etal2005_ILR} study workers who enrolled in the 1990s when the long-term Treasury yields were around seven percent and inflation rates around three percent. TIPS were not introduced until the late 1990s, and the published Treasury real long-term rates became available in early 2000, at which point they hovered just above four percent---consistent with the choice by \citet{Jacobson_etal2005_ILR}.} Finally, we follow \citet{Jacobson_etal2005_ILR} and assume that workers pay 25 percent of their earnings in various taxes. Because federal income tax rates were lower during our sample period than in the 1990s, this average tax rate may be on the high side, implying a conservative private return estimate. The tax rate does not enter the social return calculations.

To compute the private cost of further education, we rely on IPEDS and the Digest of Education Statistics by NCES. To estimate the yearly out-of-pocket expenses for a typical enrollee, we first subtract the expected amount of grant from the sum of annual tuition and book costs during the 2010-2011 academic year (AY). While IPEDS reports tuition and book costs directly, we need to calculate the expected amount of grant. To do this, we multiply the reported average grant amount among students receiving a grant (an estimate of $E[\text{grant amount}|\text{receiving grant}]$) by the fraction of students receiving a grant (an estimate of $\Pr(\text{receiving grant})$). For community colleges, the two measures are \$4,194 and 73 percent, respectively; for technical centers, they are \$4,559 and 74 percent. We estimate out-of-pocket expenses separately for the two types of institutions and take an average weighted by the proportion of enrollees attending each.\footnote{Unlike two- and four-year institutions that report cost information for each AY, most technical centers report cost statistics by program, and we use the average costs corresponding to the largest program offered.} To the extent that UI claimants are more likely to receive a grant, or receive larger grants when they do, we overstate out-of-pocket costs and understate private returns.

To estimate the social investment in an additional enrollee, which encapsulates both the enrollee's expenses and government subsidies, we follow \citet{Rouse1998} and proxy it with the variable cost per full-time-equivalent (FTE) student. Like \citet{Rouse1998}, who estimates the cost of educating a student in an associate degree program, we obtain our variable cost by excluding from the total expenditure any fixed costs (spending on administration, public service, operation and maintenance) and research outlays (which are not important for training enrollees in our sample and account for less than 0.1 percent of the total expenditure). We compute the cost measures using information for the 2010-2011 AY from \citet{Digest2012}, and the annual social investment per FTE student is \$8,084 in a community college.\footnote{This figure represents the national average and is not Ohio-specific. \citet{Digest2012} only break down the total expenditure into detailed categories at the national level, meaning that we cannot conduct the \citet{Rouse1998} exercise for Ohio. That said, the reported total expenditure per two-year FTE student is \$12,398 nationally and \$12,346 in Ohio. Given the small difference, the national average is likely to serve as a good substitute.} Expenditure information is not available for technical centers in \citet{Digest2012}. To be conservative, we proxy it with a higher measure of \$14,122, which is the variable cost for an FTE student at a four-year institution.\footnote{To estimate the cost of educating a two-year college student in a four-year institution, \citet{Rouse1998} adjusts the variable cost of four-year colleges because it is cheaper to add a two-year student than an upper class undergraduate or graduate student. We arrive at the dollar figure by adopting the same adjustment for the cost of educating a student enrolled in a technical center.} Lastly, since the average length of enrollment is 4.5 terms, we multiply the yearly cost measures by $4.5/2=2.25$, assuming two academic terms per year. There could be more than two terms per year, and enrollees in our sample might not have enrolled full time, both of which imply an overestimate of cost and a conservative estimate of the return.

Putting all the numbers together, the net present value of the average post-tax earnings impacts is \$37,056 against an out-of-pocket cost of \$5,976 (\$5,855 based on the 2010-2011 AY, converted to 2012 dollars for comparability with our earnings impacts), and the resulting private net benefit an enrollee accrues is \$31,080. The net present value of the average pre-tax earnings impacts is \$49,408 against a social investment of \$20,635 (\$20,217 based on the 2010-2011 AY, converted to 2012 dollars), and the resulting social net benefit per enrollee is \$28,773. Another way to interpret these quantities is that the worker gets \$6.20 for every dollar invested, and society gets \$2.39. These ratios are high compared to those reported by \citet{Jacobson_etal2005_ILR}, whose private and social benefit-to-cost ratio estimates are in the range of \$1.69-\$4.52 and \$1.20-\$2.61, respectively (their Table 5A columns 1-4). The lower two percent real interest rate we use is largely responsible for the differences---with a four percent rate as in \citet{Jacobson_etal2005_ILR}, our private and social benefit-to-cost ratios become \$4.47 and \$1.73, respectively.

Our estimates imply a private annual internal rate of return (IRR) of 15 percent and a social IRR of 8 percent. \citet{Jacobson_etal2005_ILR} do not report the private IRR, so we compare our estimate against those by \citet{HeckmanLochnerTodd2006Handbook}. Our private IRR is somewhat higher than those reported by \citet{HeckmanLochnerTodd2006Handbook} for a comparable education level, which come from applying generalized Mincer regressions to Decennial Census data. After accounting for tuition and taxes, \citet{HeckmanLochnerTodd2006Handbook} find private IRRs between 8 and 14 percent for two more years of study among those with 12 years of education (their Table 4).\footnote{Our sample is arguably more comparable to this population than the population going from 14 to 16 years of education, as the latter consists of many who obtain a Bachelor's degree. Our higher private IRR of 15 percent may be driven by the different time period examined. \citet{HeckmanLochnerTodd2006Handbook} show their IRRs increased by 6 to 17 percentage points between 1950 and 1990, and further gains are likely going into the 2000s. Another potential reason for our higher IRR is that our sample of UI claimants with previous labor force attachment is more positively selected.} Our social IRR is comparable to the preferred estimates of 4 to 14 percent by \citet{Jacobson_etal2005_ILR} that pass a specification check (their Table 5B column 2).

We have mentioned reasons that these estimates may be conservative (e.g., our approximation of tax rates and the use of investment per FTE student to measure technical center training cost). Additionally, the ``long-run'' earnings estimates used for this calculation come from an early cohort of enrollees who reentered the labor market at the height of the Great Recession and likely saw lower earnings gains compared to later cohorts, as seen in Appendix Figure \ref{fig:by_year}. Finally, our implicit assumption that out-of-pocket expenses are paid in cash is also likely to understate the private IRR.\footnote{According to IPEDS data for the 2010-2011 AY, 51 percent of Ohio students at community colleges and 40 percent at technical centers took out loans. The average loan amounts were \$4,501 and \$5,548, respectively. As in \citet{Rouse1998} and \citet{Jacobson_etal2005_ILR}, loans (and grants) to students do not affect the social IRR calculation (the corresponding investment is calculated from the institutions' perspective, which is agnostic about how tuition and related expenses are paid).} For example, if the average out-of-pocket expenses are financed entirely with an unsubsidized Stafford loan at the 6.8 percent nominal rate during our sample enrollment period, the private IRR increases to 20.7 percent.\footnote{We assume a two percent annual inflation rate; it is almost exactly the average in the time span of our post-enrollment data.} This higher IRR reflects the fact that when training returns exceed borrowing costs, it is optimal to cover education expenses with loans.

Of course, by using average measures that include \emph{ex post} earnings impacts in our calculations, we ignore both uncertainty and heterogeneity. When making enrollment decisions, workers are unsure how much they will make post-training. Individuals also have different abilities, leading to different labor market returns and capacity to pay back loans. Similarly, workers incur different non-monetary training costs we do not account for---these include psychic costs or long hours spent studying, which \citet{Carneiroetal2003}, \citet{Cunhaetal2005,Cunhaetal2007}, and \citet{CunhaandHeckman2007,CunhaandHeckman2008} show can be substantial. We acknowledge that accounting for uncertainty and heterogeneity tends to lower the return estimates (\citealp{HeckmanLochnerTodd2006Handbook}), but we abstract away from them to facilitate comparisons with existing estimates. Similarly, we follow \citet{Rouse1998}, \citet{Jacobson_etal2005_ILR}, and \citet{HeckmanLochnerTodd2006Handbook} and abstract away from general equilibrium effects.\footnote{For example, sending many more unemployed claimants back to school to prepare for a new career in healthcare will shift the supply of new healthcare workers, which will in turn depress wages and crowd out other workers (\citealp{McCalletal2016_Handbook} refer to these as relative skill price and displacement effects, respectively). The general equilibrium effect could also be positive---as a referee suggests, having a more educated workforce may improve health, reduce crime, and strengthen democracy (see \citealp{Lochner2011} for a nice review).}

With these abstractions, the only assumption we invoke that may substantially overstate the benefit of further education is the persistence of training effects over a twenty-year horizon (between year ten and year thirty). While this extrapolation is consistent with \citet{Jacobson_etal2005_ILR}, it may not serve as a good approximation. For example, many of our enrollees might have lost jobs again in the COVID-19 pandemic, including those working in healthcare where routine and elective procedures were postponed (\citealp{Scott2020Vox}), and the earnings gain might have decreased substantially as a result. One way to interpret the numbers without the full twenty-year extrapolation is to ask how long it takes to break even on the educational investment. The impacts graphed in Figure \ref{fig:long_run_effects} suggest private (social) investment breaks even at 8 (14) years after enrollment begins.

\section{Conclusion\label{sec:Conclusion}}

In this paper, we estimate the effects of retraining unemployed workers. Linking together high-quality administrative records, we follow the population of Ohio UI claimants between 2004 and 2011 who enroll in public postsecondary institutions. Adopting a matching method that bridges two strands of the dynamic treatment effect literature, we find that enrollees experience a ``lock-in'' effect of depressed earnings immediately after enrolling but see an average earnings gain of at least \$348 per quarter (about six percent) over the third and fourth years post-enrollment. We find that much of the gain is driven by enrollees who take courses and are subsequently employed in health-related fields. A longer-run follow-up of an early subsample suggests that the gains persist and widen over a ten-year period.

With our earnings impact estimates, we conduct a cost-benefit analysis from the private and social perspectives. We find that the out-of-pocket investment from an average worker breaks even 8 years after enrolling, whereas the social investment (including both private investment by the worker and government subsidies) breaks even in 14 years. Assuming an average worker stops working around the normal retirement age as in \citet{Jacobson_etal2005_ILR}, the private benefit is over \$6 for every dollar of her out-of-pocket investment in schooling, and the social benefit is about \$2.4 for every dollar of total educational investment. The implied private and social IRRs are 15 percent and 8 percent, respectively, which are comparable to estimates by \citet{HeckmanLochnerTodd2006Handbook} and \citet{Jacobson_etal2005_ILR}.

\citet{BraxtonTaska2023,BraxtonandTaska2025} show that technological change that leads to new skill requirements is a major driver of earnings declines after job loss. This suggests that policies that encourage and enable unemployed workers to retrain, including efforts to expand community college access, can therefore be beneficial in the long run. Specifically, policies that target UI recipients to ease their transition to enrollment may be effective based on existing evidence. \citet{BarrTurner2015} show that UI benefit extensions during the Great Recession increased enrollment among the unemployed. Using the same temporal variation in benefit extensions, we replicate their analysis in Appendix \ref{sec:UI-enrollment} and find a similar effect within Ohio: for every 10-week increase in UI benefits, there is a ten percent increase in enrollment, or an additional 1,200 enrollees per year. Combined with our estimates of positive private and social returns to enrollment, this indicates that UI policies may have social benefits beyond what is typically considered.

While we find positive returns to enrollment on average, we caution that our cost-benefit analysis abstracts from the uncertainty and heterogeneity in enrollment effects. Indeed, an open question remains as to whether gains vary across training types. Although we find that the retraining effects are driven by course-taking and subsequent employment in healthcare and construction, evidence that can guide policymakers and unemployed workers on the most effective type of training is a fruitful avenue for future research.
\pagebreak{}
\section*{References}

\begin{singlespace}

\begin{btSect}[jpe]{opportunity}
\btPrintCited
\end{btSect}
\end{singlespace}

\clearpage{}
\begin{figure}[H]
\caption{Earnings of Enrollees and Matched Non-enrollees}
\label{fig:matched_overall}
\begin{centering}
\smallskip{}
\par\end{centering}
\begin{centering}
\includegraphics{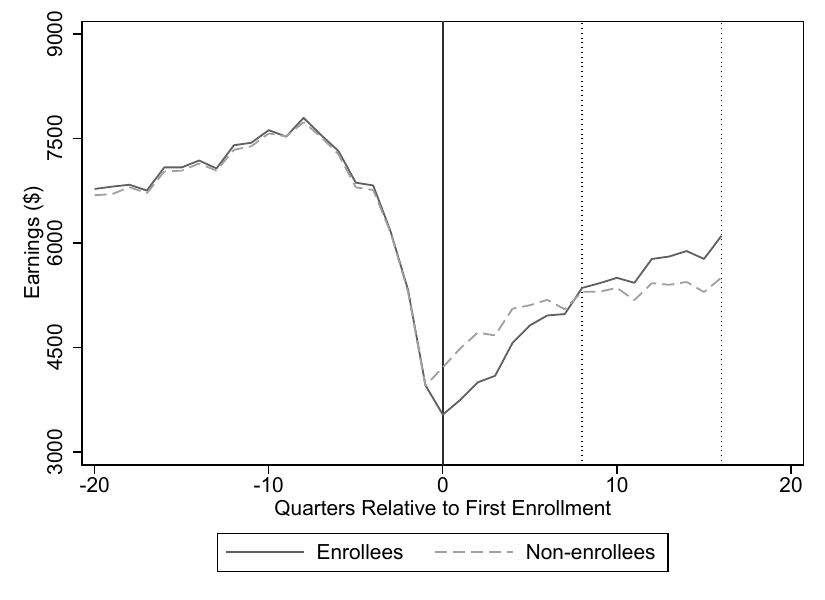}
\par\end{centering}
\centering{}%
\begin{minipage}[t]{0.8\columnwidth}%
Notes: This figure plots the average quarterly earnings of enrollee and matched non-enrollee UI claimants. The solid vertical line denotes the quarter of first enrollment, and the vertical dashed lines denote eight and 16 quarters after first enrollment. $N=141,758$, corresponding to 136,074 unique individuals.
\end{minipage}
\end{figure}

\clearpage{}
\begin{figure}[H]
\caption{Earnings of Enrollees Who Participated in WIA and Matched Non-enrollees}
\label{fig:wia_subgroup}
\begin{centering}
\includegraphics{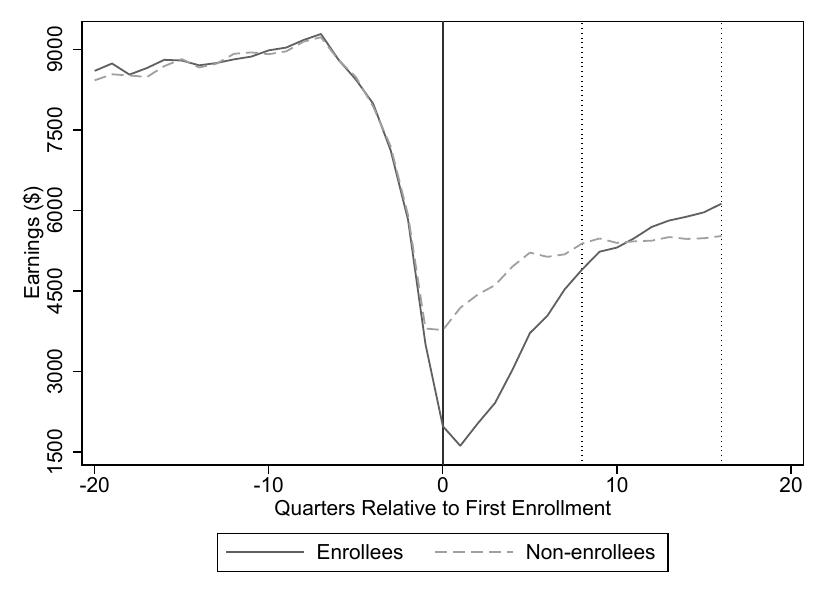}
\par\end{centering}
\centering{}%
\begin{minipage}[t]{0.8\columnwidth}%
Notes: This graph shows the average quarterly earnings of enrollees who received WIA training services, and their matched non-enrollees. The solid vertical line denotes the quarter of first enrollment, and the vertical dashed lines denote eight and 16 quarters after first enrollment. $N=16,624$ UI claims, corresponding to 16,487 unique individuals.
\end{minipage}
\end{figure}
\clearpage{}

\begin{figure}[H]
\caption{Long-run Earnings of Enrollees and Matched Non-enrollees}
\label{fig:long_run_effects}
\begin{centering}
\includegraphics{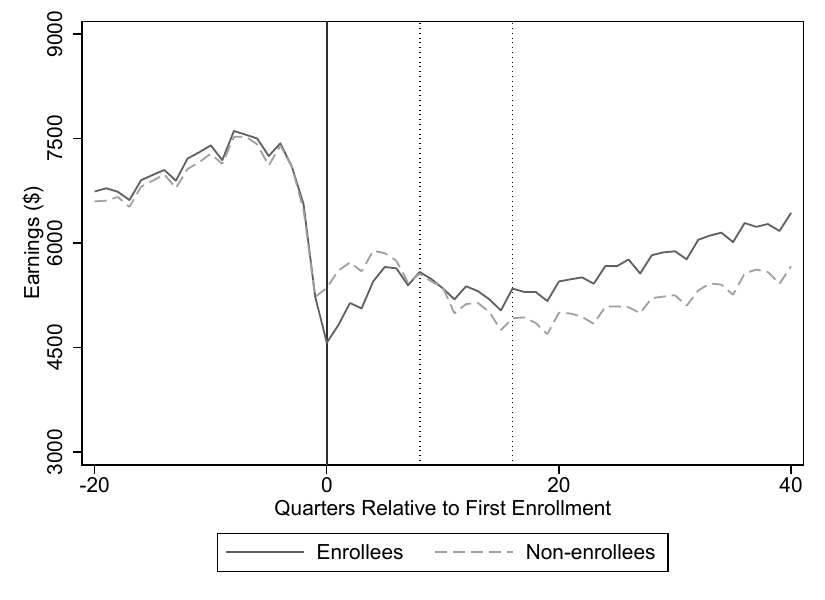}
\par\end{centering}
\centering{}%
\begin{minipage}[t]{0.8\columnwidth}%
Notes: This graph shows the average quarterly earnings of UI claimants who enrolled from 2004 through the third quarter of 2007 and their matched non-enrollees. The solid vertical line denotes the quarter of first enrollment, and the vertical dashed lines denote eight and 16 quarters after first enrollment. $N=31,666$ UI claims, corresponding to 30,973 unique individuals.
\end{minipage}
\end{figure}

\begin{figure}[H]
\caption{Industry Switching Among Enrollees and Matched Non-enrollees}
\label{fig:ind_switch-prob}\label{fig:ind_switch}
\begin{centering}
\smallskip{}
\par\end{centering}
\begin{centering}
(A) Probability of Switching Industries Over Time
\par\end{centering}
\begin{centering}
\includegraphics[scale=0.85]{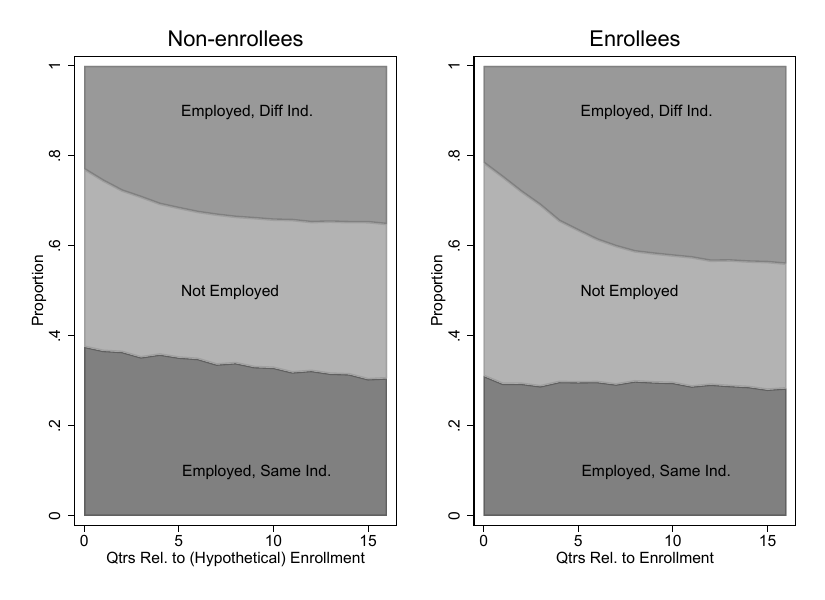}
\par\end{centering}
\begin{centering}
(B) Decomposition of Earnings: the Role of Industry Switching
\par\end{centering}
\begin{centering}
\includegraphics[scale=0.85]{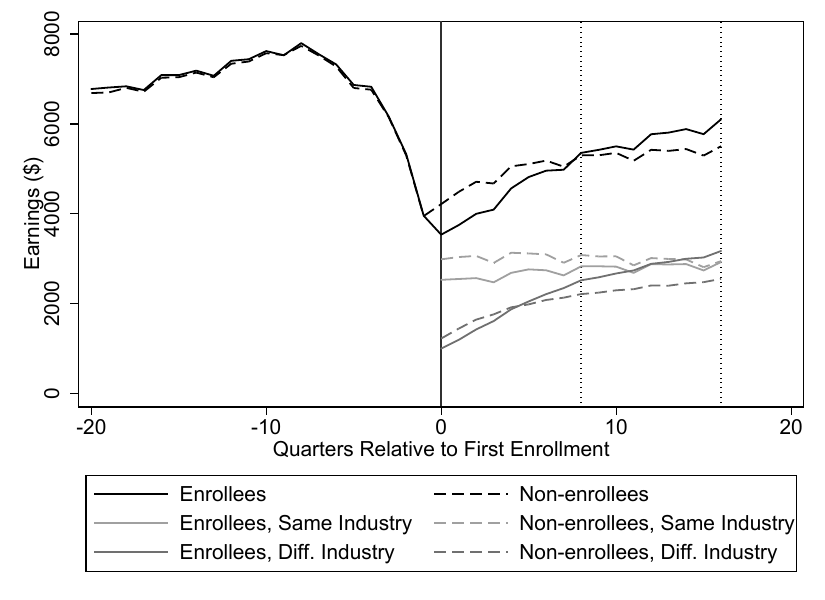}
\par\end{centering}
\centering{}%
\begin{minipage}[t]{0.9\columnwidth}%
Notes: The figures in Panel A plot the probability of employment in the same (pre-layoff) two-digit industry, employment in a different two-digit industry, and non-employment over time for enrollees and matched non-enrollee UI claimants. Panel B plots the average quarterly earnings of enrollee and matched non-enrollee UI claimants (black solid and dashed lines). The gray lines disaggregate the post-enrollment earnings into two components: average quarterly earnings from the pre-layoff (``same'') industry and from a different industry, each scaled by the probability of employment in either the same or different industries. The solid (dashed) gray lines sum up to the solid (dashed) black lines. $N=141,758$, corresponding to 136,074 unique individuals.
\end{minipage}
\end{figure}

\clearpage{}
\begin{figure}[H]
\caption{Industry-Related Course-Taking, by Industry of Employment at 16th Quarter Post-Enrollment}
\label{figrelated_course}
\begin{centering}
\smallskip{}
\par\end{centering}
\begin{centering}
(A) Women
\par\end{centering}
\begin{centering}
\includegraphics[scale=0.8]{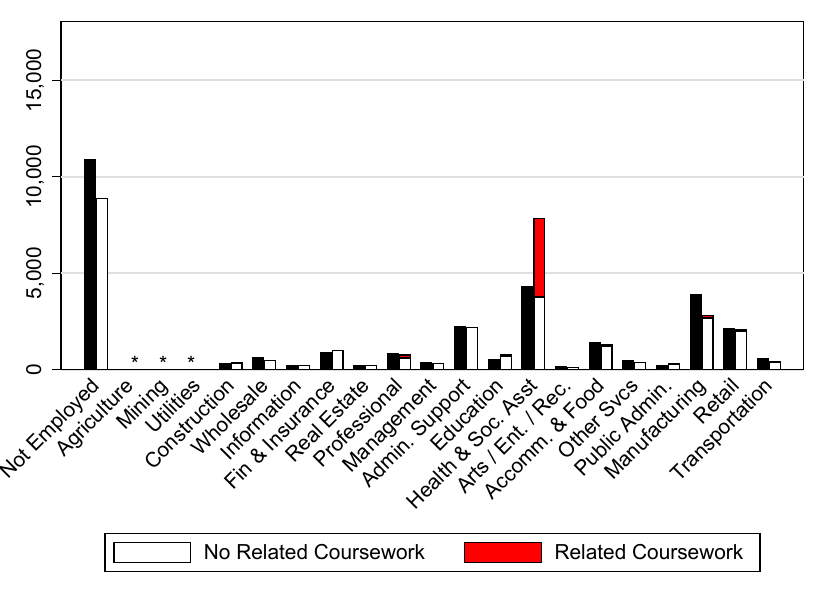}
\par\end{centering}
\begin{centering}
(B) Men
\par\end{centering}
\begin{centering}
\includegraphics[scale=0.8]{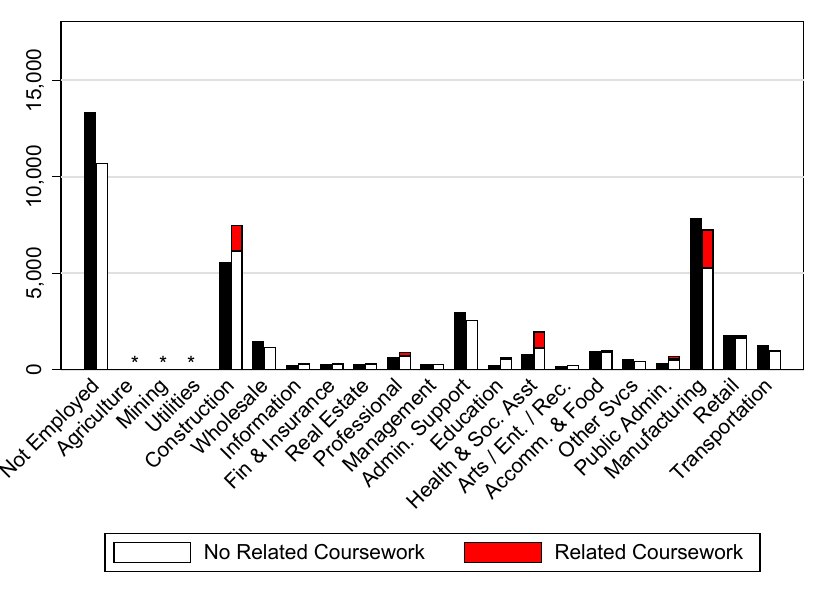}
\par\end{centering}
\centering{}%
\begin{minipage}[t]{0.8\columnwidth}%
Notes: In each plot, the right side of each pair of bars shows the number of enrollees who take or do not take sector-related coursework. Courses are defined as sector-related if the course prepares a worker for an occupation where a large proportion of the occupation is concentrated in one industry (see text for details). The left bar shows the number of matched non-enrollees in each sector. Agriculture, Mining, and Utilities sectors have fewer than 200 workers in each enrollee/non-enrollee cell and are not plotted. If there are 10 or fewer enrollees in a sector who take sector-related coursework, it is not plotted.
\end{minipage}
\end{figure}

\clearpage{}
\begin{figure}[H]
\caption{Industry-Related Credential Receipt, by Industry of Employment at 16th Quarter Post-Enrollment}
\label{fig:related_credential}
\begin{centering}
\smallskip{}
\par\end{centering}
\begin{centering}
(A) Women
\par\end{centering}
\begin{centering}
\includegraphics[scale=0.8]{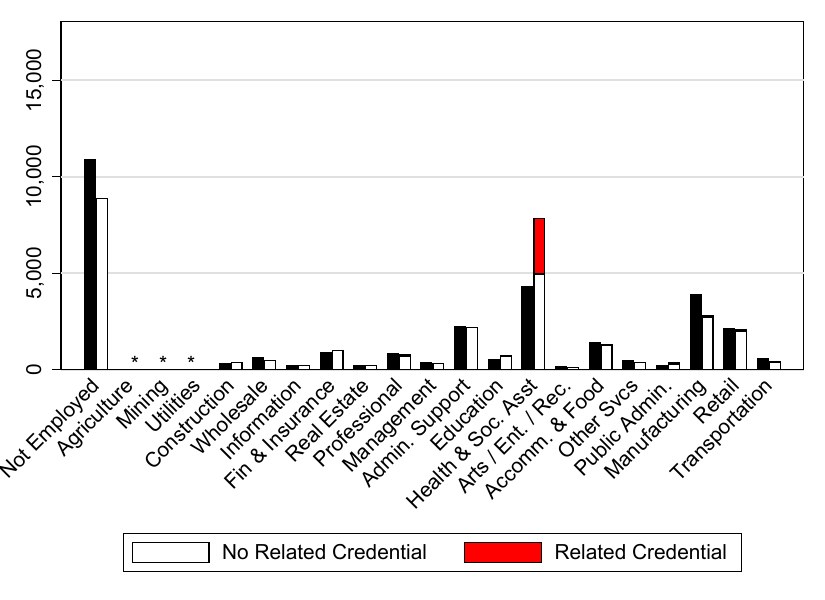}
\par\end{centering}
\begin{centering}
(B) Men
\par\end{centering}
\begin{centering}
\includegraphics[scale=0.8]{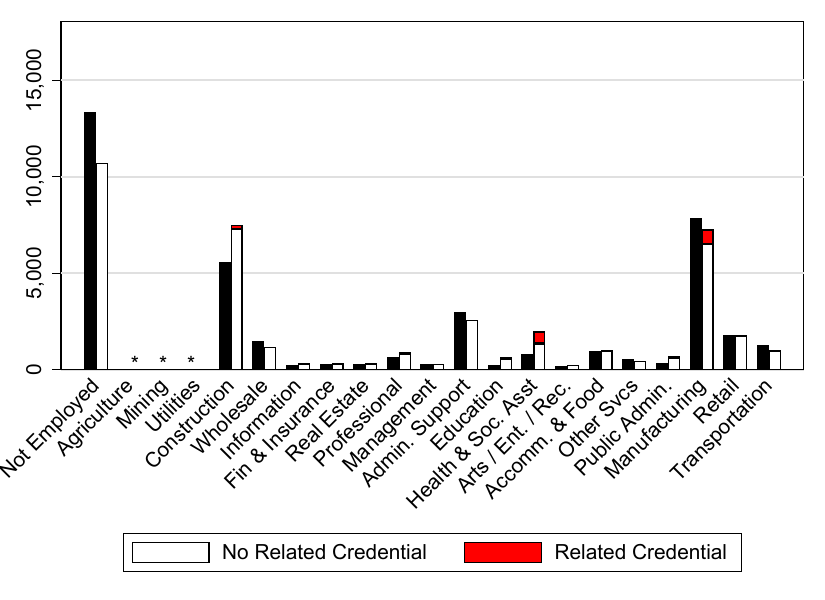}
\par\end{centering}
\centering{}%
\begin{minipage}[t]{0.8\columnwidth}%
Notes: In each plot, the right side of each pair of bars shows the number of enrollees who obtain or do not obtain sector-related credentials. Credentials are defined as sector-related if the credential prepares a worker for an occupation where a large proportion of the occupation is concentrated in one industry (see text for details). The left bar shows the number of matched non-enrollees in each sector. Agriculture, Mining, and Utilities sectors have fewer than 200 workers in each enrollee/non-enrollee cell and are not plotted. If there are 10 or fewer enrollees in a sector who have sector-related credentials, it is not plotted.
\end{minipage}
\end{figure}

\newpage{}
\begin{table}[H]
\caption{(A) Descriptive Characteristics of Enrollees and Non-enrollees}
\label{tab:descriptives}\medskip{}

\begin{centering}
\includegraphics[viewport=45bp 300bp 612bp 737bp]{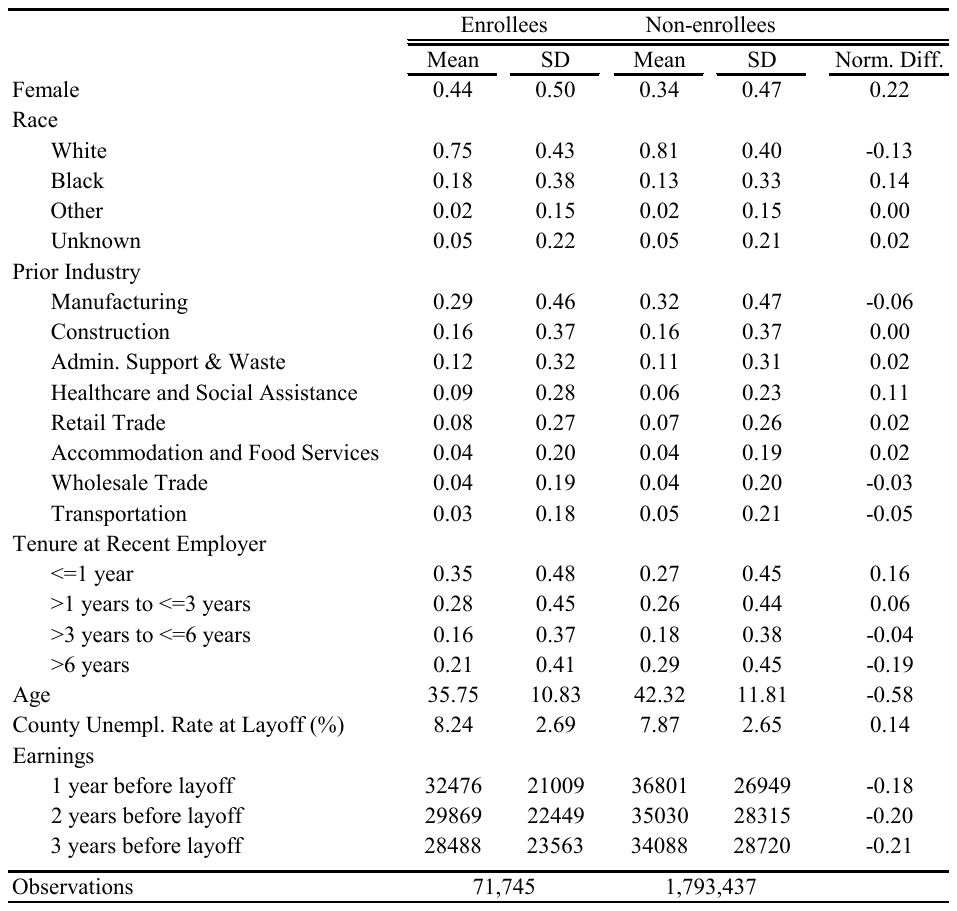}
\par\end{centering}
\centering{}%
\begin{minipage}[t]{0.96\columnwidth}%
Notes: ``Enrollees'' (``Non-enrollees'') are UI claimants who enroll (do not enroll) in a public postsecondary institution in Ohio within two years of filing a UI claim. ``SD'' denotes standard deviation. ``Norm. Diff.'' is the normalized difference defined in Section \ref{subsec:Estimation-Strategy-and}. All earnings are expressed in 2012 dollars.
\end{minipage}
\end{table}
\newpage{}

\setcounter{table}{0} 
\begin{table}[H]
\caption{(B) Descriptive Characteristics of Enrollees and \uline{Matched} Non-enrollees}
\label{tab:descriptives_matched}\medskip{}

\begin{centering}
\includegraphics[viewport=66bp 285bp 612bp 737bp]{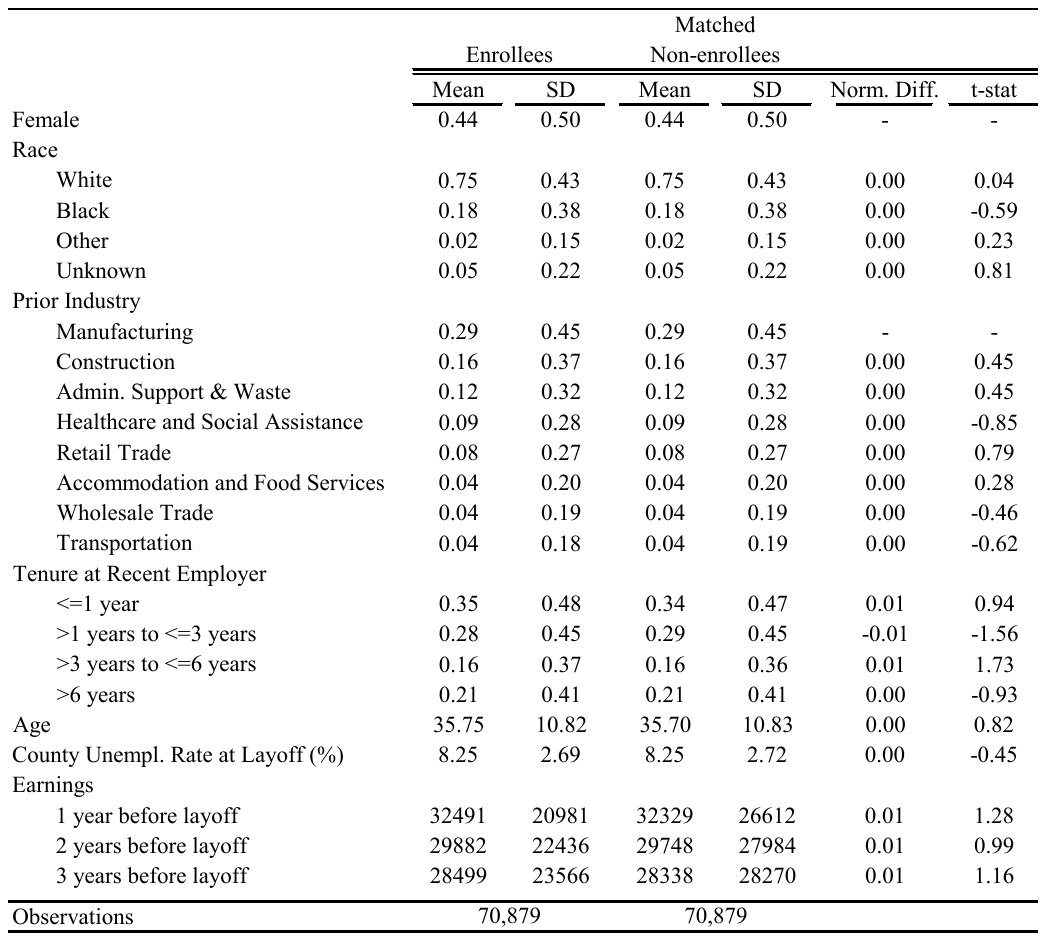}
\par\end{centering}
\centering{}%
\noindent\begin{minipage}[t]{1\columnwidth}%
Notes: ``Enrollees'' are UI claimants who enroll in a public postsecondary institution in Ohio within two years of filing a UI claim, and ``Matched Non-enrollees'' are their matched comparison group. ``SD'' denotes standard deviation. ``Norm. Diff.'' is the normalized difference defined in Section \ref{subsec:Estimation-Strategy-and}. ``t-stat'' is the t-statistic corresponding to the difference in means between enrollees and matched non-enrollees.
\end{minipage}
\end{table}
\newpage{}

\begin{table}[H]
\caption{Enrollment Characteristics}
\label{tab:enrollment_chars}\medskip{}

\includegraphics[viewport=3bp 468bp 612bp 737bp]{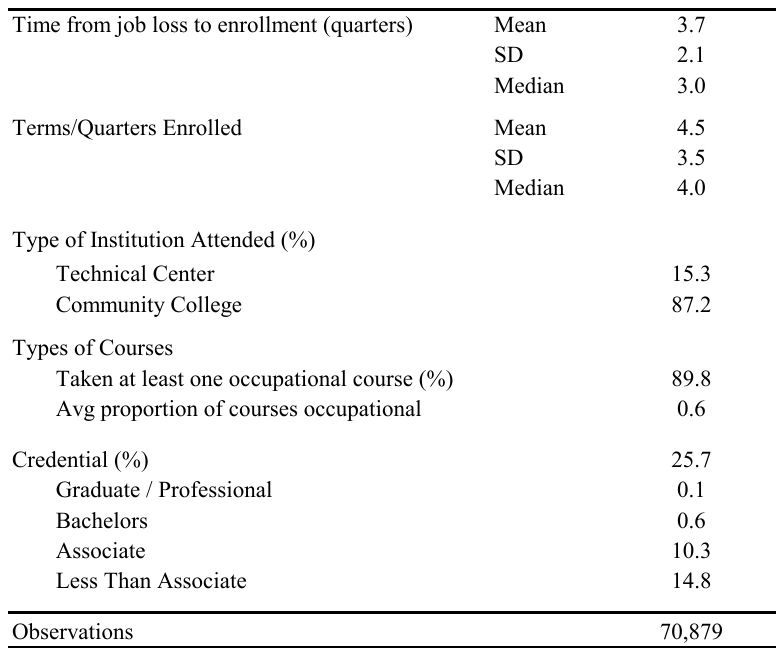}
\centering{}%
\begin{minipage}[t]{0.75\columnwidth}%
\vspace{42bp}
Notes: Type of Institution Attended, Terms/Quarters Enrolled, Types of Courses, and Credential are calculated within four years of first enrollment. Enrollees may attend more than one type of institution over the four-year period. ``Less than Associate'' credentials include less than two-year awards from HEI and any credential from OTC. ``SD'' is standard deviation.
\end{minipage}
\end{table}

\newpage{}
\begin{table}[H]
\caption{Enrollment Effect Estimates}
\label{tab:att_estimates}\medskip{}

\includegraphics[viewport=-10bp 500bp 612bp 740bp]{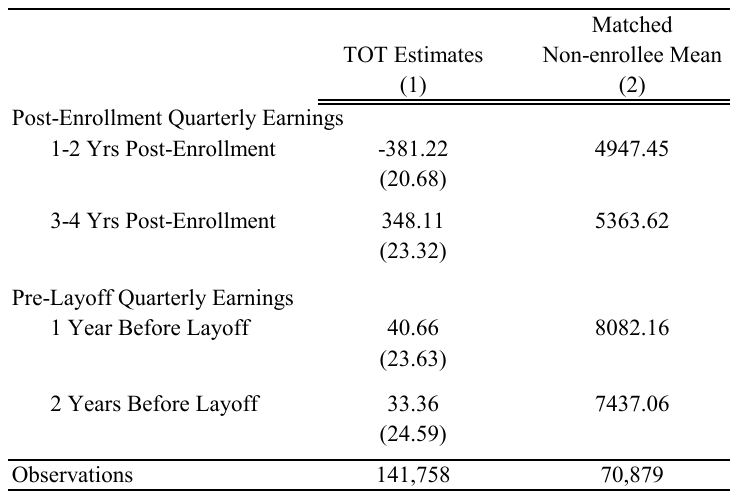}
\centering{}%
\begin{minipage}[t]{0.72\columnwidth}%
Notes: Column (1) shows the estimated effect of enrollment on earnings in the time period denoted by the row headings. The post-enrollment estimates are lower bounds on the TOT effects (Section \ref{subsec:Identification}); for brevity, we refer to them and the analogous estimates in Table \ref{tab:subgroup_ests} as enrollment effects or TOT estimates. Column (2) shows the mean earnings in the matched non-enrollee sample. \citet{Abadie_Imbens2016} standard errors are in parentheses.
\end{minipage}
\end{table}
\newpage\vspace{-10bp}
\begin{sidewaystable}[H]
\caption{Enrollment Effect Estimates By Subgroup}
\label{tab:subgroup_ests}\vspace{-45bp}

\begin{centering}
\includegraphics[viewport=0bp 68bp 792bp 617bp,scale=0.85]{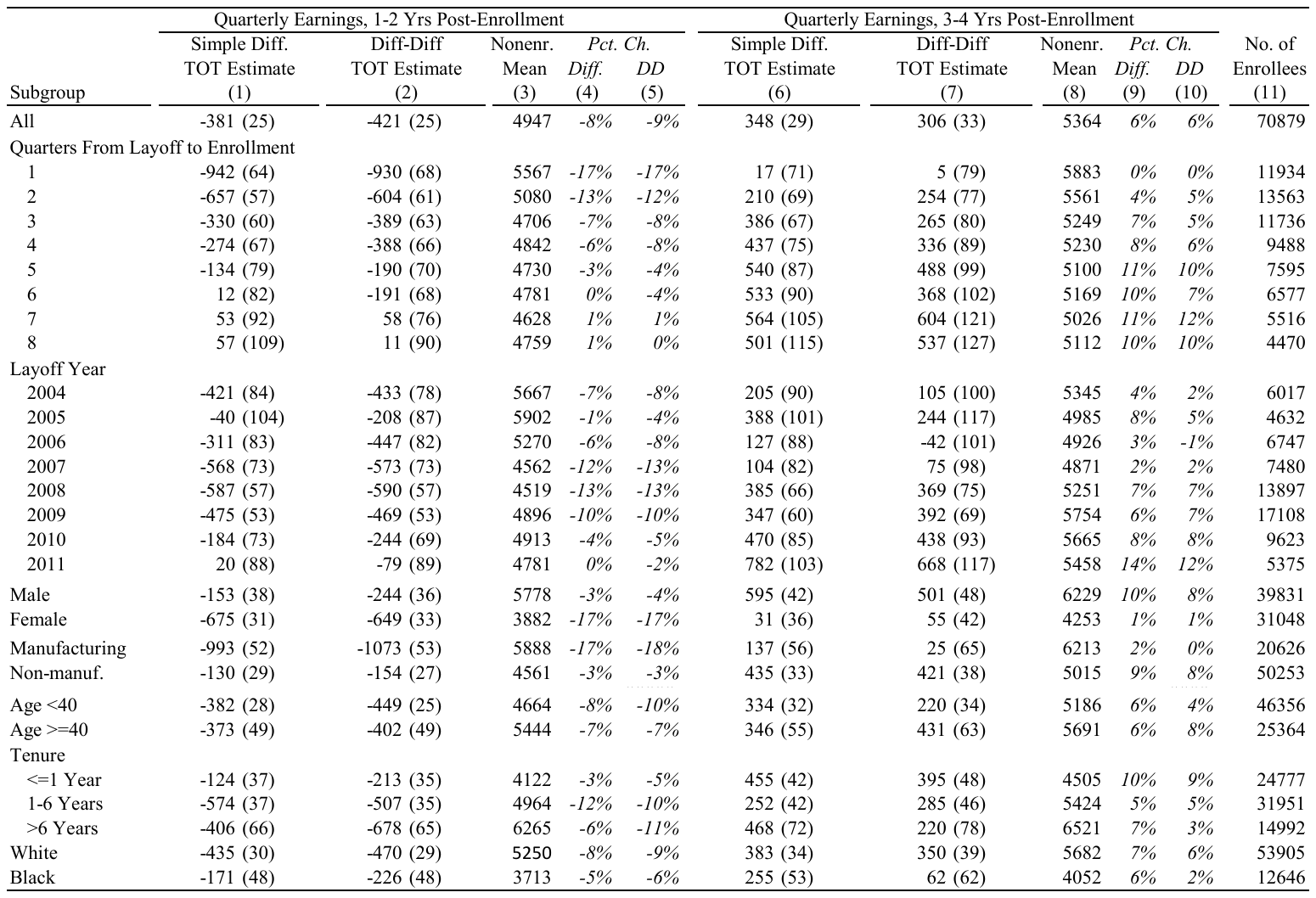}\vspace{5bp}
\par\end{centering}
\noindent\begin{minipage}[t]{1\columnwidth}%
Notes: Columns (1) and (2) show the estimated effect of enrollment on quarterly earnings in the two years after first enrollment, computed by taking the difference between enrollees and matched non-enrollees, and by matched difference-in-differences, respectively, for each subgroup denoted by the row headings. Columns (6) and (7) show the analogous results for the third and fourth years after enrolling. Columns (3) and (8) show the mean quarterly earnings of the matched non-enrollees. Columns (4), (5), (9), and (10) express the columns (1), (2), (6), and (7) as percentages of the matched non-enrollee mean. Standard errors for the mean pairwise difference between enrollees and matched non-enrollees are reported in parentheses.
\end{minipage}
\end{sidewaystable}

\newpage\pagenumbering{arabic}

\appendix
\begin{center}
\textbf{\Large{}Appendix (For Online Publication Only)}{\Large\par}
\par\end{center}

\section{Data Sources and Sample Construction: Additional Details\label{sec:Data Appendix}}

\renewcommand{\theequation}{A\arabic{equation}}

\setcounter{equation}{0}

\setcounter{footnote}{0}

\counterwithin{table}{section} 
\counterwithin{figure}{section}

\subsection{Construction of Analysis Sample and Variables\label{subsec:analysis_sample}}

Our unemployment insurance (UI) claims sample consists of eligible, regular UI claims (i.e., claims that are filed when a worker is first unemployed). To create the sample, we first construct UI spells by grouping together regular claims and all associated extension claims (e.g., benefits under the Extended Benefit, Emergency Unemployment Compensation, and Trade Readjustment Allowance programs), which have the same benefit year beginning date. We then check for cases where an individual has two overlapping UI spells (i.e., one spell begins before the previous one ends), which we further group together as part of the same UI spell. These steps help ensure that the beginning of each spell corresponds to when a worker is first laid off, rather than a continuation of an ongoing unemployment spell. We drop observations that have missing gender, age under 16 or over 100, and non-Ohio zip codes. For enrollees in our sample, we only use the UI spell that is closest to (but before) enrollment. Non-enrollees are allowed to have multiple UI spells in the sample.

Several of our key variables come from the quarterly wage database that is part of the Ohio UI system. Our main outcome measures are quarterly earnings. Although the earnings data we receive are topcoded, the censoring points are sufficiently high so as to be not very relevant for our sample, as is evident from the summary statistics of Table \ref{tab:descriptives_matched}: prior to 2009, earnings above \$99,999 were censored at \$99,999; for all quarters starting in 2009, the top one percent of earnings have been topcoded to the average of the top one percent. The wage data also contain employer pseudo-IDs and industry, which allow us to construct job tenure measures. Since these employer pseudo-IDs are not in the UI claims data, we first search in the wage data for an employer that matches the industry reported in the UI claims data in the five quarters before layoff, starting with the most recent quarter. Once we have identified the employer that matches the industry in the claims data, we define tenure as the number of quarters since the first time a worker is observed to have worked for the employer in the wage data. We then create three tenure categories: less than one year, one to six years, and more than six years (since we only observe wages starting in 1995Q2, the maximum tenure for the earliest claims in the sample is 8.75 years). If we do not find an employer with a matching industry or if the industry is missing in the claims data, we report the tenure at the most recent employer in the wage data. Only a small fraction of the wage data have missing industry (0.5 percent).

Since we only observe earnings within Ohio, it is possible that we are understating the earnings of workers who work out-of-state. To gauge the extent of this issue, we examine the out-migration rates of workers who claim UI in the Survey of Income and Program Participation (SIPP). In the 2004 (2008) SIPP panels, there were 139 (261) individuals who reported receiving UI benefits while residing in Ohio at some point during the four years of the panel. Of these individuals, 8 (4) individuals moved out-of-state, and 3 (2) found jobs in their new destination in the 2004 (2008) panels. Although the samples of UI recipients are small in the SIPP, these rates of out-migration are similar to the overall yearly migration rates for Ohio found using IRS Statistics of Income data (2005-2016) and the American Community Survey (2010-2016), which are about two percent. To get a sense for how out-migration may affect our main estimates, consider the extreme scenario where two percent of the non-enrollee control group have no in-state earnings due to out-migration and that \emph{none} of the enrollee treatment group has migrated. Then replacing the zero earnings of the bottom two percent of the control group with their mean earnings will increase the average control earnings by about \$107 (two percent times \$5,364 as reported in Table \ref{tab:att_estimates}) three to four years post-enrollment, resulting in a treatment effect estimate of \$241 per quarter. Even in this very conservative scenario, there is a positive enrollment effect of more than four percent.

\subsection{Enrollment, Course, and Credential Data\label{subsec:courses_appendix}}

Our enrollment data come from two sources: 1) the Higher Education Information (HEI) data, which cover two- and four-year public colleges, and 2) the Ohio Technical Center (OTC) data, which contain information on technical centers (the OTC data also include training that takes place in correctional facilities and high schools, though these account for less than one percent of our enrollee sample). The HEI data, available from the summer of 1999 to the spring of 2017, are all reported at the person-institution-term level so we can observe enrollment, courses taken, and credentials/degrees obtained in every term for each individual. Terms are ``Winter'', ``Spring'', ``Summer'', and ``Autumn'', which we map to the first, second, third, and fourth quarters of the calendar year in our analysis. Starting in 2013, ``Winter'' terms were eliminated as all public colleges moved to a semester system. The OTC data, spanning years between 2002 and 2017, contain start and end dates for each course, and credentials are often associated with specific courses.

Each course and credential has a Classification of Instructional Programs (CIP) code that denotes the subject area. In the HEI data, 0.3 percent of courses taken by enrollees in our sample are missing CIP codes (none of the credentials are missing CIP codes). In the OTC data, although we see that 10 percent of courses in our enrollee sample have missing CIP codes, there is another variable containing an internal subject code that we can use to fill in most of the missing CIP codes. That is, for courses that have missing CIP codes but non-missing internal subject codes, we fill in the most common CIP code associated with that internal subject code. This procedure effectively reduces the percentage of OTC courses with missing CIP codes to 0.7 percent. Also, the HEI course information for 2006 appears to be incomplete relative to the enrollment data: 17 percent of the enrollee-by-term observations in 2006 are not associated with courses. This missing data issue only affects results in Section \ref{subsec:What-Does-Schooling} and may slightly undercount the number of enrollees that take coursework related to specific industries or areas.

To categorize courses and credentials to specific (two-digit) industries, we create a mapping of CIP codes to industries using the data from the National Center for Education Statistics (NCES) and the Bureau of Labor Statistics (BLS). The NCES provides a list of occupations that each CIP subject prepares for, while the National Employment Matrix (2018) from the BLS has data on the share of workers of a particular occupation in a specific industry (\citealp{NCES_CIPSOC_app}; \citealp{BLS_NEM_app}). To create this mapping, we consider only ``occupational'' (as opposed to ``academic'') CIP codes, as defined by the NCES. We also remove ``postsecondary teacher'' occupations (SOC 25-1000) because many subjects list educators of that subject as a possible occupation, and these tend to be categorized as postsecondary teachers, most of whom work in the Educational Services sector. Since we believe most workers who train in a particular subject are not aiming to teach in the area, we do not want to falsely attribute a subject to the education sector. We then order the occupations within CIP subject by the total employment of the occupation according to the National Employment Matrix, and map each CIP code to the most populous occupation. We further map a CIP code to a two-digit NAICS industry if, for the occupation the CIP maps to, more than a quarter of the employment is within that industry. We use the largest industry if there are multiple industries that meet this criterion. We designate a CIP code as not associated with an industry if 1) the largest industry accounts for less than 25 percent of the occupation the CIP maps to, 2) the CIP is not associated with any occupation, or 3) the mapped ``industry'' is self-employment (i.e., more workers in the mapped occupation are self-employed than working in any particular sector). An example of a course not mapped to an industry is Human Resources Management (CIP 52.1001): although it prepares a worker for the Human Resource Specialist occupation, human resource specialists work in many different industries and cannot be assigned to a specific one.

Courses and credentials in the OTC data are mapped to industries via their CIP code. When a worker takes more than one course, we use the course with the most course hours (and randomly pick a course with the same largest number of course hours). In the HEI data, where each enrollee typically takes more than one course, we first map each of her courses to an industry and then assign to her coursework the modal industry over all her courses (for this assignment, we require that she take at least three courses in that industry). When we observe a worker earn more than one credential, we map her credentials to the modal (non-missing) industry over all her credentials, and in the case of ties, we keep the industry mapped from the highest credential. When a worker is observed in both the OTC and HEI data, and the industries of her coursework do not match, we use the one in the OTC data (because courses in OTC are more easily matched to an industry as there are fewer courses); if the industries of the credential do not match, we use the one in the HEI data (because institutions in the HEI data tend to confer higher credentials).

Unfortunately, we do not observe enrollment in private institutions in our data. As we noted in Section \ref{sec:Data}, we believe unobserved private enrollment is more likely in the control group and disproportionately at for-profit institutions, which likely will lead us to understate the effects of training, since \citet{CelliniTurner2019_app} find positive (though statistically insignificant) effects of for-profit colleges. 

It is hard to know the share of our matched control group that enrolled in private institutions. To approximate this share, we turn to the October Current Population Survey (CPS), which contains detailed information on schooling. In the October CPS, we restrict to the years 2004-2011, those who are aged 25-39, and those enrolled in a two-year institution or ``taking courses.'' We also restrict to those for whom industry and occupation information is not missing (it is missing only for those who are in the labor force and have never worked, or out of the labor force and have not worked within the past year). Unfortunately, while institution ownership information is available for those attending ``regular'' school, it is missing for a large fraction of the students who are just ``taking courses'': the institution ownership status is missing for none of the 3,550 two-year college enrollees, but for 81 percent of the 3,917 who are ``taking courses'' but not in ``regular'' school. We observe that about 10 percent of enrollees attend private institutions, 48 percent public institutions, and 43 percent institutions of unknown ownership.

If we assume that all of the ``unknown ownership'' group is private, that will mean that roughly a half (53 percent) of enrollees attend private institutions. Since we observe that about 4 percent of UI claimants enroll in public institutions, this will imply that another 4 percent or so enroll in private institutions. It is unclear, though, how many of our matched controls actually enroll in private institutions. At one extreme, if we assume that ALL of the matched controls enroll in private institutions, and that the effect of private enrollment is 21 percent lower than public institutions (from \citealp{CelliniTurner2019_app}), this will imply that the effect of public enrollment relative to not enrolling at all is $0.06/0.21=29$ percent (using our main estimate of six percent). If we assume instead that the matched controls have the same proportion of private enrollees as the general population (four percent), the effect of public enrollment relative to not enrolling will be around $0.06/(1-0.04*0.79)=6.2$ percent, which is quite close to our overall estimate. 

Finally, we note that in the CPS, private enrollees differ from public enrollees in terms of past labor market characteristics. As shown in Table \ref{tab:cps_private_enrollment}, we find that private enrollees are less likely to come from the retail industry and more likely to come from finance, insurance, and real estate, as well as professional services (which includes health professionals). These estimates suggest that some of our subgroup estimates (i.e., by past industry) may be subject to more substitution bias from failing to observe private enrollments. Notably, the estimates for those who were previously in healthcare may suffer from more of a downward bias than, say, those who were previously in manufacturing, though, again, it is hard to speak about magnitudes.

\subsection{Workforce Investment Act Data\label{subsec:WIA_data_appendix}}

To analyze the enrollment effects for workers who trained under the Workforce Investment Act (WIA) program, our analysis sample has been merged with the WIA administrative program data (from the WIA Standardized Record Data system). The data contain quarterly snapshots of WIA participants and exiters between 2006Q1 and 2015Q4. The earliest snapshot (2006Q1) contains participants who exited the program starting in 2004Q1. The analysis sample in Section \ref{subsec:WIA} contains only WIA participants who received job training from WIA. We observe the month of WIA registration, the beginning and end months of WIA training, program funding stream (e.g., adult, dislocated worker, or youth), and type of training (e.g., on-the-job, skill upgrading, entrepreneurial skills). We say that an enrollee from our main analysis sample receives WIA training if we observe her starting WIA training within 24 months after the UI claim date. We only observe WIA participation for workers in our main analysis sample of UI claims.

\section{Identification Results: Details, Proofs, and Generalizations\label{sec:identification_appendix}}

In this section, we provide additional details to Sections \ref{subsec:Identification} and \ref{subsec:Relations-to-DTE-lit}. In Section \ref{subsec:Discussion-of-Assumption2}, we connect Assumption \ref{assu:S2_potential_outcome} to models of training decisions by \citet{Heckman1978_app}, \citet{AshenfelterandCard1985_app}, and \citet{HeckmanandRobb1985JOE_app,HeckmanandRobb1985book_app}. We supply proofs and expand on the partial identification results in Section \ref{subsec:Details-on-Bounds}. Section \ref{subsec:S-period-Generalization} generalizes our assumptions and identification results in Section \ref{subsec:Identification} from two periods to an arbitrary number of periods, $S$. In Section \ref{subsec:DTE_lit_details}, we fill in the details of Section \ref{subsec:Relations-to-DTE-lit} on Robins (Section \ref{subsec:Robins_details}), Lechner (Section \ref{subsec:Lechner_IPW}), and the literature modeling unobserved heterogeneity (Section \ref{subsec:unobserved_heterogeneity}).

\subsection{Discussion of Assumption \ref{assu:S2_potential_outcome}\label{subsec:Discussion-of-Assumption2}}

In the two-period case, our Assumption \ref{assu:S2_potential_outcome} states that 
\begin{equation}
E[Y_t(0)|D_1=0,D_2=1,\mathbf{X}^{1}]\leqslant E[Y_t(0)|D_1=0,D_2=0,\mathbf{X}^{1}]\label{eq:AC_2_period}
\end{equation}
for $t \geqslant 1$. That is, among observationally equivalent workers at the beginning of period 1, those who start training in period 2 have weakly lower future potential earnings, $Y_t(0)$, than those who never enroll. This selection condition is the key ingredient behind our lower bound result for period-1 enrollees (and thus the overall TOT) in Proposition \ref{prop:Prop_S2_lb}. 

Given its importance, we discuss Assumption \ref{assu:S2_potential_outcome} further in this section. In particular, we show that, under two mild additional conditions, Assumption \ref{assu:S2_potential_outcome} is implied by dynamic extensions of classic models of training participation by \citet{Heckman1978_app}, \citet{AshenfelterandCard1985_app}, and \citet{HeckmanandRobb1985JOE_app,HeckmanandRobb1985book_app}. These models posit that a worker decides to retrain if and only if she receives a low earnings signal. More specifically, for the static one-shot training decision these papers consider, a worker decides to retrain if and only if her earnings $k$ periods prior ($k \geqslant 0$) plus an idiosyncratic error term fall below a fixed threshold. \citet{AshenfelterandCard1985_app} evaluate training effects by considering different nonnegative values of $k$, while the papers (coauthored) by Heckman focus on the case of $k=0$, deriving the training decision rule from an economic model of a worker maximizing expected earnings.\footnote{Specifically, the Heckman papers show that a risk-neutral worker decides to train if the present value of the earnings gain exceeds the sum of the direct training cost and the opportunity cost of foregone earnings during training. The present value of the earnings gain is captured by the fixed threshold that does not vary across individuals: workers may not know their potential individual treatment effect \emph{ex ante}, but base their training decision on the population average effect instead (\citealp[181]{HeckmanandRobb1985book_app}). Note that this assumption of a constant expected gain is different from and more plausible than assuming a constant treatment effect---see \citet{Heckmanetal1999_HoLEv3_app} and \citet{Smith2022_app} for empirical and methodological reviews of treatment effect heterogeneity.}

We extend the logic to our dynamic setting where training can commence at more than one point in time. Formally, at the beginning of period $s$, the decision rule for a previously unenrolled worker is
\begin{equation}
D_s=1 \Longleftrightarrow Y_{s-k}(0)+\upsilon_s<\bar{y}_s\label{eq:training_decision_rule}
\end{equation}
for some $k\geqslant0$, where $\bar{y}_s$ is a constant and $\upsilon_s$ a random variable.\footnote{When $k=0$, the decision rule involves a potential earnings variable $Y_s(0)$ that is not yet realized. We follow both \citet{AshenfelterandCard1985_app} and the Heckman papers by assuming perfect foresight.}  \citet{AshenfelterandCard1985_app} assume this random variable to be completely idiosyncratic in the sense that it is independent to any component of the earnings process; in our setting, this translates to the shock being independent to the joint distribution of potential outcomes in any period. 

The dynamic setting also requires us to take a stand on the shock distribution over time. For simplicity of exposition, we impose serial independence of the shocks (it can be relaxed as we discuss below). Together, these restrictions form our first additional condition:
\begin{equation}
\{Y_{t}(0)\}_{t\in\mathbb{Z}}\independent \{\upsilon_s\}_{s\in\mathbb{N}}, \text{ and the shocks } \{\upsilon_s\} \text{ are independent across time and individuals.}\label{eq:augmented_independence_upsilon}
\end{equation}

\begin{remark}
\label{rem:training_cost}
\normalfont
\citet[244]{HeckmanandRobb1985JOE_app} interpret the $\upsilon$ random variable as the \emph{negative} of the worker's training subsidy (or equivalently, $\upsilon$ reflects training costs) and allow for nonzero correlation between this random variable and potential earnings through observables. Their correlation is negative, for example, if high ability individuals have higher potential earnings and incur a lower training cost. Training cost could be lower for these workers because training takes less time or because they are better at finding training subsidies. Additionally,  serial independence in Assumption (\ref{eq:augmented_independence_upsilon}) seems at odds with the subsidy/cost interpretation, under which the $\upsilon_s$'s are likely to be positively correlated over time. As we discuss below, inequality (\ref{eq:AC_2_period}) still holds under models that relax (\ref{eq:augmented_independence_upsilon}).
\end{remark}

Before we introduce our second additional condition, we need to be more explicit about what the conditioning set in (\ref{eq:AC_2_period}) encapsulates. Without loss of generality, we think of $\mathbf{X}^{s}$ as a vector of past earnings leading up to $D_{s}$, for $s=1,2$. Specifically, the conditioning set $\mathbf{X}^{1}$ in (\ref{eq:AC_2_period}) consists of earnings up to $K$ periods before layoff: $\mathbf{X}^{1}=(Y_{-K},Y_{-(K-1)},\ldots,Y_{0})$. While we also include in our matching specification other demographic and labor market indicators, such as gender, race, age, industry, and layoff timing, we can think of inequality (\ref{eq:AC_2_period}) as holding within each cell defined by these variables, so that the inequality also holds when these variables are incorporated into $\mathbf{X}^{1}$.

Our second additional condition is no-anticipation: for $s\geqslant 1$,
\begin{equation}
\mathbf{X}^{s}=\{Y_{-K}(0),Y_{-(K-1)}(0),...,Y_{0}(0),...,Y_{s-1}(0)\};\text{}Y_{1}=Y_{1}(0),...,Y_{s}=Y_{s}(0)\text{ when }D_{1}=...=D_{s}=0.\label{eq:no-ant}    
\end{equation}
Relevant to the two-period case, setting $s=1$ in (\ref{eq:no-ant}) implies that $\mathbf{X}^{1}=\{Y_{-K}(0),Y_{-(K-1)}(0),\ldots,Y_{0}(0)\}$ and $Y_{1}=Y_{1}(0)$ when $D_{1}=0$. We show in Appendix \ref{subsec:Robins_details} below that this is implied by the typical no-anticipation assumption from the dynamic treatment effect literature (see \citealp{AbbringandHeckman2007_app} for a nice exposition and discussion on the no-anticipation assumption), which puts restrictions on potential outcomes under different treatment sequences. In both parts of (\ref{eq:no-ant}), the potential outcomes for a not-yet enrolled worker are not affected by whether she would pursue training subsequently. As a result, the first part of (\ref{eq:no-ant}) implies that the baseline covariate vector $\mathbf{X}^{1}=(Y_{-K},Y_{-(K-1)},\ldots,Y_{0})$ coincides with the vector of potential outcomes under no-treatment. And the second part implies that period-1 outcomes among period-1 non-enrollees do not depend on their potential treatment status in period 2. As we argue below, the no-anticipation assumption is implicitly maintained in the classic papers referenced above; we simply make it explicit here. 

Finally, we borrow the specification of the earnings process from \citet{AshenfelterandCard1985_app}: for individual $i$ in time $t$, \citet{AshenfelterandCard1985_app} specify her earnings as the sum of a random effect, a time fixed effect, a treatment effect if she is a trainee, and a transitory error term.\footnote{Note this specification implies no-anticipation: the potential outcome $Y_{it}(0)$ is the same regardless of whether the worker would enroll in period $t+1$, and it is equal to the actual outcome for a worker not yet enrolled as of period $t$.} Under this specification, the potential earnings process is $Y_{it}(0)=\omega_i+\lambda_t+\epsilon_{it}$, where $\omega_i$ is the random effect, $\lambda_t$ the (deterministic economy-wide) time fixed effect, and $\epsilon_{it}$ the transitory error term.\footnote{For clarity, we include the $i$ subscript in this paragraph. We will drop it again for simplicity after the assumption statements (\ref{eq:income_process}) and (\ref{eq:shock_process}), and all expectations are understood to be taken across $i$.} We follow \citet{AshenfelterandCard1985_app} and specify $\epsilon_{it}$ as an AR(1) process with coefficient $\rho$ ($\rho\geqslant0$; they estimate $\rho$ to be between 0.73 and 0.80). Lastly, both \citet{AshenfelterandCard1985_app} and the Heckman papers assume that the conditional expectation of potential earnings is linear in the training selection variable, which would be the case if the $Y_{it}(0)$ process and the $\upsilon$'s are joint normal. We impose this joint normality here. Putting all the pieces together, we make the following functional form assumptions by following \citet{AshenfelterandCard1985_app}: 
(a) For each individual $i$, the income process $\{Y_{it}(0)\}_{t\in\mathbb{Z}}$ is joint normal, and
\begin{equation}
Y_{it}(0)=\omega_{i}+\lambda_{t}+\epsilon_{it} \text{ with } \epsilon_{it}=\rho\epsilon_{i(t-1)}+\eta_{it}\text{ and }\rho\in[0,1),\label{eq:income_process}
\end{equation}
where $\omega_i$ is i.i.d. across $i$ with $var(\omega_{i})\equiv\sigma_{\omega}^{2}$, $\eta_{it}$ is i.i.d. across $i$ and $t$ with $var(\eta_{it})\equiv\sigma_{\eta}^{2}$, and, as a result, $var(\epsilon_{it})\equiv\sigma_{\epsilon}^{2}=\sigma_{\eta}^{2}/(1-\rho^{2})$.\footnote{Heckman's specification of the potential earnings process nests that of \citet{AshenfelterandCard1985_app}. For individual $i$ in time $t$, her potential earnings are $Y_{it}(0)=\mathbf{W}_{it}\mathbf{\varsigma}+U_{it}$, where $\mathbf{W}$ is a vector of observables. In \citet[181]{HeckmanandRobb1985book_app}, the error term is further specified as $U_{it}=\omega_{i}+\epsilon_{it}$, where $\epsilon_{it}$ follows an AR(1) process. Thus, the Heckman specification becomes (\ref{eq:income_process}) when $\mathbf{W}_{it}$ consists of the time fixed effect, a case we focus on here.}
(b) The idiosyncratic shock
\begin{equation}
\upsilon_{is} \text{ is normally distributed } \forall i,s.\label{eq:shock_process}
\end{equation}
\begin{remark}
\label{rem:income_positive_dependence}
\normalfont
The specification of the income process is important for our results in this section because it implies that a low income signal in one period would predict low potential incomes in other periods. While our specification gives rise to mean reversion---a critical feature and challenge in the early studies of training evaluation and difference-in-differences methods---it precludes ``regression past the mean.'' That is, it does not allow the possibility that a low income signal today would predict high potential incomes in the future.
\end{remark}
\noindent We now state the main proposition of this section:
\begin{prop}
\label{prop:AC_2_period}
The training decision rule (\ref{eq:training_decision_rule}), Assumptions (\ref{eq:augmented_independence_upsilon}) and (\ref{eq:no-ant}), and specifications of the income and shock processes (\ref{eq:income_process}) and (\ref{eq:shock_process}) imply inequality (\ref{eq:AC_2_period}).
\end{prop}
\noindent To prove Proposition \ref{prop:AC_2_period}, we first establish several results in the form of one lemma and two corollaries.
\begin{lem}
\label{lem:MTP2}Given the training decision rule (\ref{eq:training_decision_rule}) with income and shock
processes (\ref{eq:income_process}) and (\ref{eq:shock_process}),
Assumptions (\ref{eq:augmented_independence_upsilon}) and (\ref{eq:no-ant}) imply that

\noindent (a) for any $t\geqslant1$, the ($K+t+1$)-dimensional random vector
$\mathbf{Y}_{t}\equiv(\mathbf{X}^{1},Y_{1}(0),...,Y_{t-1}(0),Y_{t}(0))$
is Multivariate Totally Positive of Order 2 ($\text{MTP}_{\text{2}}$);

\noindent (b) for $k=0,1$, the conditional expectation function
\[
E[Y_{2-k}(0)|Z_{1}\geqslant\bar{y}_1,Z_{2},\mathbf{X}^{1}]
\]
is increasing in $Z_{2}$ where $Z_{1}\equiv Y_{1-k}(0)+\upsilon_{1}$
and $Z_{2}\equiv Y_{2-k}(0)+\upsilon_{2}$.
\end{lem}
\begin{proof}
(a) First, under (\ref{eq:no-ant}) and (\ref{eq:income_process}),
$\mathbf{Y}_{t}$
is multivariate normal. Per \citet{KarlinandRinott1980a_app}, $\mathbf{Y}_{t}$
is $\text{MTP}_{\text{2}}$ if and only if the inverse of its variance-covariance
matrix has nonpositive off-diagonal entries. By (\ref{eq:income_process}),
\[
cov(\mathbf{Y}_{t})\equiv\Sigma_{\mathbf{Y}}=\sigma_{\omega}^{2}\mathbf{1}\mathbf{1}^{\prime}+\Sigma_{\boldsymbol{\epsilon}}
\]
where $\mathbf{1}$ is a ($K+t+1$)-dimensional vector of ones
and $\Sigma_{\boldsymbol{\epsilon}}$ is the variance-covariance matrix
of the $\epsilon$'s, of which the $jk$-th entry is $(\Sigma_{\boldsymbol{\epsilon}})_{jk}=\sigma_{\eta}^{2}/(1-\rho^{2})\rho^{|j-k|}$.
By the Sherman-Morrison formula (\citealp{ShermanandMorrison1950_app};
\citealp{Bartlett1951_app}),
\[
\Sigma_{\mathbf{Y}}^{-1}=\Sigma_{\boldsymbol{\epsilon}}^{-1}-\frac{\sigma_{\omega}^2\Sigma_{\boldsymbol{\epsilon}}^{-1}\mathbf{1}\mathbf{1}^{\prime}\Sigma_{\boldsymbol{\epsilon}}^{-1}}{1+\sigma_{\omega}^2\mathbf{1}^{\prime}\Sigma_{\boldsymbol{\epsilon}}^{-1}\mathbf{1}}.
\]
It is a textbook result (e.g., \citealp[217]{Lindsey2004_app}) that
$\Sigma_{\boldsymbol{\epsilon}}^{-1}$ is band-diagonal
\[
\Sigma_{\boldsymbol{\epsilon}}^{-1}=\frac{1}{\sigma_{\eta}^{2}}\left(\begin{array}{ccccc}
1 & -\rho & \ldots & 0 & 0\\
-\rho & 1+\rho^{2} & -\rho & 0 & 0\\
\vdots & \vdots & \ddots & \vdots & \vdots\\
0 & 0 & \ldots & 1+\rho^{2} & -\rho\\
0 & 0 & \ldots & -\rho & 1
\end{array}\right)
\]
where the first minor diagonal entries are $-\rho$ and all other
off-diagonal entries are 0. Notice that $\Sigma_{\boldsymbol{\epsilon}}^{-1}\mathbf{1}$
is a vector of row sums of $\Sigma_{\boldsymbol{\epsilon}}^{-1}$ which
are all positive. Therefore, all entries in the numerator $\sigma_{\omega}^2\Sigma_{\boldsymbol{\epsilon}}^{-1}\mathbf{1}\mathbf{1}^{\prime}\Sigma_{\boldsymbol{\epsilon}}^{-1}$
are positive, as is the scalar $1+\sigma_{\omega}^2\mathbf{1}^{\prime}\Sigma_{\boldsymbol{\epsilon}}^{-1}\mathbf{1}$
in the denominator. Since all off-diagonal entries of $\Sigma_{\boldsymbol{\epsilon}}^{-1}$
are nonpositive, all off-diagonal entries of $\Sigma_{\mathbf{Y}}^{-1}$
are nonpositive. Hence, $\mathbf{Y}_{t}$ is $\text{MTP}_{\text{2}}$.

(b) We focus on the case of $k=1$, as the proof for $k=0$ is analogous. First, we establish that the random vector $(Y_{1}(0),Z_{1},Z_{2},\mathbf{X}^{1})$
is $\text{MTP}_{\text{2}}$. By Assumptions (\ref{eq:augmented_independence_upsilon}) and (\ref{eq:income_process}),
the joint density of $(Y_{1}(0),Z_{1},Z_{2},\mathbf{X}^{1})$ is 
\[
f_{(Y_{1}(0),Z_{1},Z_{2},\mathbf{X}^{1})}(y_{1},z_{1},z_{2},y_{-K},...,y_{-1},y_{0})=f_{(Y_{1}(0),\mathbf{X}^{1})}(y_{1},y_{-K},...,y_{-1},y_{0})\cdot f_{\upsilon_{1}}(z_{1}-y_{0})\cdot f_{\upsilon_{2}}(z_{2}-y_{1}).
\]
Property (1.16) of \citet{KarlinandRinott1980a_app} together with part
(a) of this lemma implies that $f_{(Y_{1}(0),\mathbf{X}^{1})}$ is
$\text{MTP}_{\text{2}}$. By (\ref{eq:shock_process}), $\upsilon_1$ and $\upsilon_2$ are normally distributed. Since the normal density function is log concave, analogous to the proof of Proposition 3.7
of \citet{KarlinandRinott1980a_app}, $f_{\upsilon_{1}}(z_{1}-y_{0})$ and $f_{\upsilon_{2}}(z_{2}-y_{1})$
are $\text{MTP}_{\text{2}}$.\footnote{\citet{KarlinandRinott1980a_app} use the term Polya frequency
function of order 2, which is equivalent
to log concavity.} By Proposition 3.3 of \citet{KarlinandRinott1980a_app},
$(Y_{1}(0),Z_{1},Z_{2},\mathbf{X}^{1})$ is $\text{MTP}_{\text{2}}$.
By Proposition 2.5(a) of \citet{RinottandScarsini2006_app}, the conditional
distribution of $(Y_{1}(0),Z_{1},Z_{2},\mathbf{X}^{1})|Z_{1}\geqslant\bar{y}_1$
is also $\text{MTP}_{\text{2}}$ since $Z_{1}\geqslant\bar{y}_1$ defines
a rectangular set that satisfies the condition of the proposition.
By property (1.16) of \citet{KarlinandRinott1980a_app}, $(Y_{1}(0),Z_{2},\mathbf{X}^{1})|Z_{1}\geqslant\bar{y}_1$
is $\text{MTP}_{\text{2}}$, and our desired result follows from Theorem
4.1 of \citet{KarlinandRinott1980a_app}.
\end{proof}
\begin{cor}
\label{cor:subvec_MTP2}Under Assumptions (\ref{eq:augmented_independence_upsilon}) and (\ref{eq:no-ant}) and given the income and shock processes (\ref{eq:income_process})
and (\ref{eq:shock_process}), the random vector $(\mathbf{X}^{1},Y_{1}(0),Y_{t}(0))$
is $\text{MTP}_{\text{2}}$ for $t\geqslant2$, and $(\mathbf{X}^{1},Y_{1}(0),Y_{2}(0),Y_{t}(0))$ is 
$\text{MTP}_{\text{2}}$ for $t\geqslant3$.
\end{cor}
\begin{proof}
The corollary follows from Lemma \ref{lem:MTP2} and Property (1.16)
of \citet{KarlinandRinott1980a_app}, which states that any subvector
of an $\text{MTP}_{\text{2}}$ random vector is also $\text{MTP}_{\text{2}}$.
\end{proof}
\begin{cor}
\label{cor:AC_2_period}Under Assumptions (\ref{eq:augmented_independence_upsilon}) and (\ref{eq:no-ant}) and given the income and shock processes (\ref{eq:income_process})
and (\ref{eq:shock_process}),
\[
E[Y_{1}(0)|Y_{0}(0)+\upsilon_{1}\geqslant\bar{y}_1,Y_{1}(0)+\upsilon_{2}<\bar{y}_2,\mathbf{X}^{1}]\leqslant E[Y_{1}(0)|Y_{0}(0)+\upsilon_{1}\geqslant\bar{y}_1,Y_{1}(0)+\upsilon_{2}\geqslant\bar{y}_2,\mathbf{X}^{1}].
\]
\end{cor}
\begin{proof}
Corollary \ref{cor:AC_2_period} is a simple consequence of Lemma
\ref{lem:MTP2}(b).
\end{proof}
\noindent Now we go on to prove Proposition \ref{prop:AC_2_period}:
\begin{proof}
We consider two separate cases, depending on whether the income variables that help decide training participation are part of the conditioning vector $\mathbf{X}^{1}$. First case: $k\in[2,K+1]$, which implies the income variables $Y_{1-k}$ and $Y_{2-k}$ are elements of $\mathbf{X}^{1}$.
Using the $Z$ notation from Lemma \ref{lem:MTP2}: for $t\geqslant 1$,
\begin{align*}
 & E[Y_{t}(0)|D_{1}=0,D_{2}=1,\mathbf{X}^{1}]\\
= & E[Y_{t}(0)|Z_1\geqslant \bar{y}_1,Z_2<\bar{y}_2,\mathbf{X}^{1}]\\
= & E[E[Y_{t}(0)|\upsilon_1,\upsilon_2,\mathbf{X}^1]|Z_1\geqslant \bar{y}_1,Z_2<\bar{y}_2,\mathbf{X}^{1}]\\
= & \xi_{t}+\mathbf{X}^{1} \boldsymbol{\zeta}_t\\
= & E[E[Y_{t}(0)|\upsilon_1,\upsilon_2,\mathbf{X}^1]|Z_1\geqslant \bar{y}_1,Z_2\geqslant \bar{y}_2,\mathbf{X}^{1}]\\
= & E[Y_{t}(0)|D_{1}=0,D_{2}=0,\mathbf{X}^{1}]
\end{align*}
where the first equality is given by the training decision rule (\ref{eq:training_decision_rule}) and no anticipation (\ref{eq:no-ant}), the second equality by the law of iterated expectations, the third equality by Assumption (\ref{eq:augmented_independence_upsilon}) and the Gaussian income process (\ref{eq:income_process}), and the last two equalities by the reverse logic of the first three equalities.

Second case: $k\notin[2,K+1]$. We focus on the case of $k=1$, and the same argument applies in the cases where $k=0$ or where $k>K+1$.\footnote{We include three years of pre-layoff earnings in our propensity score estimation. It is more likely that the most recent (and post-layoff) earnings determine the training decision than earnings from more than three years ago, as reflected in the discussion of ``temporal alignment'' by \citet{CalonicoSmith2017_app}.} Under our assumptions, for $t\geqslant 1$,
\begin{align*}
 & E[Y_{t}(0)|D_{1}=0,D_{2}=1,\mathbf{X}^{1}]\\
\overset{\text{(i)}}{=} & E[Y_{t}(0)|Y_{0}(0)+\upsilon_{1}\geqslant\bar{y}_1,Y_{1}(0)+\upsilon_{2}<\bar{y}_2,\mathbf{X}^{1}]\\
\overset{\text{(ii)}}{=} & E[E[Y_{t}(0)|\upsilon_{1},\upsilon_{2},Y_{1}(0),\mathbf{X}^{1}]|Y_{0}(0)+\upsilon_{1}\geqslant\bar{y}_1,Y_{1}(0)+\upsilon_{2}<\bar{y}_2,\mathbf{X}^{1}]\\
\overset{\text{(iii)}}{=} & E[E[Y_{t}(0)|Y_{1}(0),\mathbf{X}^{1}]|Y_{0}(0)+\upsilon_{1}\geqslant\bar{y}_1,Y_{1}(0)+\upsilon_{2}<\bar{y}_2,\mathbf{X}^{1}]\\
\overset{\text{(iv)}}{=} & E[\alpha_{t}+\beta_{t}Y_{1}(0)+\mathbf{X}^{1}\boldsymbol{\gamma}_{t}|Y_{0}(0)+\upsilon_{1}\geqslant\bar{y}_1,Y_{1}(0)+\upsilon_{2}<\bar{y}_2,\mathbf{X}^{1}]\\
\overset{\text{(v)}}{=} & \alpha_{t}+\beta_{t}E[Y_{1}(0)|Y_{0}(0)+\upsilon_{1}\geqslant\bar{y}_1,Y_{1}(0)+\upsilon_{2}<\bar{y}_2,\mathbf{X}^{1}]+\mathbf{X}^{1}\boldsymbol{\gamma}_{t}\\
\overset{\text{(vi)}}{\leqslant} & \alpha_{t}+\beta_{t}E[Y_{1}(0)|Y_{0}(0)+\upsilon_{1}\geqslant\bar{y}_1,Y_{1}(0)+\upsilon_{2}\geqslant\bar{y}_2,\mathbf{X}^{1}]+\mathbf{X}^{1}\boldsymbol{\gamma}_{t}\\
\overset{\text{(vii)}}{=} & E[Y_{t}(0)|D_{1}=0,D_{2}=0,\mathbf{X}^{1}],
\end{align*}
where (i) follows from the training decision rule (\ref{eq:training_decision_rule}), (ii) from the law of iterated expectations, (iii) from Assumption (\ref{eq:augmented_independence_upsilon}), (iv) from the income process specification (\ref{eq:income_process}), (v) from the linearity of the expectation operator, (vi) from Corollary \ref{cor:subvec_MTP2} (which together with Theorem 4.1 of \citealp{KarlinandRinott1980a_app} implies $\beta_{t}\geqslant0$ for $t\geqslant2$; note: $\beta_{t}=1$ when $t=1$) and Corollary \ref{cor:AC_2_period}, and (vii) from the reverse process of (i)-(v).
\end{proof}
\begin{remark}
\label{rem:corr_shock_spec}
\normalfont
As mentioned in Remark \ref{rem:training_cost}, if we interpret $\upsilon$ as reflecting the cost of training per \citet{HeckmanandRobb1985JOE_app}, it may be positively serially correlated and negatively correlated with potential earnings. We can relax Assumption (\ref{eq:augmented_independence_upsilon}) to allow these correlations. A simple way to do so is to set $\upsilon_{is}=r \omega_{i}+u_{is}$ where $r \in [-1,0]$ and where the process $\{u_{is}\}$ is serially independent and is also independent to potential earnings. We can show that, for this more general case, Proposition \ref{prop:AC_2_period} still holds. Because the proof is similar in spirit to what we have above but notationally more convoluted as it involves a long vector of primitive error terms of $\omega_i$, $\{u_{is}\}_{s=1,2}$, $\{\eta_{is}\}_{s=-K+1,...,t}$, and $\epsilon_{i(-K)}$, we omit it here.
\end{remark}

\begin{remark}
\label{rem:2period_testability}
\normalfont
Under the framework above, setting $t=1$ leads to
\begin{equation}
E[Y_{1}|D_{1}=0,D_{2}=1,\mathbf{X}^{1}]\leqslant E[Y_{1}|D_{1}=0,D_{2}=0,\mathbf{X}^{1}].\label{eq:test_assumption_2}
\end{equation}
Inequality (\ref{eq:test_assumption_2}) amounts to an empirical test of Assumption \ref{assu:S2_potential_outcome}. It states that among period-1 non-enrollees with the same $\mathbf{X}^{1}$, those who enroll in period 2 have lower average earnings in period 1. This test is similar in spirit to the pre-program test by \citet{HeckmanHotz1989_app} in that it uses pre-treatment data of the to-be treated group (and the control group) to assess the validity of identifying assumptions.
\end{remark}

While we focus on the two-period case in this section for ease of exposition, we generalize our assumptions and results---including the test in Remark \ref{rem:2period_testability}---to an arbitrary number of periods in Appendix \ref{subsec:S-period-Generalization}.

\subsection{Details on Partial Identification Results\label{subsec:Details-on-Bounds}}

\subsubsection*{Proof of Proposition \ref{prop:Prop_S2_lb}}

We first state and prove two lemmas. Lemma \ref{lem:lem_X_lb} has a similar form to part (a) of Proposition \ref{prop:Prop_S2_lb}, except the estimand conditions on the covariates directly, as opposed to the propensity scores:
\begin{lem}
\label{lem:lem_X_lb}Under Assumptions \ref{assu:S2_seq_unconfound} and \ref{assu:S2_potential_outcome} and provided that $p_{1}(\mathbf{X}^{1}),p_{2}(\mathbf{X}^{2})<1$,
\begin{align}
E[Y_t|D_{1} & =1]-E\left[E[Y_t|D=0,\mathbf{X}^{1}]|D_{1}=1\right]\leqslant E[Y_t(1)-Y_t(0)|D_{1}=1]\text{ for }t\geqslant1\label{eq:lem_X_lb_eq1}\\
E[Y_t|D_{2} & =1]-E\left[E[Y_t|D=0,\mathbf{X}^{2}]|D_{2}=1\right]=E[Y_t(1)-Y_t(0)|D_{2}=1]\text{ for }t\geqslant2.\label{eq:lem_X_lb_eq2}
\end{align}
\end{lem}
\begin{proof}
Equation (\ref{eq:lem_X_lb_eq2}) directly follows from Assumption \ref{assu:S2_seq_unconfound} and the overlap condition $p_{2}(\mathbf{X}^{2})<1$, and the proof is similar to that of the static identification result (\ref{eq:static_identification}).

To prove inequality (\ref{eq:lem_X_lb_eq1}), first notice that
\begin{align}
 & E[Y_t(0)|D_{1}=1,\mathbf{X}^{1}]\nonumber \\
\overset{\text{(i)}}{=} & E[Y_t(0)|D_{1}=0,\mathbf{X}^{1}]\nonumber\\
= & \sum_{d_{2}=0,1}E[Y_{t}(0)|D_{1}=0,D_{2}=d_{2},\mathbf{X}^{1}]\Pr(D_{2}=d_{2}|D_{1}=0,\mathbf{X}^{1})\nonumber \\
\leqslant & E[Y_t(0)|D_{1}=0,D_{2}=0,\mathbf{X}^{1}]\nonumber \\
= & E[Y_t(0)|D=0,\mathbf{X}^{1}] = E[Y_t|D=0,\mathbf{X}^{1}],\label{eq:lem_X_lb_proof}
\end{align}
where equality (i) holds because of Assumption \ref{assu:S2_seq_unconfound}. Our desired result follows:
\begin{align*}
 & E[Y_t|D_{1}=1]-E[E[Y_t|D=0,\mathbf{X}^{1}]|D_{1}=1]\\
\leqslant & E[Y_t(1)|D_{1}=1]-E\left[E[Y_t(0)|D_{1}=1,\mathbf{X}^{1}]|D_{1}=1\right]\\
= & E[Y_t(1)-Y_t(0)|D_{1}=1].
\end{align*}

\end{proof}
Lemma \ref{lem:ps_property} is a variant of a general property of the propensity score---conditional on the propensity score, any function of covariate inputs of the propensity score is balanced across treatment and control.
\begin{lem}
\label{lem:ps_property}Let $g(\mathbf{X}^{1})$ be a function of $\mathbf{X}^{1}$,
\[
g(\mathbf{X}^{1})\independent D_{1}|D_{2}=0,p_{1}(\mathbf{X}^{1}).
\]
\end{lem}
\begin{proof}
Note that
\begin{align*}
 & \Pr(D_{1}=1|g(\mathbf{X}^{1}),D_{2}=0,p_{1}(\mathbf{X}^{1}))\\
= & E[D_{1}|g(\mathbf{X}^{1}),D_{2}=0,p_{1}(\mathbf{X}^{1})]\\
= & E\left[E[D_{1}|\mathbf{X}^{1},D_{2}=0]|g(\mathbf{X}^{1}),D_{2}=0,p_{1}(\mathbf{X}^{1})\right]\\
= & E[p_{1}(\mathbf{X}^{1})|g(\mathbf{X}^{1}),D_{2}=0,p_{1}(\mathbf{X}^{1})]\\
= & p_{1}(\mathbf{X}^{1}),
\end{align*}
and by similar reasoning 
\[
\Pr(D_{1}=1|D_{2}=0,p_{1}(\mathbf{X}^{1}))=p_{1}(\mathbf{X}^{1}).
\]
The independence result follows.
\end{proof}
Now we prove Proposition \ref{prop:Prop_S2_lb}.
\begin{proof}
The proof of equation (\ref{eq:ID_ps_T2_lb}) is analogous to that of the static propensity score matching result per \citet{RosenbaumRubin1983_app}, so we omit it here.

To prove inequality (\ref{eq:ID_ps_T1_lb}), we start with the propensity score analog of equation (\ref{eq:lem_X_lb_proof}) in the proof of Lemma \ref{lem:lem_X_lb} 
\begin{align*}
 & E[Y_t(0)|D_{1}=1,p_{1}(\mathbf{X}^{1})]\\
\overset{\text{(i)}}{=} & E\left[E[Y_t(0)|D_{1}=1,\mathbf{X}^{1}]|D_{1}=1,p_{1}(\mathbf{X}^{1})\right]\\
\overset{\text{(ii)}}{\leqslant} & E\left[E[Y_t(0)|D=0,\mathbf{X}^{1}]|D_{1}=1,p_{1}(\mathbf{X}^{1})\right]\\
\overset{\text{(iii)}}{=} & E\left[E[Y_t(0)|D=0,\mathbf{X}^{1}]|D_{1}=1,D_{2}=0,p_{1}(\mathbf{X}^{1})\right]\\
\overset{\text{(iv)}}{=} & E\left[E[Y_t(0)|D=0,\mathbf{X}^{1}]|D_{1}=0,D_{2}=0,p_{1}(\mathbf{X}^{1})\right]\\
\overset{\text{(v)}}{=} & E[Y_t|D=0,p_{1}(\mathbf{X}^{1})],
\end{align*}
where equalities (i) and (v) follow the law of iterated expectations, inequality (ii) follows (\ref{eq:lem_X_lb_proof}), equality (iii) follows the fact that $D_{1}=1$ implies $D_{2}=0$, and equality (iv) follows Lemma \ref{lem:ps_property} and the observation that the inner conditional expectation is just a function of $\mathbf{X}^{1}$. The lower bound result in (\ref{eq:ID_ps_T1_lb}) easily follows: 
\begin{align*}
 & E[Y_t|D_{1}=1]-E\left[E[Y_t|D=0,p_{1}(\mathbf{X}^{1})]|D_{1}=1\right]\\
\leqslant & E[Y_t(1)|D_{1}=1]-E\left[E[Y_t(0)|D_{1}=1,p_{1}(\mathbf{X}^{1})]|D_{1}=1\right]\\
= & E[Y_t(1)-Y_t(0)|D_{1}=1].
\end{align*}

The lower bound result for the overall TOT in Part (b) of Proposition \ref{prop:Prop_S2_lb} is a simple consequence of part (a). Aggregating (\ref{eq:ID_ps_T1_lb}) and (\ref{eq:ID_ps_T2_lb}) with weights equal to the respective share of workers trained in period 1 and period 2 delivers the desired result.
\end{proof}

\subsubsection*{Construction of Upper Bounds via Propensity Score Matching}

Proposition \ref{prop:Prop_S2_ub} below constructs the upper bounds for the TOT parameters via propensity score matching. Proposition \ref{prop:Prop_S2_ub}(a) bounds $\delta_{1t}$ from above. Proposition \ref{prop:Prop_S2_ub}(b) aggregates the upper bound on $\delta_{1t}$ from part (a) and the $\delta_{2t}$ from equation (\ref{eq:ID_ps_T2_lb}) to arrive at an upper bound on the overall TOT parameter $\delta_t$. This upper bound construction uses the treatment-now-versus-no-treatment-now propensity scores $\tilde{p}_{1}(\mathbf{X}^{1})\equiv\Pr(D_{1}=1|\mathbf{X}^{1})$, as introduced in Section \ref{subsec:Relations-to-DTE-lit}, and $\tilde{p}_{s}(\mathbf{X}^{s})\equiv\Pr(D_{s}=1|\mathbf{X}^{s},D_{1}=\dots=D_{s-1}=0)$ for $s\geqslant2$, which reflect the probability of treatment in period $s$ among workers not yet treated conditional on observables up to $s$ (note that $\tilde{p}_{2}(\mathbf{X}^{2})$ coincides with $p_{2}(\mathbf{X}^{2})$ defined in Section \ref{subsec:Identification}).

\begin{prop}
\label{prop:Prop_S2_ub}Under Assumption \ref{assu:S2_seq_unconfound}
and provided that $\tilde{p}_{1}(\mathbf{X}^{1}),\tilde{p}_{2}(\mathbf{X}^{2})<1$
and $Y_t(0)\geqslant0$ for $t\geqslant1$,

\noindent (a): For $t\geqslant1$, 
\begin{align}
& E\left[Y_t|D_{1}=1\right]-E\left[E[Y_t|D=0,\tilde{p}_{1}(\mathbf{X}^{1})]\Pr(D_{2}=0|D_{1}=0,\tilde{p}_{1}(\mathbf{X}^{1}))|D_{1}=1\right]\nonumber \\
\geqslant & E[Y_t(1)-Y_t(0)|D_{1}=1];\label{eq:ID_ps_T1_ub}
\end{align}
(b): For $t\geqslant2$, 
\begin{align*}
& \sum_{s=1}^{2}\left\{ E\left[Y_t|D_{s}=1\right]-E\left[E[Y_t|D=0,\tilde{p}_{s}(\mathbf{X}^{s})]\Pr(D_{2}=0|D_{s}=0,\tilde{p}_{s}(\mathbf{X}^{s}))|D_{s}=1\right]\right\} \Pr(D_{s}=1|D=1)\\
\geqslant & E[Y_t(1)-Y_t(0)|D=1].
\end{align*}
\end{prop}

\begin{proof}
For part (a), notice that under Assumption \ref{assu:S2_seq_unconfound} and provided $Y_t(0)\geqslant0$,
\begin{align*}
 & E[Y_t(0)|D_{1}=1,\tilde{p}_{1}(\mathbf{X}^{1})]\\
= & E[Y_t(0)|D_{1}=0,\tilde{p}_{1}(\mathbf{X}^{1})]\\
= & \sum_{d_{2}=0,1}E[Y_{t}(0)|D_{1}=0,D_{2}=d_{2},\tilde{p}_{1}(\mathbf{X}^{1})]\Pr(D_{2}=d_{2}|D_{1}=0,\tilde{p}_{1}(\mathbf{X}^{1}))\\
\geqslant & E[Y_t(0)|D=0,\tilde{p}_{1}(\mathbf{X}^{1})]\Pr(D_{2}=0|D_{1}=0,\tilde{p}_{1}(\mathbf{X}^{1})).
\end{align*}
Our desired result follows:
\begin{align*}
 & E[Y_t(1)-Y_t(0)|D_{1}=1]\\
\leqslant & E[Y_t(1)|D_{1}=1]-E\left[E[Y_t(0)|D=0,\tilde{p}_{1}(\mathbf{X}^{1})]\Pr(D_{2}=0|D_{1}=0,\tilde{p}_{1}(\mathbf{X}^{1}))|D_{1}=1\right]\\
= & E[Y_t|D_{1}=1]-E\left[E[Y_t|D=0,\tilde{p}_{1}(\mathbf{X}^{1})]\Pr(D_{2}=0|D_{1}=0,\tilde{p}_{1}(\mathbf{X}^{1}))|D_{1}=1\right].
\end{align*}

\noindent For part (b), noting that $\Pr(D_{2}=0|D_{2}=0,\tilde{p}_{2}(\mathbf{X}^{2}))=1$, the period 2 quantity in the summand is simply the left hand side of equation (\ref{eq:ID_ps_T2_lb}). The result follows from aggregating inequality (\ref{eq:ID_ps_T1_ub}) and equation (\ref{eq:ID_ps_T2_lb}) using weights equal to the respective share of workers trained in the two periods.
\end{proof}

\subsubsection*{Estimating the Bounds}

As mentioned in Section \ref{subsec:Identification}, the identification results of Proposition \ref{prop:Prop_S2_lb}(a) lead to standard propensity score matching estimators of the lower bound for $\delta_{1t}$ and the value of $\delta_{2t}$. We denote these estimators by $\hat{\delta}_{1t}^{\text{lb}}$ and $\hat{\delta}_{2t}$, whose asymptotic variances are $v_{1t}^{\text{lb}}$ and $v_{2t}$, respectively. Following \citet{Abadie_Imbens2016}, the distributions of $\hat{\delta}_{1t}^{\text{lb}}$ and $\hat{\delta}_{2t}$ are asymptotically normal, $v_{1t}^{\text{lb}}$ and $v_{2t}$ can be consistently estimated, and inference can be conducted accordingly. Based on Proposition \ref{prop:Prop_S2_lb}(b), the natural lower bound estimator, $\hat{\delta}_{t}^{\text{lb}}$, for the overall TOT parameter is simply $\hat{\delta}_{t}^{\text{lb}}=\pi_{1}\hat{\delta}_{1t}^{\text{lb}}+\pi_{2}\hat{\delta}_{2t}$, where $\pi_{1}$ and $\pi_{2}$ are the respective shares of workers beginning enrollment in period 1 and 2. For inference, we follow \citet[441]{Imbensandrubin2015_app} and treat $\pi_{1}$ and $\pi_{2}$ as non-random. While the same worker may appear in multiple UI claims and the same unit may contribute to both $\hat{\delta}_{1t}^{\text{lb}}$ and $\hat{\delta}_{2t}$, such multiplicity is limited in our empirical setting (see the counts of UI claims and unique individuals reported in the notes to Figure \ref{fig:matched_overall} and Appendix Figure \ref{fig:by_enroll_timing}). Therefore, for simplicity, we calculate our reported standard errors based on estimates of $\pi_{1}^{2}v_{1t}^{\text{lb}}+\pi_{2}^{2}v_{2t}$, which approximates the asymptotic variance of $\hat{\delta}_{t}^{\text{lb}}$ by omitting covariance between the cohort-specific estimators. We assess sensitivity to this approximation in two ways. First, as reported in Section \ref{subsec:Overall-Returns}, IPW estimation with a county-level block bootstrap yields similar standard errors. Second, taking the cohort-specific variance estimates as given, we calculate a ceiling on the aggregate standard error by the Cauchy-Schwarz inequality, which is attained under perfect positive correlation across cohort-specific estimators, an extreme assumption. In our empirical setting with eight enrollment cohorts, this ceiling is less than three times the reported standard errors for our main estimates, and our estimates remain highly statistically significant.

Based on Proposition \ref{prop:Prop_S2_ub}, we can also estimate the upper bound for $\delta_{1t}$, denoted by $\delta_{1t}^{\text{ub}}$, via propensity score matching. With nearest neighbor matching, for example, the dependent variable for each matched observation in the $D=0$ population is the product of $Y_t$ and an estimate of $\Pr(D_{2}=0|D_{1}=0,\tilde{p}_{1}(\mathbf{X}^{1}))$. We can follow \citet{Heckman_etal1997_app,Heckmanetal1998_app,Heckmanetal1998_ECMA_app,HeckmanSmith1999_app} and estimate this conditional probability function at different values of the propensity score using a local linear regression. Under standard regularity conditions for nonparametric regressions and the correct parametric specification of the propensity score model---an assumption maintained by \citet{Abadie_Imbens2016_app}---the nonparametric estimate is consistent, and the sampling error can be ignored asymptotically when conducting inference using the \citet{Abadie_Imbens2016_app} procedure. Finally, we can estimate the upper bound for $\delta_t$ with $\hat{\delta}_{t}^{\text{ub}}=\pi_{1}\hat{\delta}_{1t}^{\text{ub}}+\pi_{2}\hat{\delta}_{2t}$ and conduct inference similarly to that for $\hat{\delta}_{t}^{\text{lb}}$.

\subsection{\texorpdfstring{$S$}{S}-period Generalization\label{subsec:S-period-Generalization}}

Now we generalize Assumptions \ref{assu:S2_seq_unconfound} and \ref{assu:S2_potential_outcome} and the partial identification results from two periods to $S$ periods. First, analogous to Section \ref{subsec:Identification}, we use $D_{s}$ to denote whether a worker first enrolls in period $s$ for $s=1,\dots,S$. As such, at most one $D_{s}$ can take on the value of $1$, and when it does, the rest of the period-specific enrollment variables equal zero, i.e., $D_{s^\prime}=0$ for $s^\prime \neq s$. Our main treatment variable of interest $D$ is defined as $D\equiv\sum_{s=1}^{S}D_{s}$; it follows that $D=1$ if a worker enrolls sometime during the $S$ periods and $D=0$ otherwise.

Second, we define the period-$s$ propensity score as 
\[
p_{s}(\mathbf{X}^{s})\equiv\Pr(D_{s}=1|D_{s'}=0\text{ for all }s'\neq s,\mathbf{X}^{s})
\]
for $s=1,\dots,S$. Just like our definition of $p_{1}(\mathbf{X}^1)\equiv\Pr(D_{1}=1|D_{2}=0,\mathbf{X}^1)$ in the two-period case, the general propensity score $p_{s}(\mathbf{X}^{s})$ conditions on future treatment status (unless $s=S$). This reflects the fact that our comparison group consists only of workers who never enroll during the $S$ periods, as opposed to those who have not enrolled before period $s$ but may enroll in subsequent periods. The propensity score in the latter scenario would only condition on $D_{s^\prime}=0 \text{ for all } s^\prime < s$, which is the one used by the treatment-now-versus-no-treatment-now literature.

With the generalized definition of the treatment variables and the propensity score, we can state the assumptions that generalize Assumptions \ref{assu:S2_seq_unconfound} and \ref{assu:S2_potential_outcome}, respectively: 
\begin{assumption}
\label{assu:Sgeneral_seq_unconfound}$Y_t(0)\independent D_{1}|\mathbf{X}^{1}$ for $t\geqslant1$ and $Y_t(0)\independent D_{s}|D_{s-1}=\dots=D_{1}=0,\mathbf{X}^{s}$ for $s=2,\dots,S$ and $t\geqslant s$.
\end{assumption}
\begin{assumption}
\label{assu:Sgeneral_potential_outcome}$E[Y_t(0)|D_{s+l}=1,\mathbf{X}^{s}]\leqslant E[Y_t(0)|D=0,\mathbf{X}^{s}]$  for $1\leqslant s\leqslant S$, $1\leqslant l\leqslant S-s$, and $t\geqslant s$.
\end{assumption}

We now state the generalized results for an arbitrary $S$. We begin with Proposition \ref{prop:AC_S-period}, which generalizes Proposition \ref{prop:AC_2_period} and establishes that Assumption \ref{assu:Sgeneral_potential_outcome} is implied by classic models of training participation. Remark \ref{rem:Speriod_testability} generalizes Remark \ref{rem:2period_testability} and states the testable implications of Assumption \ref{assu:Sgeneral_potential_outcome}. Finally, Propositions \ref{prop:Sgeneral_ps_lb} and \ref{prop:Sgeneral_ps_ub} construct lower and upper bounds of the TOT effects, respectively. The proofs of the propositions of this section are analogous to their two-period counterparts, but involve much more complex notations, and are therefore omitted. 
\begin{prop}
\label{prop:AC_S-period}
The training decision rule (\ref{eq:training_decision_rule}), Assumptions (\ref{eq:augmented_independence_upsilon}) and (\ref{eq:no-ant}), and specifications of the income and shock processes (\ref{eq:income_process}) and (\ref{eq:shock_process}) imply Assumption \ref{assu:Sgeneral_potential_outcome}.
\end{prop}
\begin{remark}
\label{rem:Speriod_testability}
\normalfont
We can generalize the testable implication of Assumption \ref{assu:S2_potential_outcome} in Remark \ref{rem:2period_testability} to those of Assumption \ref{assu:Sgeneral_potential_outcome}. Specifically, comparing workers who have similar $\mathbf{X}^s$ ($1\leqslant s\leqslant S$), those who enroll $l$ periods later ($1\leqslant l\leqslant S-s$) have lower average earnings than their never-enrolled counterparts in all $l$ interim periods. We conduct these tests in Section \ref{subsec:Overall-Returns} and show evidence consistent with Assumption \ref{assu:Sgeneral_potential_outcome} and hence models of training participation by \citet{Heckman1978_app}, \cite{AshenfelterandCard1985_app}, and \citet{HeckmanandRobb1985JOE_app,HeckmanandRobb1985book_app}.
\end{remark}

\begin{prop}
\label{prop:Sgeneral_ps_lb}Under Assumptions \ref{assu:Sgeneral_seq_unconfound} and \ref{assu:Sgeneral_potential_outcome} and provided that $p_{s}(\mathbf{X}^{s})<1$ for all $s$,

\noindent (a): for $s=1,\dots,S-1$ and $t \geqslant s$,
\[
E\left[Y_t|D_{s}=1\right]-E\left[E[Y_t|D=0,p_{s}(\mathbf{X}^{s})]|D_{s}=1\right]\leqslant E[Y_t(1)-Y_t(0)|D_{s}=1],
\]
and for $t \geqslant S$
\[
E\left[Y_t|D_{S}=1\right]-E\left[E[Y_t|D=0,p_{S}(\mathbf{X}^{S})]|D_{S}=1\right]=E[Y_t(1)-Y_t(0)|D_{S}=1];
\]
(b): for $t\geqslant S$,
\[
\sum_{s=1}^{S}\left\{ E\left[Y_t|D_{s}=1\right]-E\left[E[Y_t|D=0,p_{s}(\mathbf{X}^{s})]|D_{s}=1\right]\right\} \Pr(D_{s}=1|D=1)\leqslant E[Y_t(1)-Y_t(0)|D=1].
\]
\end{prop}
\bigskip
\begin{prop}
\label{prop:Sgeneral_ps_ub}Under Assumption \ref{assu:Sgeneral_seq_unconfound} and provided that $\tilde{p}_{s}(\mathbf{X}^{s})<1$ for all $s$ and $Y_t(0)\geqslant0$ for all $t$,

\noindent (a): for $s=1,\dots,S-1$ and $t \geqslant s$,
\[
E\left[Y_t|D_{s}=1\right]-E\left[E[Y_t|D=0,\tilde{p}_{s}(\mathbf{X}^{s})]\Pr(D=0|D_{1}=\dots=D_{s}=0,\tilde{p}_{s}(\mathbf{X}^{s}))|D_{s}=1\right]\geqslant E[Y_t(1)-Y_t(0)|D_{s}=1];
\]
(b): for $t \geqslant S$,
\begin{align*}
 & \sum_{s=1}^{S}\left\{ E\left[Y_t|D_{s}=1\right]-E\left[E[Y_t|D=0,\tilde{p}_{s}(\mathbf{X}^{s})]\Pr(D=0|D_{1}=\dots=D_{s}=0,\tilde{p}_{s}(\mathbf{X}^{s}))|D_{s}=1\right]\right\} \Pr(D_{s}=1|D=1)\\
\geqslant & E[Y_t(1)-Y_t(0)|D=1].
\end{align*}
\end{prop}

\subsection{Relation to the Dynamic Treatment Effect Literature: Details\label{subsec:DTE_lit_details}}

\subsubsection{Identification in the Robins Framework\label{subsec:Robins_details}}

In this section, we state the Robins identification results by adapting the excellent summary of \citet{AbbringandHeckman2007_app}. For ease of exposition, we focus on the two-period setting, but the results easily generalize to more periods. With two periods, there are four possible treatment sequences per Robins: trained in both periods, trained only in period 1, trained only in period 2, and did not receive training in either period. If we denote the training decision in period $s$ by $\tilde{D}_{s}$, then a treatment sequence $g$ is the concatenation of $\tilde{D}_{1}$ and $\tilde{D}_{2}$ and takes on the value of 11, 10, 01, or 00. We further denote the corresponding potential outcome for treatment sequence $g$ at time $t\geqslant1$ by $Y_{t}^{g}$, and the observation rule (or the consistency condition) is that $Y_t=Y^{g}_{t}$ when the actual treatment sequence is $g$.

There are two key assumptions in the Robins framework. The first is sequential randomization: for each treatment sequence $g=11,10,01,00$,
\begin{equation}
\tilde{D}_{1}\independent Y_{t}^{g}|\mathbf{X}^{1}\text{ for }t\geqslant1\text{; and }\tilde{D}_{2}\independent Y_{t}^{g}|(\mathbf{X}^{1},Y_{1}),\tilde{D}_{1}\text{ for }t\geqslant2.\label{eq:Robins-SR}
\end{equation}
For $g=00$, the restrictions in (\ref{eq:Robins-SR}) coincide with our Assumption \ref{assu:S2_seq_unconfound} (CIA), when the period-2 restriction in the former is conditional on $\tilde{D}_1=0$ and the period-2 conditioning set in the latter takes the form $\mathbf{X}^{2}=(\mathbf{X}^{1},Y_{1})$. Robins's second assumption is no anticipation:
\begin{equation}
Y_{1}^{01}=Y_{1}^{00}\text{ and }Y_{1}^{11}=Y_{1}^{10}.\label{eq:Robins-NA}
\end{equation}
It states that potential outcomes in period 1 do not depend on the future treatment decision in period 2. As a side note, the first equality in (\ref{eq:Robins-NA}), along with consistency, implies our period-1 no-anticipation assumption (\ref{eq:no-ant}) in Appendix \ref{subsec:Discussion-of-Assumption2}: $Y_1=Y_1(0)$ when $D_1=0$. To see this, first note that $Y_1=Y_1(0)$ for $D=0$ (or $D_1=D_2=0$) following our observation rule, so that the substantive part of our no-anticipation assumption is $Y_1=Y_1(0)$ when $D_1=0$ but $D_2=1$. In the Robins framework, the $D_2=1$ population is subject to the $g=01$ treatment sequence, and for them, $Y_1=Y^{01}_1=Y^{00}_1$ where the second equality holds under (\ref{eq:Robins-NA}). Since our $D=0$ population is subject to the $g=00$ treatment sequence, our $Y_1(0)$ is the same as Robins's $Y^{00}_1$, and we have the desired result. Finally, we can apply this reasoning to show that the first part of (\ref{eq:no-ant}) is consistent with Robins by extending the definition of treatment sequences to cover times before the first period.

Under the assumptions in (\ref{eq:Robins-SR}) and (\ref{eq:Robins-NA}), the joint distributions of potential outcomes across time are identified for each counterfactual treatment sequence. For example, the potential outcome distributions under $g=00$ are identified as: for $t\geqslant2$,
\begin{equation}
\Pr(Y_{t}^{00}=y_{t},Y_{1}^{00}=y_{1}|\mathbf{X}^{1})=\Pr(Y_{1}=y_{1}|\tilde{D}_{1}=0,\mathbf{X}^{1})\Pr(Y_{t}=y_{t}|\tilde{D}_{1}=0,\tilde{D}_{2}=0,Y_{1}=y_{1},\mathbf{X}^{1}),\label{eq:Robins_ID}
\end{equation}
and identifications for other $g$ are similar. The proof of (\ref{eq:Robins_ID}) follows that in Section 3.2 of \citet{AbbringandHeckman2007_app} and is omitted here. The estimand in (\ref{eq:Robins_ID}) is the sequential product of conditional outcome distributions using observations along the path of $g$. 

Nonparametrically estimating these distributions can suffer from the curse of dimensionality, particularly when the conditioning set contains multiple continuous variables (\citealp{Robins1998_app,Robins2000_app}). To circumvent the challenge, Robins and various coauthors impose parametric restrictions on the relationship between potential outcomes, treatment, and confounders. For the structural nested model estimator (e.g., \citealp{Robins1994_app}), the restriction applies to the relationship between differences in potential outcomes (i.e., causal effects), treatment, and time-varying confounders in each period. For the marginal structural model estimator (e.g., \citealp{Robins1998_app}), the restriction applies to the relationship between levels of potential outcomes, summary measures of the treatment sequence (e.g., length of treatment exposure), and baseline confounders. Interested readers should consult \citet{Robins2000_app} and references therein for a summary and comparison of his various estimators.

\subsubsection{Practical Challenges in Adapting the \texorpdfstring{\citet{Lechner2009JBES_app}}{Lechner (2009)} Inverse Propensity Weighting Estimator\label{subsec:Lechner_IPW}}

Focusing on the two-period case, \citet{Lechner2009JBES_app} proposes a sequential inverse propensity weighting (IPW) estimator for the counterfactual outcome of no training for workers enrolled in the first period. In his empirical analysis, \citet{Lechner2009JBES_app} reports standard errors obtained from five different methods and notes that they reassuringly lead to the same rejection decisions in most cases. However, \citet{Lechner2009JBES_app} does not choose a preferred standard error estimator and states that ``it is beyond the scope of {[}his{]} article to investigate the issue of precise variance estimation of the IPW estimator in depth.'' 

It turns out that the IPW estimator takes a more complex form  in the general $S$-period case. With the treatment-now-versus-no-treatment-now propensity scores defined in Section \ref{subsec:Relations-to-DTE-lit} and Appendix \ref{subsec:Details-on-Bounds}, $\tilde{p}_{s}(\mathbf{X}^{s})$, the IPW estimand is given by the following proposition.
\begin{prop}
\label{prop:Lechner_IPW}Under Assumption \ref{assu:Sgeneral_seq_unconfound} and provided that $\tilde{p}_{s}(\mathbf{X}^{s})<1$ for all $s$,
\begin{equation}
E[Y_t(0)|D_{1}=1]=\frac{1}{\Pr(D_{1}=1)}E\left[\frac{\tilde{p}_{1}(\mathbf{X}^{1})\cdot Y_t \cdot 1_{[D=0]}} {\prod_{s=1}^{S}\left(1-\tilde{p}_{s}(\mathbf{X}^{s})\right)}\right]\text{ for \ensuremath{t \geqslant S}}.\label{eq:Lechner-IPW-estimand}
\end{equation}
\end{prop}
\noindent We omit the proof of Proposition \ref{prop:Lechner_IPW}, as it simply extends that in Appendix B.1 of \citet{Lechner2009JBES_app}.

According to Proposition \ref{prop:Lechner_IPW}, the denominator of the IPW estimator for the counterfactual $Y_t$ mean among period-1 enrollees is the product of eight estimated propensity scores in our eight-period setting. In comparison, the denominator in Lechner's two-period setting is the product of just two propensity scores. Because of this added complexity, we will need to reassess the sensitivity of different inferential methods. Furthermore, with eight periods, the subset of these methods based on generalized method of moments will involve many more parameters and moments, and is more likely to encounter numerical issues than in \citet{Lechner2009JBES_app}. Given these challenges, we choose not to adapt \citet{Lechner2009JBES_app} to our analysis.

\subsubsection{Studies that Model Unobserved Heterogeneity\label{subsec:unobserved_heterogeneity}}

An alternative strand of the dynamic treatment effect literature explicitly models the influence of unobserved heterogeneities. This strand typically models unobserved heterogeneities as ``random effects'', which are assumed to be independent of observed covariates. In this section, we provide a brief overview of key studies in this strand.

Under no-anticipation and a conditional independence assumption (treatment is independent to outcome conditioning on both observed covariates and random effects), \citet{AbbringandvandenBerg2003_app,AbbringandvandenBerg2004_app} prove identification of dynamic treatment effects for duration outcomes in mixed proportional hazard models for both treatment and outcome timings. The word ``mixed'' refers to the model specification that a random effect enters into a proportional hazard model multiplicatively, which is crucial for identification. 

\citet{HeckmanandNavarro2007_app} unite the literature on dynamic treatment and on discrete choice---they consider semiparametric identification in a statistical framework where treatment in each period is determined by an index-crossing model. They allow anticipatory effects, and their identification relies on a large support condition that ensures sufficient variation of the index and helps to break treatment's dependence on unobserved heterogeneity. \citet{HeckmanandNavarro2007} also propose an economically interpretable structural framework, and interested readers should consult \citet{Cunhaetal2007_app} and \citet{Heckmanetal2016JOE_app} for extensions. 

\citet{Baetal2017} study labor market transitions among female adults recommended for classroom training in the NJS. They restrict to participants who enter the program during a nonemployment spell (following \citealp{EberweinHamLaLonde1997_app}), and they specify the joint likelihood of this initial nonemployment spell, timing of training participation, and subsequent employment and nonemployment spells. The likelihood is integrated over distributions of the random effects, which are assumed to be discrete with finite supports as per \citet{HeckmanandSinger1984_app} and \citet{McCall1996_app}. \citet{Baetal2017_app} estimate their complex model via simulated annealing.

\citet{Han2021_app} specifies a nonparametric model for dynamic treatments and outcomes. Assuming exclusion restrictions (i.e., the existence of suitable instruments) and a sequential rank similarity condition on the unobserved heterogeneity variables, the paper proves identification of dynamic average treatment effects.

Most recently, \citet{Fitzenbergeretal2023_app} estimate training effects using a flexible model of employment and training transitions using German administrative data. They assume sequential randomization but add an individual specific effect into the conditioning set, which in essence extends Robins's approach by incorporating unobserved heterogeneity. They estimate their model using Bayesian Markov Chain Monte Carlo, which allows them to directly get at the posterior distribution of the unobserved effect for each individual. Doing so accounts for potential dynamic selections on unobservables, which represents an important distinction from \citet{Baetal2017_app}.

\section{Alternative Identification Strategies\label{sec:Alternative-Identification-Strat}}

\subsection{\texorpdfstring{\citet{Jacobson_etal2005_JE_app}}{Jacobson, Lalonde and Sullivan (2005a)} Fixed Effects Specifications}

In this section, we consider models similar to those estimated in \citet{Jacobson_etal2005_JE_app} (hereafter, JLS), who use longitudinal earnings data to estimate the effects of attending community colleges for a population of UI claimants. To account for unobserved individual characteristics (that are either constant or evolving linearly), JLS estimate models of the form 
\begin{equation}
Y_{it}=\beta E_{it}+\alpha_{i}+\omega_{i}t+\gamma_{t}+\delta_{it}(s_{i},z_{i})+\varepsilon_{it}\label{eq:jls-1}
\end{equation}
where $Y_{it}$ is the earnings of individual $i$ at time $t$, $E_{it}$ is an indicator for whether the individual has started school as of time $t$, $\alpha_{i}$ and $\text{\ensuremath{\omega_{i}t}}$ are individual fixed effects and linear time trends, $\gamma_{t}$ denote time fixed effects, and $\delta_{it}(s_{i},z_{i})$ are layoff effects that depend on layoff date $s_{i}$ and fixed individual characteristics $z_{i}$.\footnote{This is a simplification of \citet{Jacobson_etal2005_JE_app}'s model in that enrollment here is a binary state and we are estimating the average earnings effects after first enrolling. JLS also consider the incremental earnings effects of credits earned as well as earnings dynamics during and after school by replacing $\beta E_{it}$ with $\phi_{it}(c_{i},f_{i},l_{i},z_{i})$ where $c_{i}$ is credits earned, and $f_{i}$ and $l_{i}$ denote the enrollment entry and exit periods.} In the most basic specifications, $\delta_{it}(s_{i},z_{i})=\sum_{k=-12}^{24}D_{it}^{k}\delta_{k}$, where $D_{it}^{k}$ are a full set of dummy variables for quarter relative to layoff: $D_{it}^{k}=1$ if individual $i$ had been laid off in quarter $t-k$. We also follow JLS and include a set of heterogeneous layoff effects that allow the earnings patterns relative to layoff to parametrically depend on demographic variables $z_{i}$, which include gender, race, whether a claimant lives in one of the ten largest counties, tenure, and ten-year age bins---see \citet{Jacobson_etal2005_JE_app} for the exact form of these heterogeneous layoff patterns.

To check for pre-trends in earnings before enrollment, we estimate models of the form: 
\[
Y_{it}=\sum_{k=-20}^{23}R_{it}^{k}\beta_{k}+\alpha_{i}+\omega_{i}t+\gamma_{t}+\delta_{it}(s_{i},z_{i})+\varepsilon_{it}
\]
where $Y_{it}$ is earnings for person $i$ in quarter $t$, $R_{it}^{k}$ is an indicator for beginning enrollment in quarter $t-k$ (and is equal to 0 for all non-enrollee observations). The regressions are estimated on a 5 percent random sample to speed up computation. We plot the estimated $\beta_{k}$ from the regressions in Appendix Figure \ref{fig:jls_event_ohio}.

The specification shown in Panel A of the figure includes individual fixed effects $\text{\ensuremath{\alpha_{i}}}$, calendar quarter fixed effects $\text{\ensuremath{\gamma_{t}}}$, and quarter relative to layoff dummies; Panel B adds individual time trends $\alpha_{i}+\omega_{i}t$ ; Panel C adds heterogeneous layoff effects per the second term of equation (2) in \citet{Jacobson_etal2005_JE_app}. We find that this ``event study'' specification check yields problematic pre-enrollment earnings differences between enrollees and non-enrollees, regardless of whether we account for individual trends. Therefore, we do not rely on these fixed effects models for our analysis.

\subsection{Distance-Based Instrumental Variables}

Another research design that has been used to estimate educational effects relies on the idea that, all else equal, students are more likely to enroll if they live close to a school (\citealp{Card1993_app}). Therefore, the distance to the nearest school can be used as an instrumental variable (IV) for enrollment. Since we observe the zip codes of UI claimants in our sample, we can adapt this design to our setting by computing the linear distance from a worker's zip code to the nearest community college.\footnote{We use the nearest community college because nearly 90 percent of enrollment is in community colleges.}

Since the validity of the distance IV design hinges on the exogeneity of the distance measure, it is important to control for all other potential determinants of earnings that may also be correlated with distance to schools. In fact, we find workers who live close to a community college to be different along observable dimensions from those who live farther away. Therefore, we control for a similar set of covariates as in our matching specification: 12 quarters of pre-layoff earnings, calendar quarter of layoff, prior industry (eight indicators), prior job tenure category (less than one year, one to six years, and more than six years), age indicators (age below 19, each year from age 19 to 59, and older than 59), whether a worker has a dependent, county unemployment rate during the month of layoff, and county fixed effects. As in our main analysis, we estimate models separately for men and women, and by whether or not they previously worked in manufacturing.

However, we do not find consistently strong first-stage relationships between distance and enrollment. Moreover, we find that the results are sensitive to transformations of the instrument (e.g., including a quadratic distance term). We conclude that using a distance-based instrumental variables design to estimate enrollment effects is not desirable for our particular context.

\subsection{Timing of Layoff}

One may conjecture that the timing of layoff may induce a discontinuous change in enrollment probability. That is, workers laid off just before the beginning of a semester may be much more likely to enroll than those laid off just after. This sudden decrease in probability may be leveraged to identify enrollment effects.

However, there do not appear to be salient discontinuities in Appendix Figure \ref{fig:enroll_by_claimdate}, which plots the fraction of enrollees by UI claim date. This should not be surprising given that schooling decisions may take time to materialize. In fact, the modal number of quarters between UI claim and enrollment is two as documented in the notes of Appendix Figure \ref{fig:by_enroll_timing}. Therefore, implementing a regression discontinuity design in layoff date does not seem to be a viable path to estimating the effect of enrollment for unemployed workers.

\section{Decomposition of Enrollment Estimates by High and Low Wage Firms \label{sec:akm_firm_effects}}

In this section, we examine the role of firms in driving our estimates. Specifically, we ask the extent to which being employed in a higher quality firm (as measured by a firm wage premium as in \citealp{AbowdKramarzMargolis1999_app}---hereafter, AKM) can explain the earnings effects of enrollment.

For this analysis, we first estimate the AKM effects. To do this, we use the universe of earnings records we have, from 2003 to 2014. We estimate models of the form 
\[
y_{it}=\theta_{i}+\psi_{\mathbf{J}(i,t)}+\varepsilon_{it}
\]
where $y_{it}$ is the log earnings of individual $i$ in quarter $t$, $\theta_{i}$ is the individual-specific wage effect, and $\psi_{\mathbf{J}(i,t)}$ is the firm effect for the firm $\mathbf{J}(i,t)$ that employs person $i$ in quarter $t$. The firm effects are only identified for the ``largest connected set'' of firms through worker mobility. These effects are estimated on a sample of 9,593,167 individuals.

Panel A of Appendix Figure \ref{fig:firmakm-prob} shows the probability that enrollees (right panel) and their matched non-enrollees (left panel) are employed at firms that have higher wage premiums than their pre-layoff firm, lower (including same) wage premiums than their pre-layoff firm, or are not employed, over the four years in the post-period.\footnote{A small proportion of workers cannot be categorized because AKM estimates are missing for either their pre-layoff or post-layoff firms.} We see that at the beginning of the period, enrollees and their matched non-enrollees are about equally likely to be employed in a firm with a higher wage premium than pre-layoff, but that non-enrollees are more likely to be employed overall (driven by employment at firms with lower wage premiums). At the end of the post-period, however, enrollees are more likely to be employed in a higher wage premium firm relative to matched non-enrollees. Panel B of Appendix Figure \ref{fig:firmakm-prob} decomposes the earnings effects into the relative contributions of employment in higher premium or lower premium firms, similar to Figure \ref{fig:ind_switch}. We find that, in the 16th quarter post-enrollment, 86 percent of the earnings difference is explained by employment in higher premium firms. 

One natural question, given this finding and Figure \ref{fig:ind_switch}, is whether workers who move to different industries are also moving to higher premium firms. We find that the story is nuanced. Appendix Table \ref{tab:decompqtr16} further decomposes the overall earnings difference between enrollees and their matched non-enrollees in the 16th quarter post-enrollment into components due to those who are employed in a different industry and a higher premium firm, the same industry and a higher premium firm, a different industry and a lower premium firm, and the same industry and a lower premium firm. The first column of this table shows that the earnings difference is not solely driven by those employed in different industries \emph{and} higher premium firms: 55 percent is explained by movement to high-premium firms in different industries but a substantial 46 percent is explained by movement to low-premium firms in different industries.  In contrast, 31 percent is explained by movement to high-premium firms in the same industry, and -36 percent by low-premium firms in the same industry. Therefore, it appears that industry switching plays a larger role than employment in high-premium firms.\footnote{Similarly, \citet{MooreandScott-Clayton2025_app} find that firm pay premiums explain only a minority of earnings losses due to layoff.} The other columns of this table further decompose the earnings effects by whether they are due to differences in employment (columns 2-4) or differences in wages conditional on employment (columns 5-7). Overall, we find that employment differentials play a larger role in explaining the earnings differences than wages conditional on employment (comparing columns 4 and 7). When we decompose the employment differential (column 2), we find that 68 and 62 percent are explained by employment in high- and low-premium firms in different industries, respectively, while only 17 and -48 percent are explained by high- and low-premium firms in the same industry.

\section{Impact of UI Benefit Policies on Enrollment\label{sec:UI-enrollment}}

In this section, we examine the role of UI benefit policies on enrollment decisions. \citet{BarrTurner2015_app} show that generous benefit durations induce more unemployed workers to pursue schooling. We replicate this finding for Ohio and discuss the implications for UI policy in light of our enrollment effect estimates. We note that while UI benefit durations affect enrollment, this policy variation cannot be readily used to estimate the effects of enrollment in our main analysis because benefit durations can directly impact unemployment and labor force participation (\citealp{Rothstein2011_app}).

\subsection{Background and Summary of \texorpdfstring{\citet{BarrTurner2015_app}}{Barr and Turner (2015)}}

Under normal economic conditions, unemployed workers in Ohio (and in most other states) are eligible to receive 26 weeks of UI benefits, which replace 50 percent of past earnings up to a cap (the cap ranges from \$323 to \$524 per week during our sample period, depending on the year and number of dependents a worker has). The duration of benefits may be increased during economic downturns. As shown in Appendix Figure \ref{fig:potdur_overtime}, benefit durations varied substantially over our sample period in Ohio, reaching 99 weeks in 2009 and persisting for a few years afterward.\footnote{The narrow valleys are due to the failure to extend EUC08 legislation before a scheduled policy expiration. However, claimants were retroactively compensated after these lapses.} The degree to which benefits were extended depended on the state unemployment rate and policies in place at a given time (see \citealp{Rothstein2011_app} for details). Using the October Education Supplement of the Current Population Survey, \citet{BarrTurner2015_app} show that these changes in UI policy across states and over time increased unemployed workers' propensity to pursue postsecondary education. We use the same temporal variation to estimate the magnitude of the effect within Ohio.

One major difference between our analysis and \citet{BarrTurner2015_app}'s is our ability to observe the timing of enrollment (and unemployment), which allows us to relate enrollment to the UI policy in place at the time of the enrollment decision.\footnote{\citet{BarrTurner2015_app}'s main regressor is the UI duration in the August of the year enrollment is observed, under the assumption that workers make decisions to enroll at the start of the academic year.} This is important because, as discussed in \citet{Rothstein2011_app}, sudden changes in federal legislation during the Great Recession resulted in changing expectations on benefit duration over the unemployment spell. One way to see this is in Appendix Figure \ref{fig:wksleft_sim}. For this figure, we simulate the number of UI benefit weeks remaining for workers at various quarters of the unemployment spell using only variation in UI extension policies over time, and we plot the averages for different cohorts of claimants.\footnote{This simulation assumes that all workers receive 26 weeks of regular benefits and are continuously unemployed starting on their claim date. The simulation code is lightly adapted from \citet{Rothstein2011_app}.} Before the recession, workers expected 26 weeks of benefits at the beginning of unemployment and would run out of benefits after two quarters. For cohorts laid off in 2008, workers began unemployment under the assumption that they were eligible for 26 weeks of benefits, but as new extensions began in June 2008, workers with relatively long spells of unemployment were eligible for the new extensions. Throughout 2009 and 2010, benefits were continually extended such that although workers used their benefits, the average number of UI benefit weeks remaining declined less quickly than would be expected mechanically with the passage of time---in fact, for the 2009 cohort, the remaining benefit weeks actually increased at times. To the extent that workers base enrollment decisions on expectations of UI remaining available, it is important to consider this policy variation over the course of unemployment.

\subsection{Effect of UI Benefit Duration on Enrollment}

Our analysis uses a 5 percent subset of our UI claims sample, which covers claimants from 2004-2011Q3.\footnote{One difference between our sample and \citet{BarrTurner2015_app}'s is that ours is a sample of \emph{layoffs} while \citet{BarrTurner2015_app}'s is a sample of \emph{unemployment spells}. The latter likely contains disproportionately more long-term unemployed workers.} For each claimant, we have a balanced panel of eight quarters starting in the first quarter after her claim. We are interested in the effect of the expected UI benefit duration on enrollment over the first two years of layoff. We estimate equations of the following form:
\[
E_{it}=\beta P_{it}+p(UR_{t};\rho)+\sum_{k=1}^{8}\delta^{k}D_{it}^{k}+\lambda_{t}+\mathbf{X}_{i}\gamma+\varepsilon_{it}
\]
where $E_{it}$ is the enrollment indicator for worker $i$ in quarter $t$, $P_{it}$ is the UI potential benefit duration in quarter $t$, $D_{it}^{k}$ are a set of indicators denoting the $k$th quarter since layoff, $UR_{t}$ is the state unemployment rate at time $t$, and $\mathbf{X}_{i}$ contain demographic and pre-layoff job characteristics. The main coefficient of interest is $\beta$, which is the effect of the benefit duration on enrollment. Since benefit durations are partially determined by state economic conditions, we include a quadratic function $p(UR_{t};\rho)$ of the unemployment rate, following \citet{BarrTurner2015_app}. We also flexibly control for time since layoff ($\delta^{k}$), year and quarter-in-year effects ($\lambda_{t}$), and worker characteristics including gender, age category, race, whether a worker reports having dependents, past wage quintile, tenure at last employer (four categories), past industry (two-digit NAICS), and past occupation (two-digit SOC).

Appendix Table \ref{tab:ui_extensions} presents our regression results. The first column shows that a ten-week increase in potential benefit duration raises the probability of enrollment by 0.15 percentage points, or a 10 percent increase. Although this effect is smaller than the point estimate found in \citet{BarrTurner2015_app}, it does fall inside their 95 percent confidence interval. As noted above, since potential durations were changing throughout the unemployment spell during the recession, the benefit duration presumed at the beginning of a worker's unemployment spell differs from the benefit duration actually experienced. In the second column of Appendix Table \ref{tab:ui_extensions}, we show that the potential duration at the beginning of the spell has no impact on enrollment beyond the impact of the potential duration at the time of enrollment. The third column estimates the effect using a fixed effects model, utilizing only the variation in potential benefit expectations over time for each spell, and finds that the effect of potential duration is slightly smaller, though still statistically significant. Finally, the fourth column restricts the sample to each worker's first unemployment spell after a UI claim, excluding quarters after the worker has returned to full-quarter employment.\footnote{We define the first unemployment spell as running from the quarter in which the UI claim begins through the last quarter before the worker is first observed working all 13 weeks of a quarter. Quarters with partial employment (fewer than 13 weeks worked) are therefore counted as part of the spell.} Although the probability of enrollment is higher in this sample, the estimate is similar in percentage terms: a 10-week increase in potential benefit duration increases enrollment by 11 percent.

\subsection{Policy Implications}

These estimates indicate that a 10-week increase in UI potential durations induces approximately 1,200 more workers to enroll annually. Assuming that our estimated earnings gain of \$348 per quarter in the third and fourth year represents an estimate of the long-run effects, the increased enrollment would imply about \$1.7 million per year in earnings gain starting in the fifth year after enrollment (since the lock-in effects in the first two years roughly equal the positive earnings effects in the third and fourth years).

This finding also suggests that there is an externality associated with extending UI. As discussed in \citet{SchmiedervonWachter2016_app} and \citet{Lee_etal_2019_suffstats_app}, the overall impact of a policy on the government budget is a critical parameter in optimal policy-making. Therefore, estimates of the implied increases in tax revenues and government expenditures related to financial aid or tuition subsidies should be accounted for in UI policy analysis.

\section*{Appendix References}

\begin{singlespace}

\begin{btSect}[jpe]{opportunity_appendix}
\btPrintCited
\end{btSect}
\end{singlespace}

\clearpage{}

\renewcommand{\thetable}{A.\arabic{table}}

\renewcommand{\thefigure}{A.\arabic{figure}}

\setcounter{figure}{0}

\setcounter{table}{0} 

\begin{sidewaysfigure}[H]
\caption{Distributions of Log Odds Ratio}

\label{fig:overlap}
\begin{centering}
\includegraphics[scale=1.4]{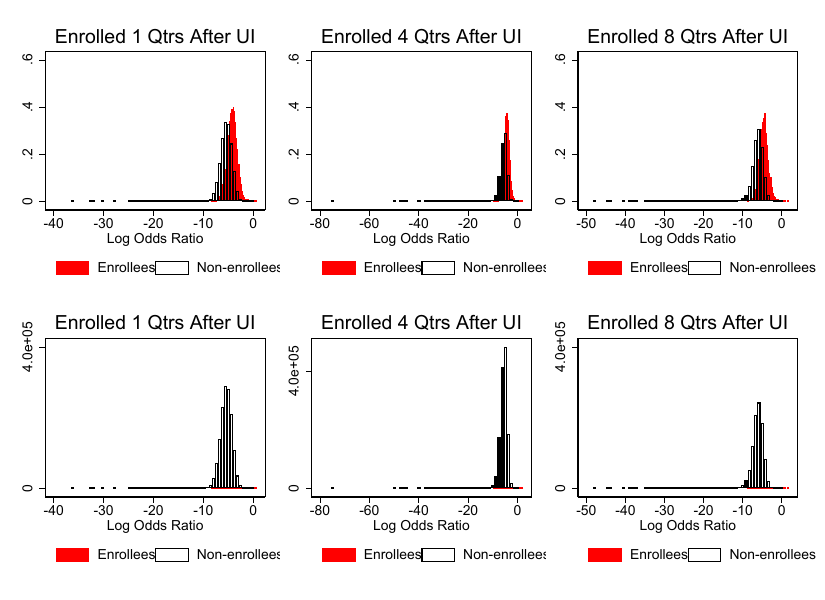}
\par\end{centering}
\centering{}%
\begin{minipage}[t]{0.7\columnwidth}%
Notes: These figures show the distributions of the estimated propensity to enroll (expressed in log-odds ratios), for UI claimants who enroll in the first, fourth, and eighth quarters after layoff, and non-enrollees. The first row of figures overlays the distribution for enrollees conditional on being an enrollee (in red), and the distribution for non-enrollees conditional on being a non-enrollee (in white). The second row shows the distribution with both enrollees and non-enrollees pooled. (Enrollees are a small proportion of the overall claimant population so while they are present, they may not be visually perceptible.) There are $N=$1,432,293; 1,348,489; 959,731 UI claims in the graphs in each column, corresponding to 1,045,644; 995,685; and 763,280 unique individuals.
\end{minipage}
\end{sidewaysfigure}

\clearpage{}
\begin{figure}[H]
\caption{Earnings of Enrollees and Non-enrollees}
\label{fig:unmatched}

\smallskip{}

\begin{centering}
\includegraphics{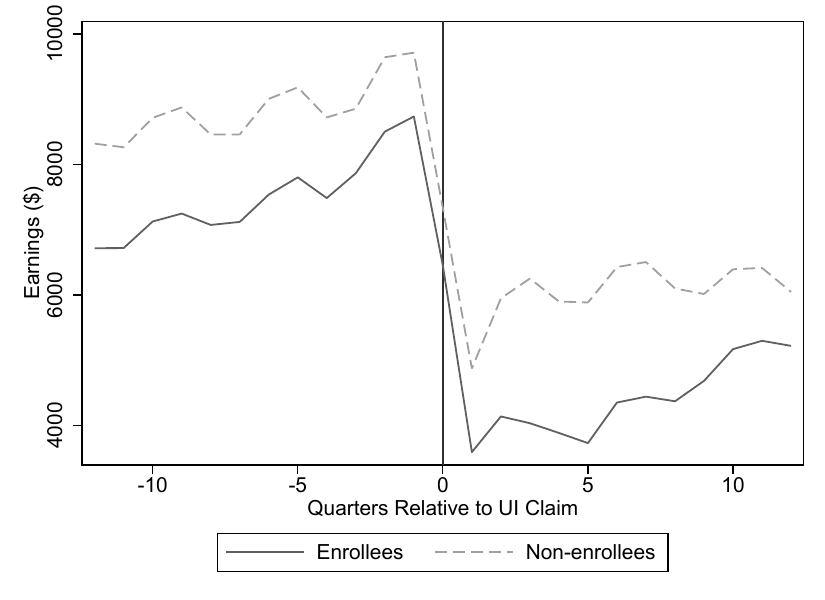}
\par\end{centering}
\centering{}%
\begin{minipage}[t]{0.8\columnwidth}%
Notes: This figure plots the average quarterly earnings of enrollee and non-enrollee UI claimants for five percent of our analysis sample. The vertical line denotes the UI claim quarter. $N=93,528$ UI claims (corresponding to 91,043 unique individuals).
\end{minipage}
\end{figure}
\clearpage{}
\begin{figure}[H]
\caption{Earnings of Enrollees and Matched Non-enrollees Using Alternative Matching Specifications}

\label{fig:partialmatch}
\begin{centering}
\smallskip{}
\par\end{centering}
\begin{centering}
(A) Matched on One Quarter of Earnings Before Layoff and Earnings Between Layoff and Enrollment
\par\end{centering}
\begin{centering}
\includegraphics[scale=0.5]{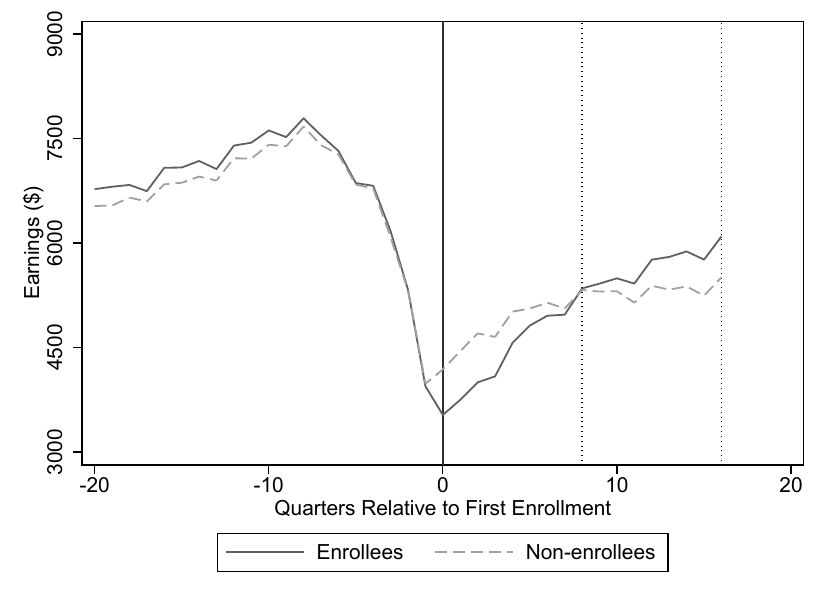}
\par\end{centering}
\begin{centering}
(B) Matched on Four Quarters of Earnings Before Layoff and Earnings Between Layoff and Enrollment
\par\end{centering}
\begin{centering}
\includegraphics[scale=0.5]{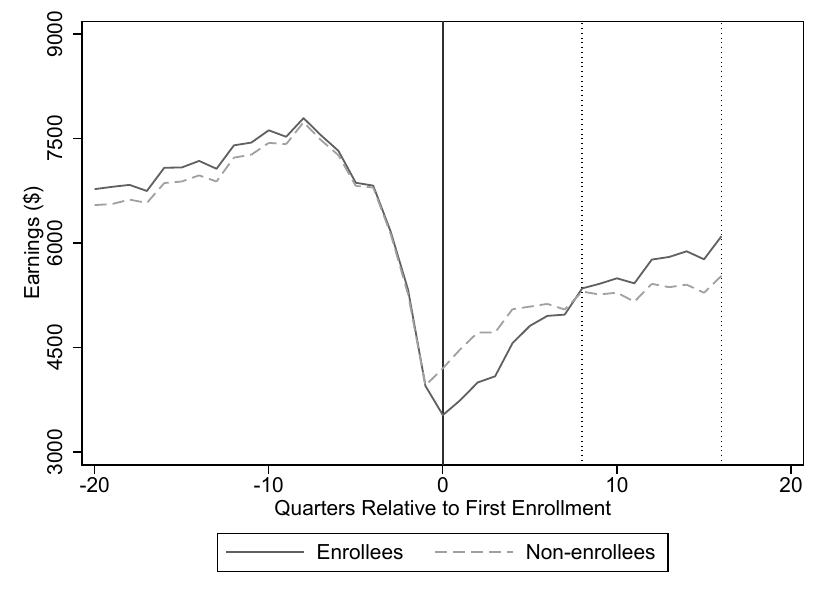}
\par\end{centering}
\begin{centering}
(C) Matched on Four Quarters of Earnings Before Layoff
\par\end{centering}
\begin{centering}
\includegraphics[scale=0.5]{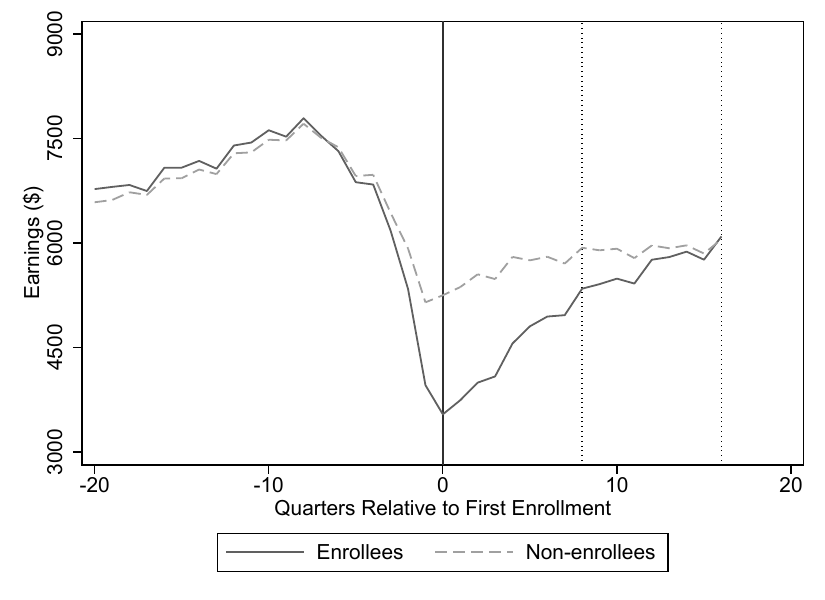}
\par\end{centering}
\centering{}%
\begin{minipage}[t]{0.8\columnwidth}%
Notes: These figures plot the average quarterly earnings of enrollee and matched non-enrollee UI claimants, where the matching is done using all demographic variables described in Section \ref{subsec:Matching-Specification} and the following alternative sets of earnings variables: one quarter of earnings before layoff and earnings between layoff and enrollment (Panel A), four quarters of earnings before layoff and earnings between layoff and enrollment (Panel B), and four quarters of earnings before layoff (Panel C). The solid vertical line denotes the quarter of first enrollment, and the vertical dashed lines denote eight and 16 quarters after first enrollment. These figures contain $N=$142,908; 142,460; and 141,958 UI claims, respectively, corresponding to 137,360; 136,741; and 136,412 unique individuals.
\end{minipage}
\end{figure}

\clearpage{}
\begin{figure}[H]
\caption{Earnings of Enrollees and Matched Non-enrollees Using Alternative Specifications}

\label{fig:zeroearn_noborder}
\begin{centering}
(A) Including Zero Earnings Indicators in Matching Specification
\par\end{centering}
\begin{centering}
\includegraphics[scale=0.85]{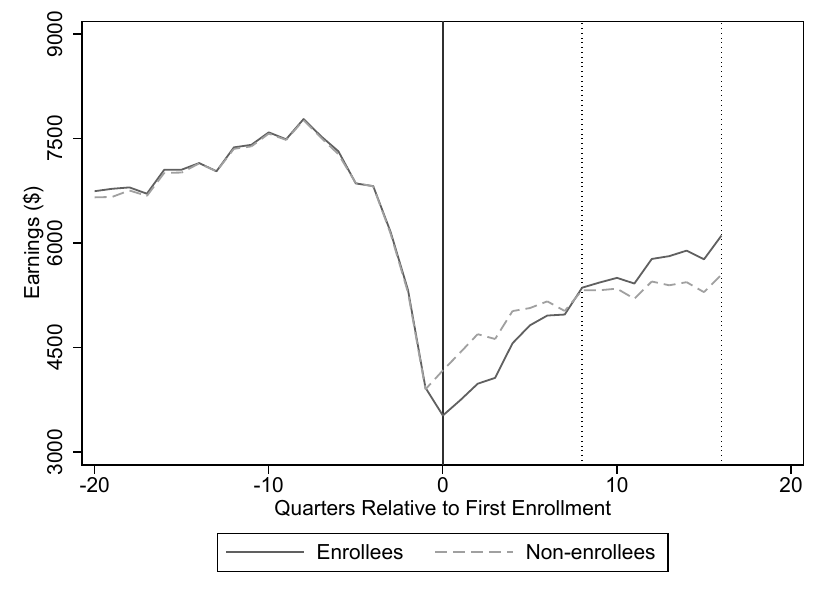}
\par\end{centering}
\begin{centering}
(B) Excluding Border Counties
\par\end{centering}
\begin{centering}
\includegraphics[scale=0.85]{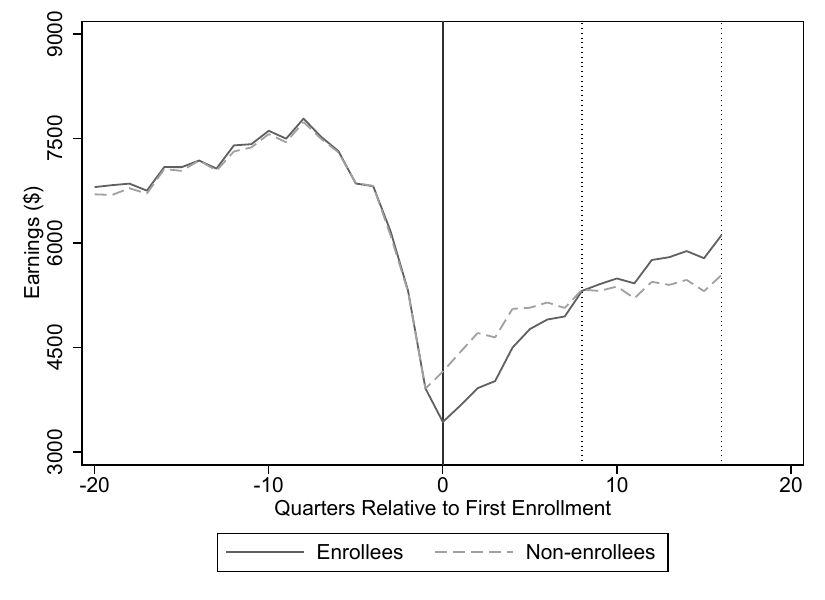}
\par\end{centering}
\centering{}%
\begin{minipage}[t]{0.8\columnwidth}%
Notes: These figures plot the average quarterly earnings of enrollee and matched non-enrollee UI claimants, where the matching is done using all earnings and demographic variables described in Section \ref{subsec:Matching-Specification} and a) adding indicators for zero earnings in each quarter (Panel A), and b) excluding counties on the Ohio border (Panel B). The solid vertical line denotes the quarter of first enrollment, and the vertical dashed lines denote eight and 16 quarters after first enrollment. These figures contain $N=$139,550 and 109,098 UI claims, respectively, corresponding to 133,726 and 104,186 unique individuals.
\end{minipage}
\end{figure}

\clearpage{}
\begin{figure}[H]
\caption{Estimated Enrollment Effects, By Number of Neighbors}

\label{fig:neighbors_analysis}
\begin{centering}
(A) Estimated Enrollment Effect
\par\end{centering}
\begin{centering}
\includegraphics[scale=0.85]{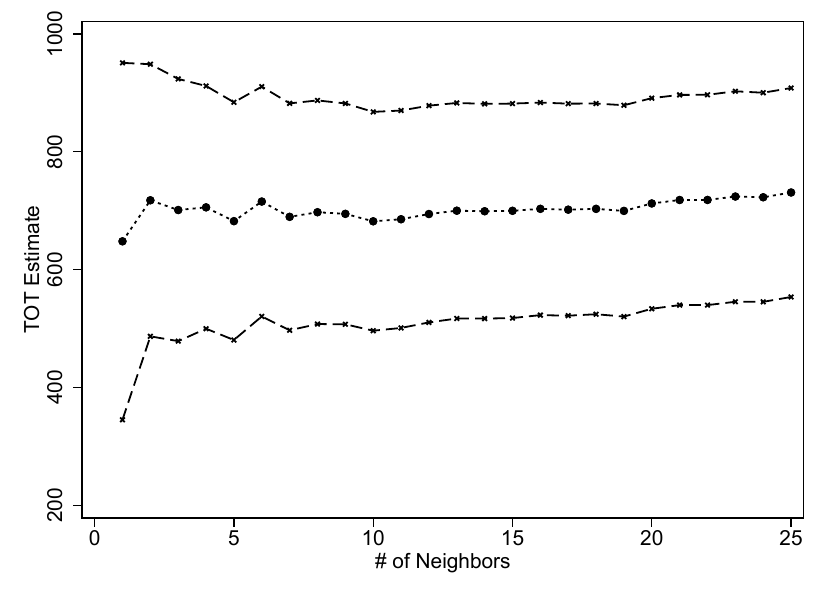}
\par\end{centering}
\begin{centering}
(B) Reduction in Variance
\par\end{centering}
\begin{centering}
\includegraphics[scale=0.85]{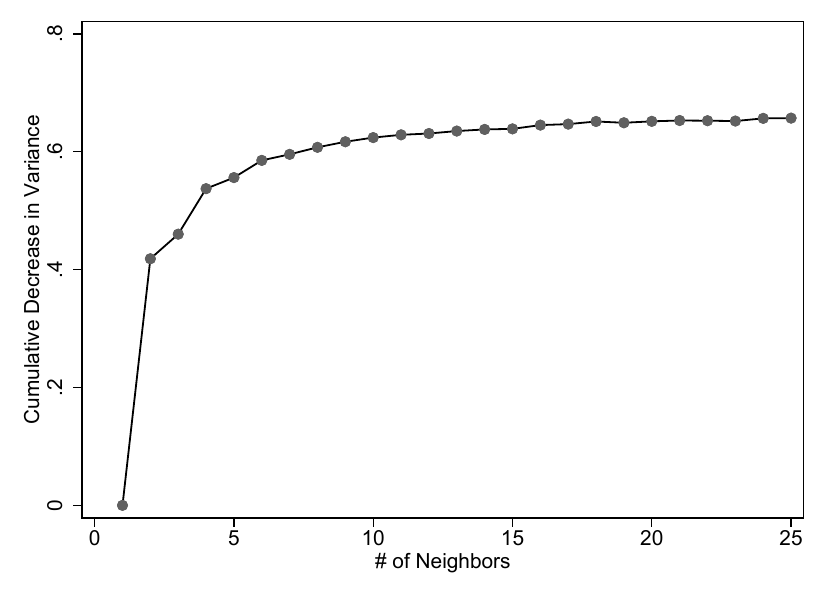}
\par\end{centering}
\centering{}%
\begin{minipage}[t]{0.8\columnwidth}%
Notes: Panel A shows how the estimated enrollment effect (3-4 years after enrolling) varies by the number of matched neighbors for men who did not previously work in manufacturing and who filed a UI claim in the first quarter of 2009. Panel B shows the proportion reduction in the variance of the estimated effect, relative to using one neighbor to match. $N=54,685$ UI claims, corresponding to 54,685 unique individuals.
\end{minipage}
\end{figure}

\clearpage{}
\begin{figure}[H]
\caption{Employment of Enrollees and Matched Non-enrollees}

\label{fig:ext_margin}
\begin{centering}
\smallskip{}
\par\end{centering}
\begin{centering}
(A) Probability of Having Any Positive Earnings
\par\end{centering}
\begin{centering}
\includegraphics[scale=0.85]{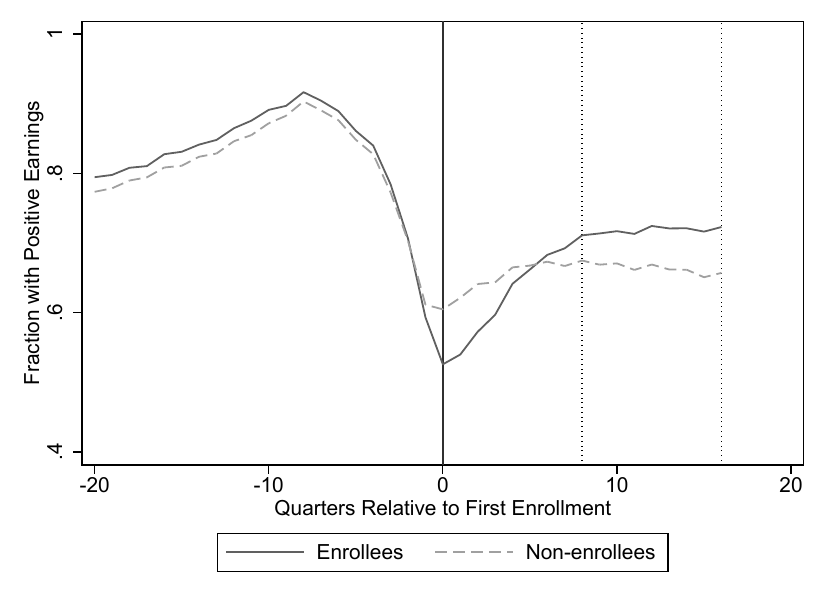}
\par\end{centering}
\begin{centering}
(B) Weeks Worked Per Quarter
\par\end{centering}
\begin{centering}
\includegraphics[scale=0.85]{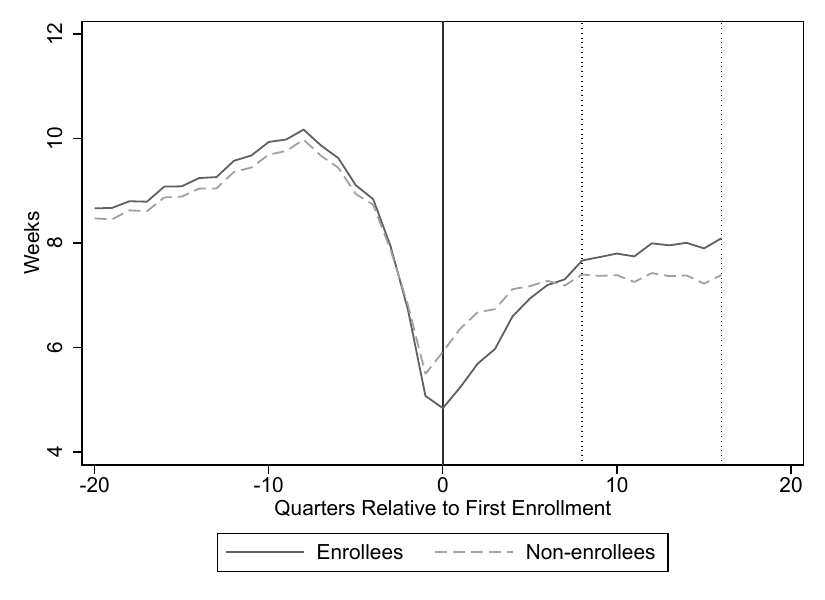}
\par\end{centering}
\centering{}%
\begin{minipage}[t]{0.8\columnwidth}%
Notes: These figures plot the fraction of enrollee and matched non-enrollee UI claimants with positive earnings in each quarter (Panel A) and the quarterly average number of weeks worked (Panel B). The solid vertical line denotes the quarter of first enrollment, and the vertical dashed lines denote eight and 16 quarters after first enrollment. $N=141,758$, corresponding to 136,074 unique individuals.
\end{minipage}
\end{figure}

\clearpage{}
\begin{figure}[H]
\caption{Distributions of Enrollee and Matched Non-enrollee Earnings}
\label{fig:distributional_effects}
\begin{centering}
\includegraphics[scale=0.5]{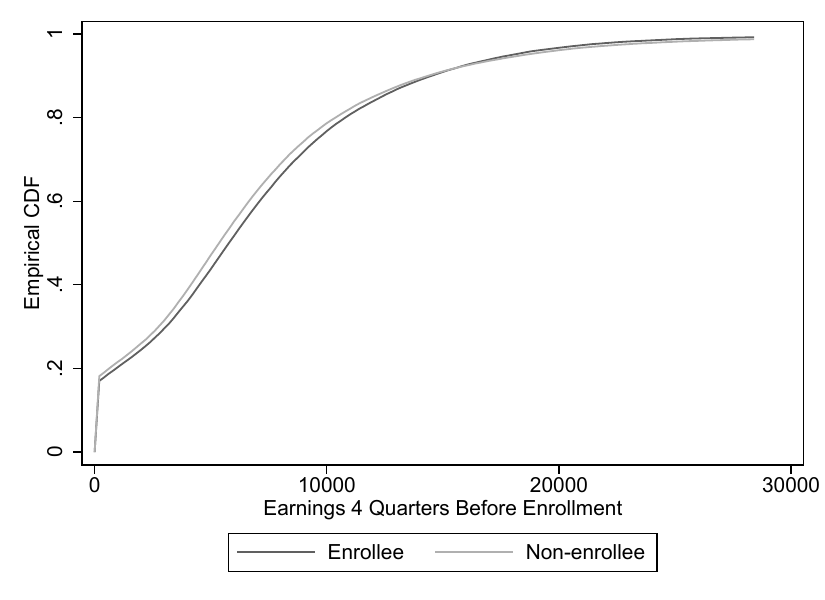}\includegraphics[scale=0.5]{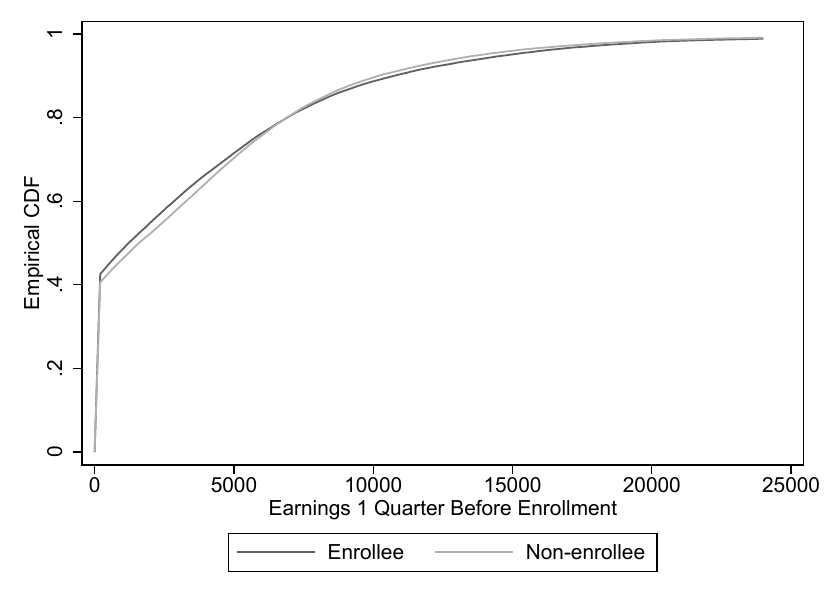}
\par\end{centering}
\begin{centering}
\includegraphics[scale=0.5]{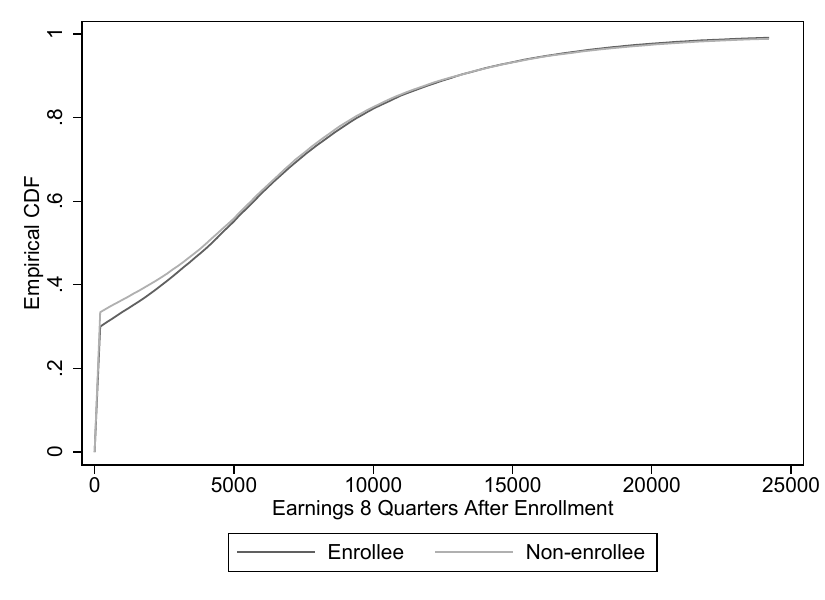}\includegraphics[scale=0.5]{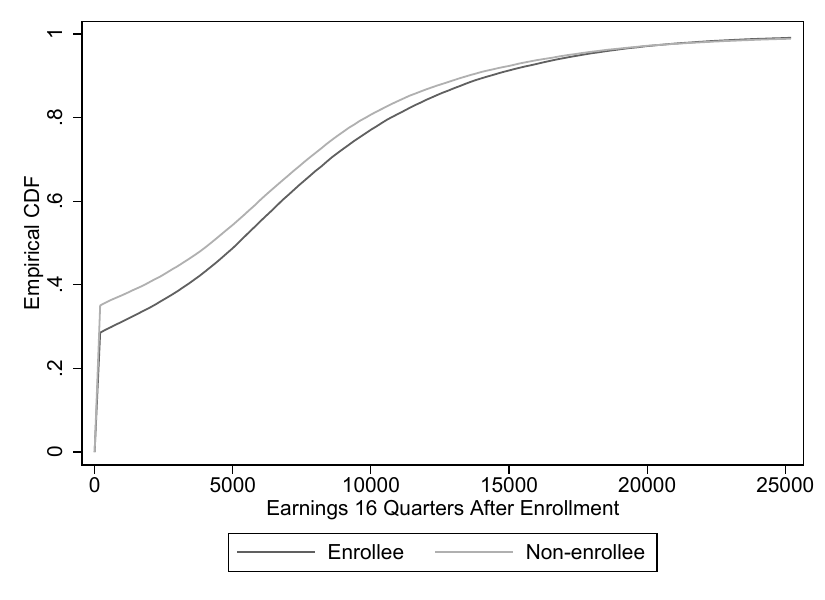}
\par\end{centering}
\centering{}%
\begin{minipage}[t]{0.8\columnwidth}%
Notes: These figures show the empirical cumulative distribution functions of earnings for enrollees and matched non-enrollees four quarters before enrollment, one quarter before enrollment, eight quarters after enrollment, and 16 quarters after enrollment. $N=141,758$, corresponding to 136,074 unique individuals, for each graph.
\end{minipage}
\end{figure}

\clearpage{}
\begin{figure}[H]
\caption{Earnings of Enrollees and Matched Non-enrollees, By Enrollment Timing}
\label{fig:by_enroll_timing}
\begin{centering}
\includegraphics[scale=0.43]{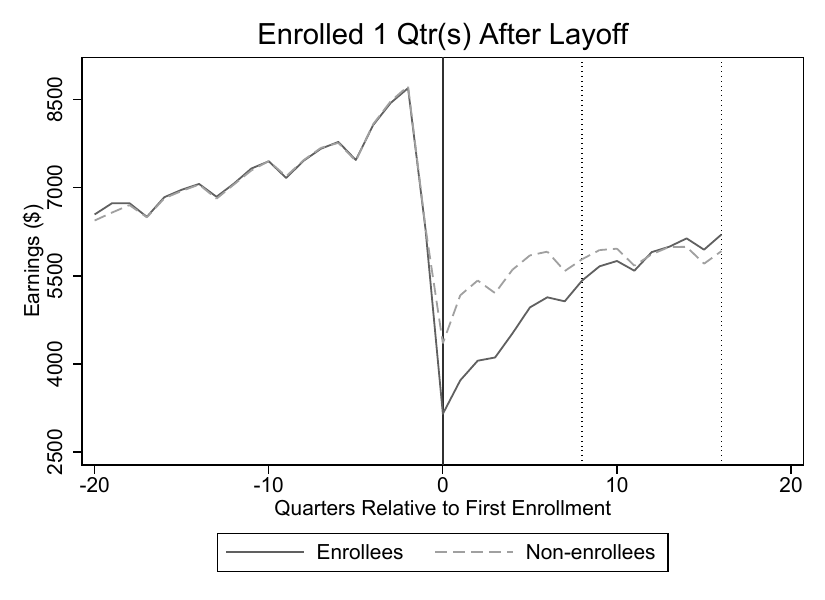}\includegraphics[scale=0.43]{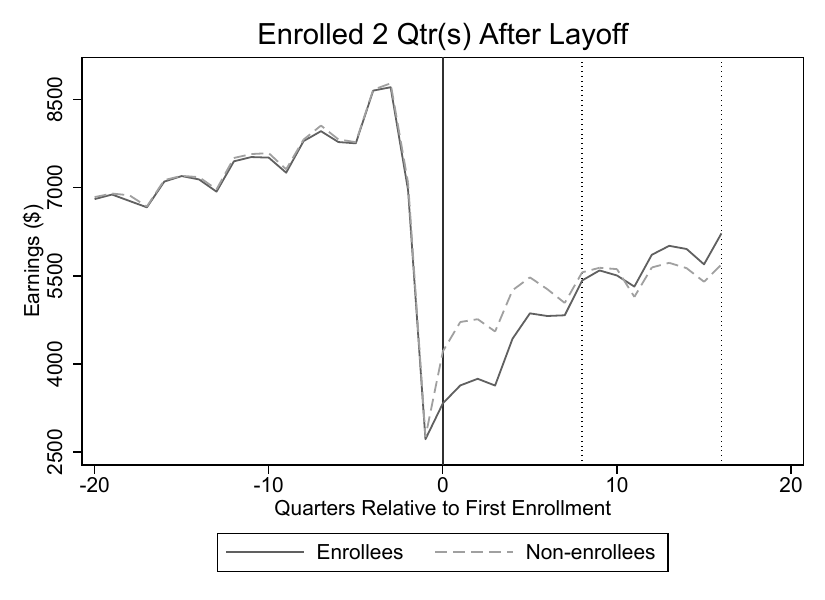}
\par\end{centering}
\begin{centering}
\includegraphics[scale=0.43]{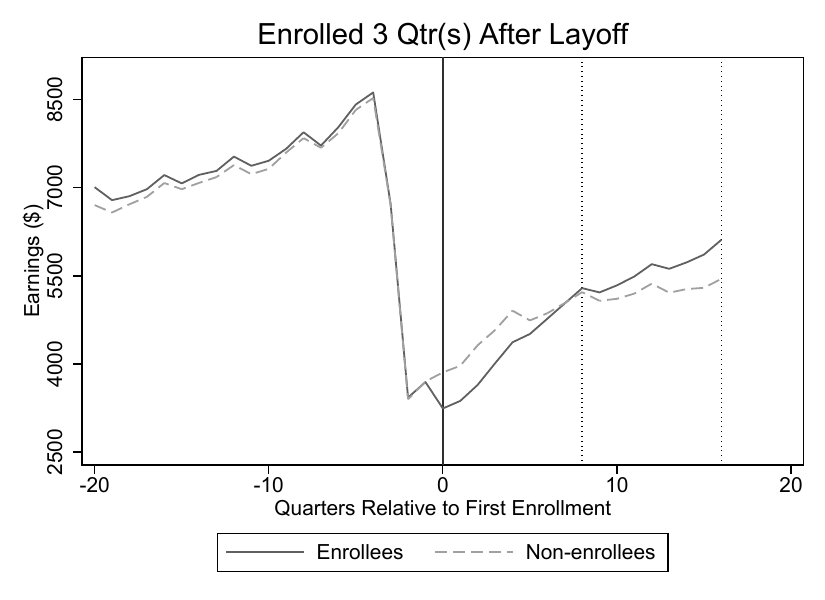}\includegraphics[scale=0.43]{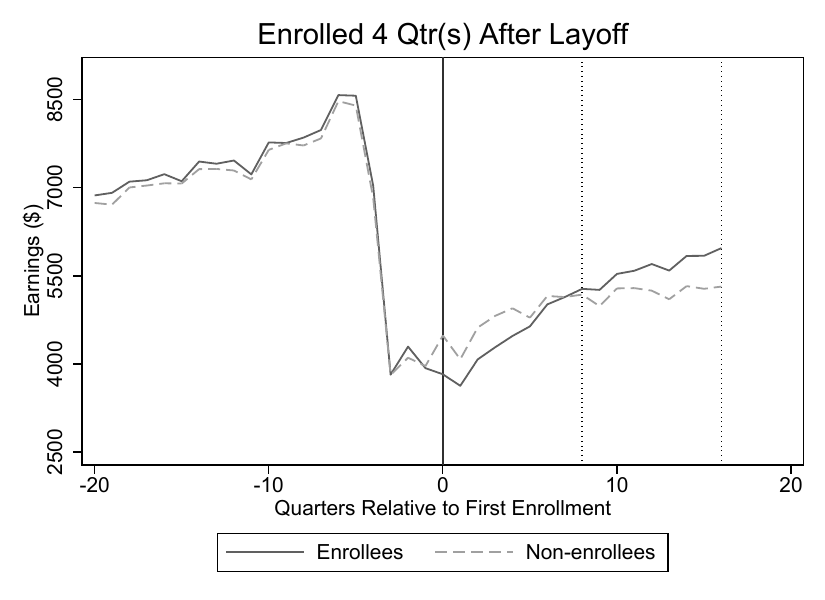}
\par\end{centering}
\begin{centering}
\includegraphics[scale=0.43]{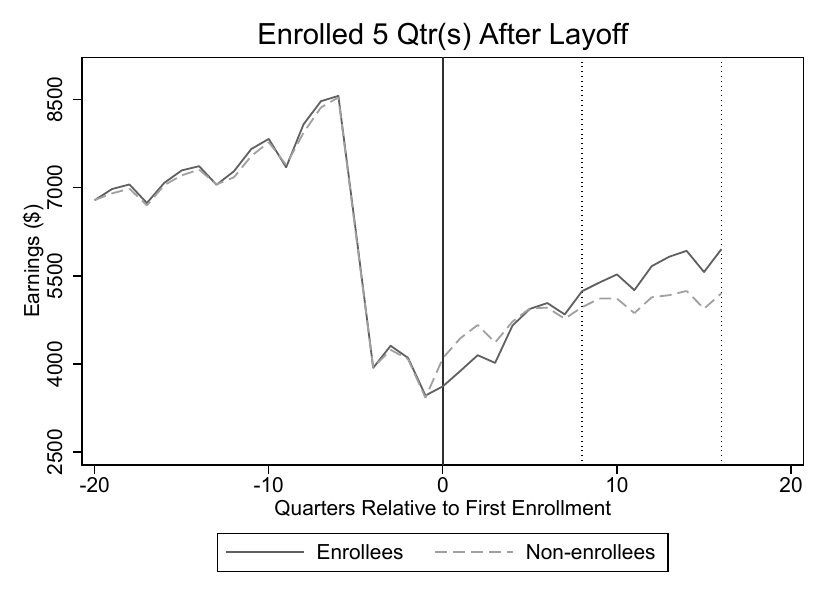}\includegraphics[scale=0.43]{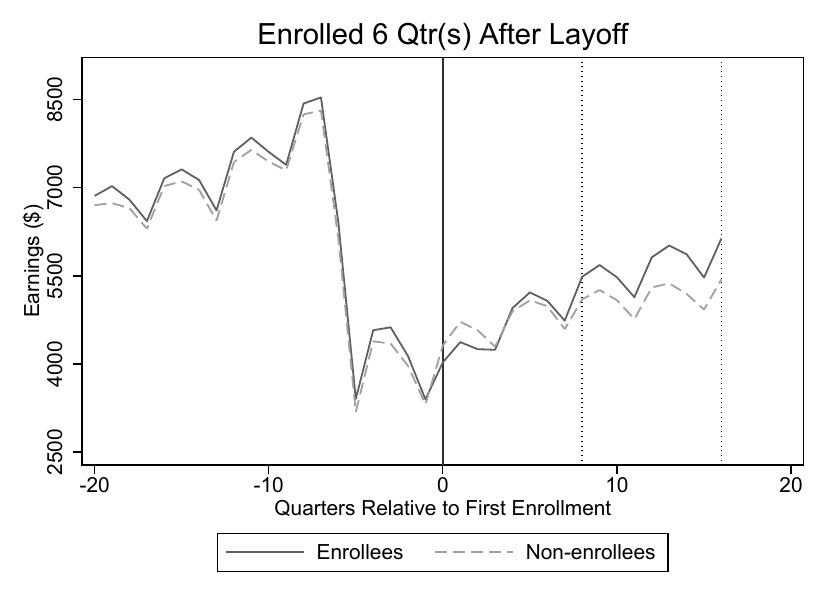}
\par\end{centering}
\begin{centering}
\includegraphics[scale=0.43]{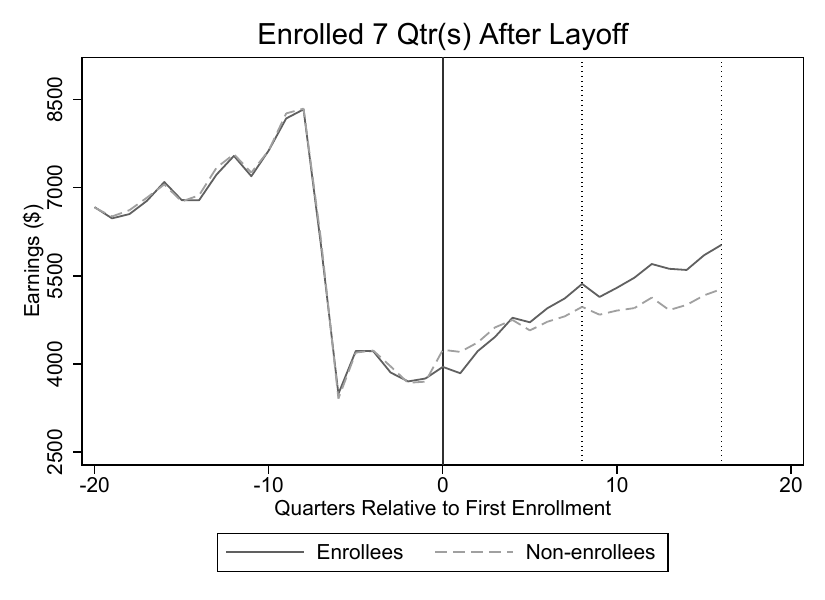}\includegraphics[scale=0.43]{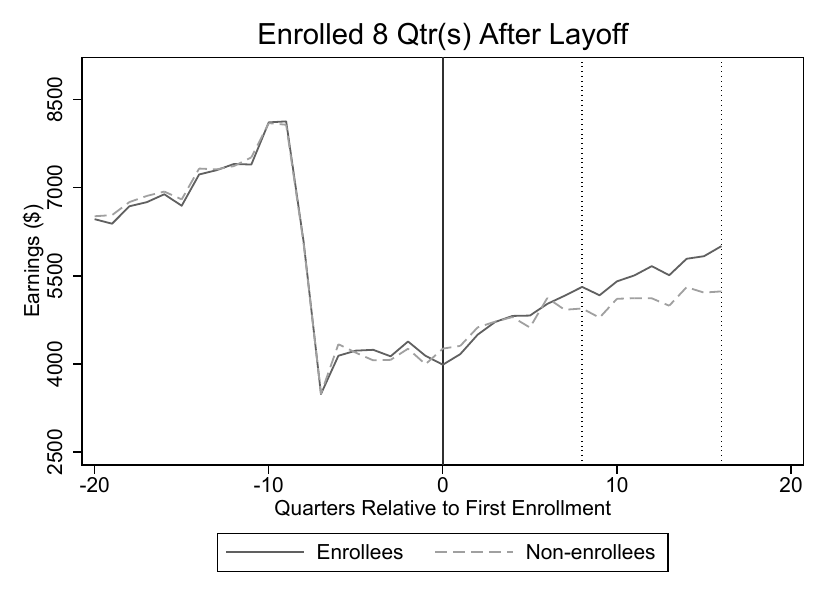}
\par\end{centering}
\centering{}%
\begin{minipage}[t]{0.7\columnwidth}%
Notes: Each graph shows the average quarterly earnings of UI claimants who enroll a certain number of quarters after filing a UI claim and their matched non-enrollees. The solid vertical lines denote the quarter of first enrollment, and the vertical dashed lines denote eight and 16 quarters after first enrollment. There are $N=$23,868; 27,126; 23,472; 18,976; 15,190; 13,154; 11,032; 8,940 UI claims in each graph, corresponding to 23,579; 26,750; 23,184; 18,769; 15,024; 13,024; 10,919; 8,866 unique individuals.
\end{minipage}
\end{figure}

\clearpage{}
\begin{figure}[H]
\caption{Earnings of Enrollees and Matched Non-enrollees, By Year of Job Loss}
\label{fig:by_year}
\begin{centering}
\includegraphics[scale=0.43]{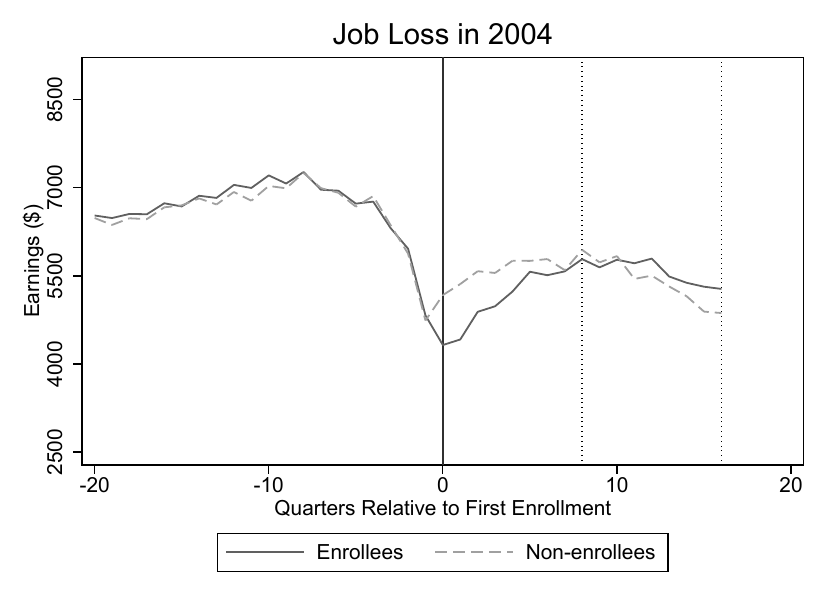}\includegraphics[scale=0.43]{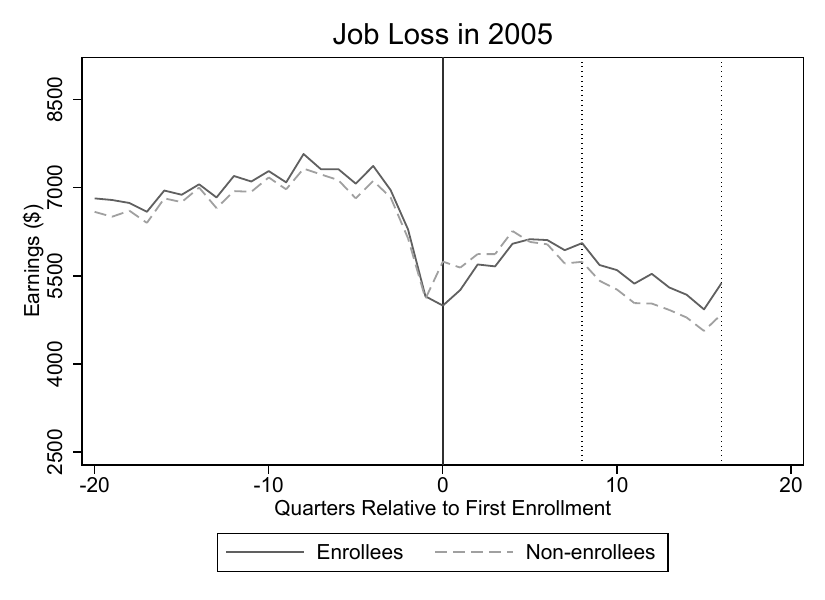}
\par\end{centering}
\begin{centering}
\includegraphics[scale=0.43]{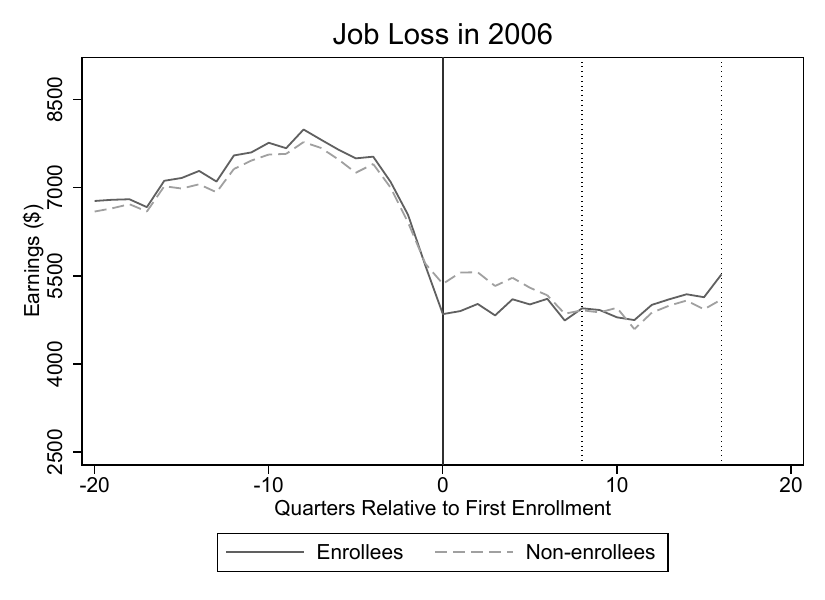}\includegraphics[scale=0.43]{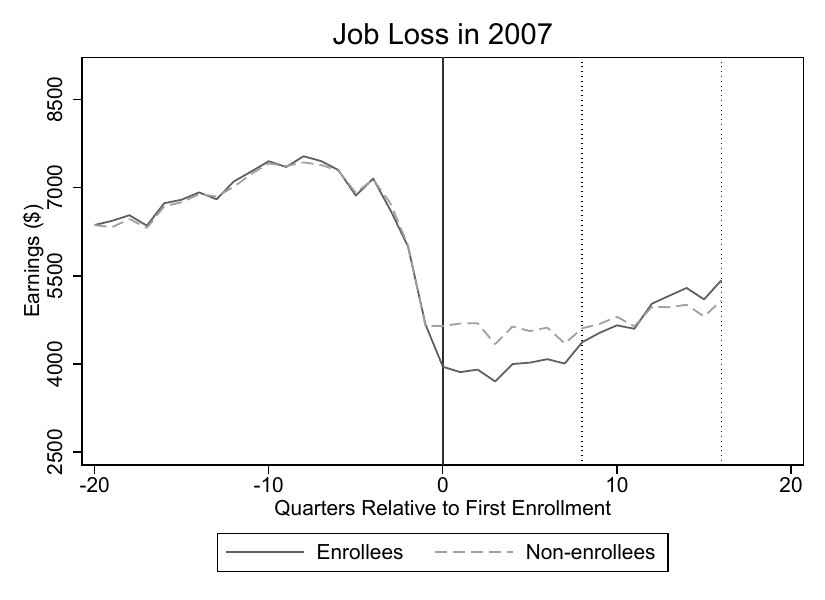}
\par\end{centering}
\begin{centering}
\includegraphics[scale=0.43]{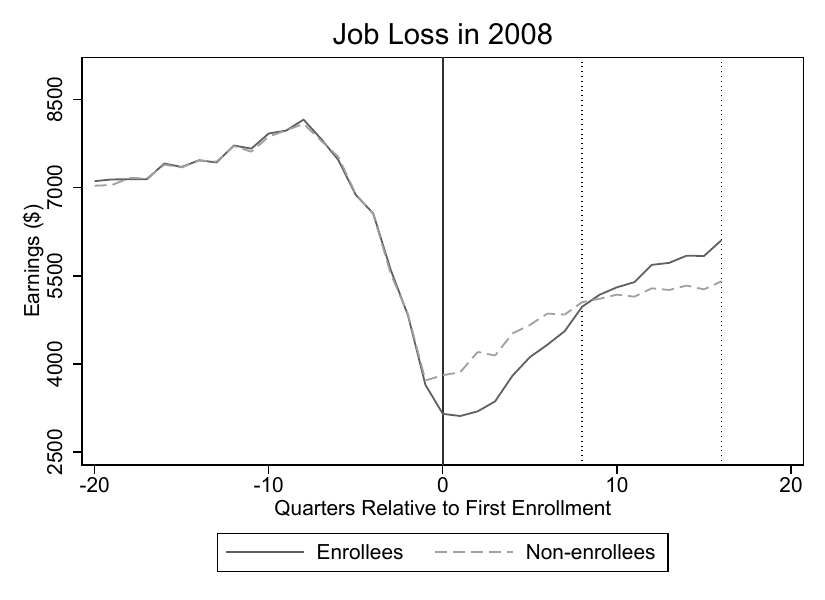}\includegraphics[scale=0.43]{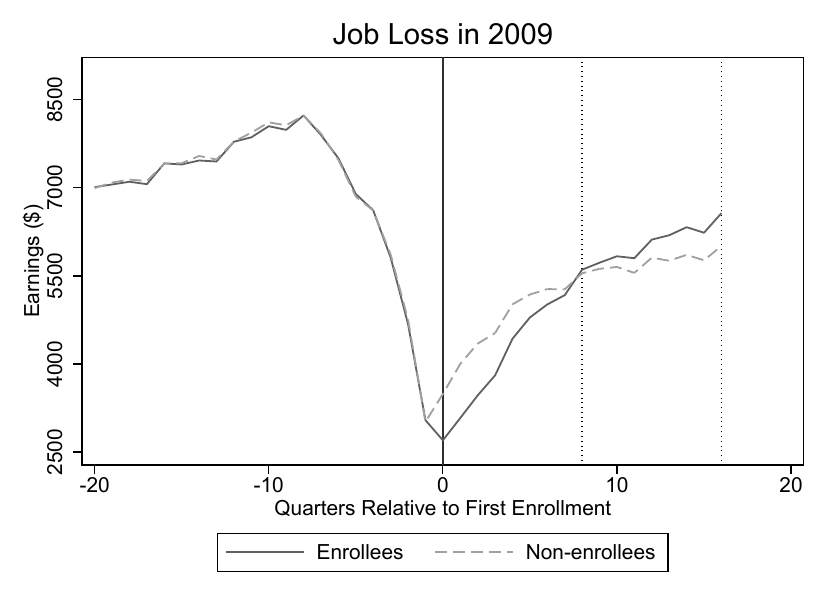}
\par\end{centering}
\begin{centering}
\includegraphics[scale=0.43]{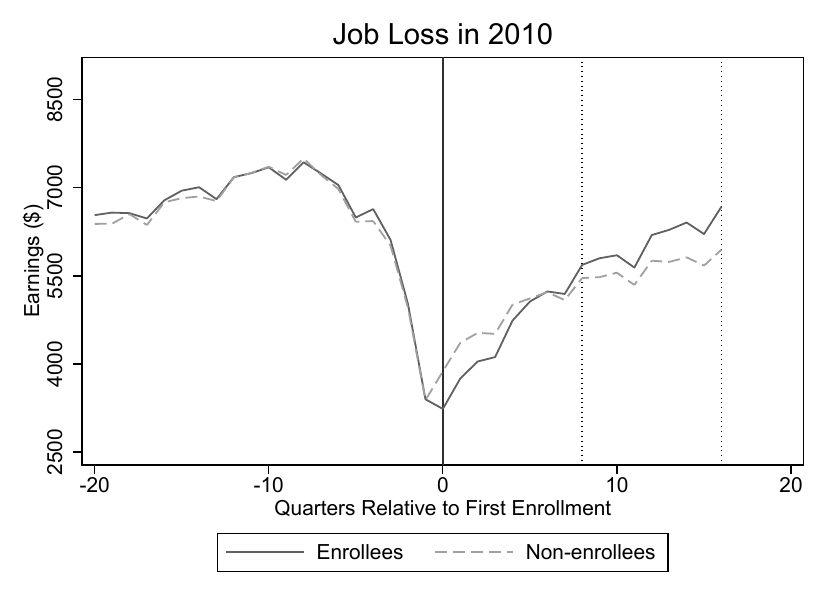}\includegraphics[scale=0.43]{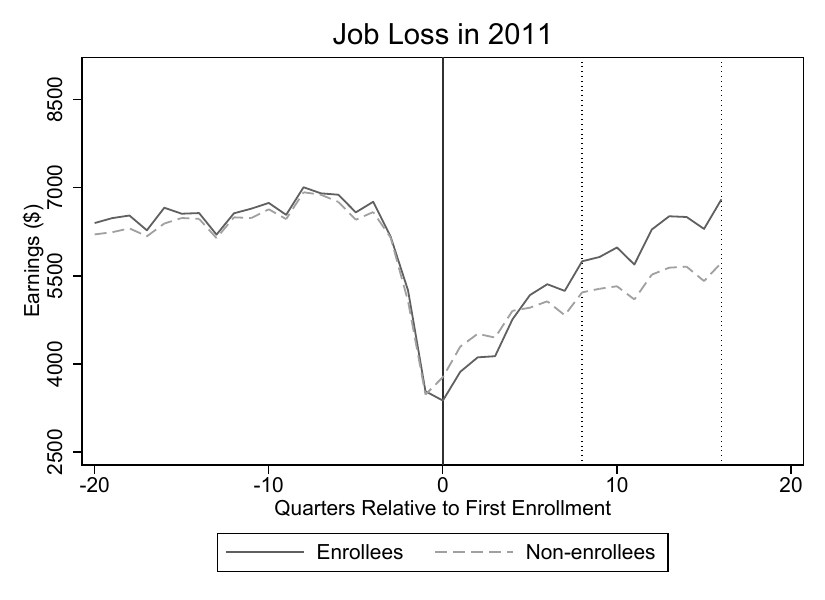}
\par\end{centering}
\centering{}%
\begin{minipage}[t]{0.7\columnwidth}%
Notes: Each graph shows the average quarterly earnings of enrollees and matched non-enrollees who filed a UI claim in a specific year. The solid vertical lines denote the quarter of first enrollment, and the vertical dashed lines denote eight and 16 quarters after first enrollment. There are $N=$12,034; 9,264; 13,494; 14,960; 27,794; 34,216; 19,246; 10,750 UI claims in each graph, corresponding to 11,800; 9,096; 13,250; 14,632; 27,132; 33,358; 18,733; 10,514 unique individuals.
\end{minipage}
\end{figure}
\clearpage{}
\begin{figure}[H]
\caption{Earnings of Enrollees and Matched Non-enrollees, By Gender}
\label{fig:by_gender}
\begin{centering}
\smallskip{}
\par\end{centering}
\begin{centering}
(A) Women
\par\end{centering}
\begin{centering}
\includegraphics[scale=0.85]{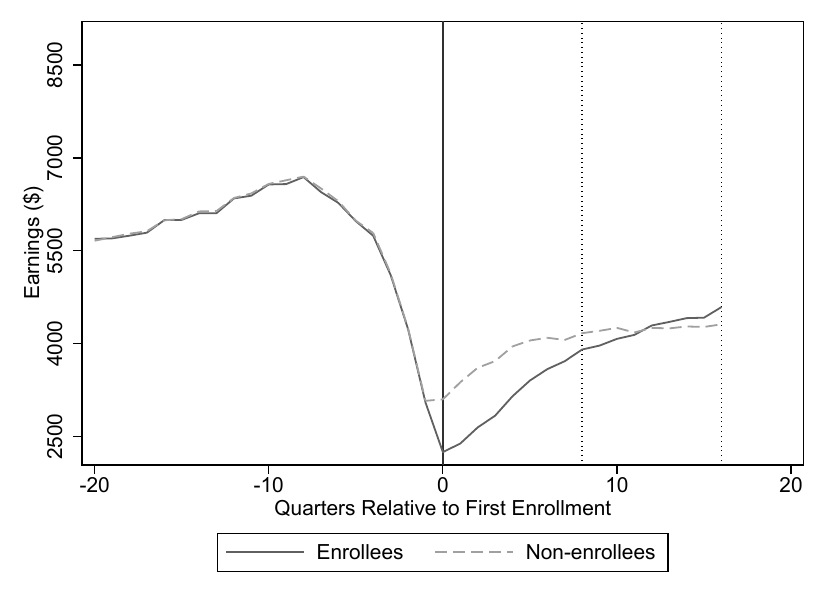}
\par\end{centering}
\begin{centering}
(B) Men
\par\end{centering}
\begin{centering}
\includegraphics[scale=0.85]{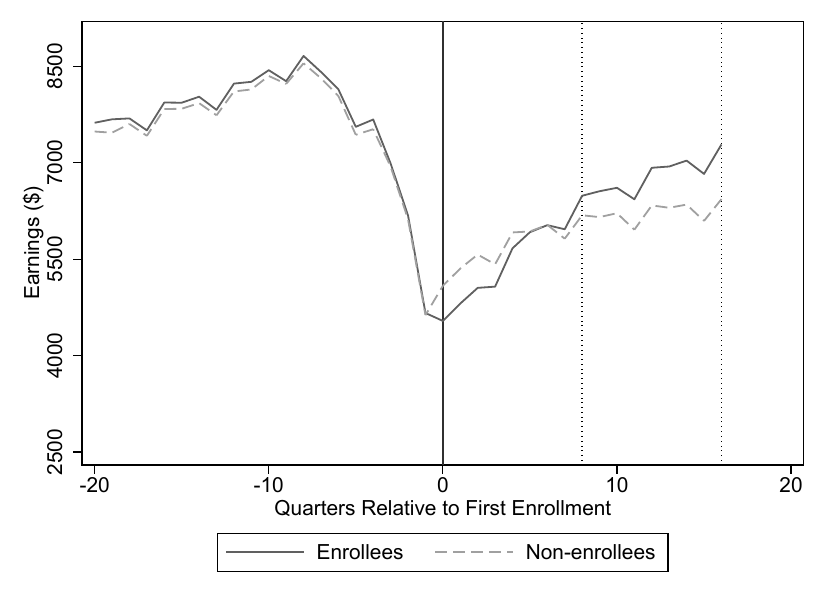}
\par\end{centering}
\centering{}%
\begin{minipage}[t]{0.8\columnwidth}%
Notes: The upper (lower) graph shows the average quarterly earnings of female (male) enrollee and matched non-enrollee UI claimants. The solid vertical lines denote the quarter of first enrollment, and the vertical dashed lines denote eight and 16 quarters after first enrollment. There are $N=$62,096 (79,662) UI claims in the upper (lower) graph, corresponding to 59,260 (76,814) unique individuals.
\end{minipage}
\end{figure}
\clearpage{}
\begin{figure}[H]
\caption{Earnings of Enrollees and Matched Non-enrollees, Manufacturing versus Non-manufacturing}
\label{fig:by_industry}
\begin{centering}
\smallskip{}
\par\end{centering}
\begin{centering}
(A) Manufacturing
\par\end{centering}
\begin{centering}
\includegraphics[scale=0.85]{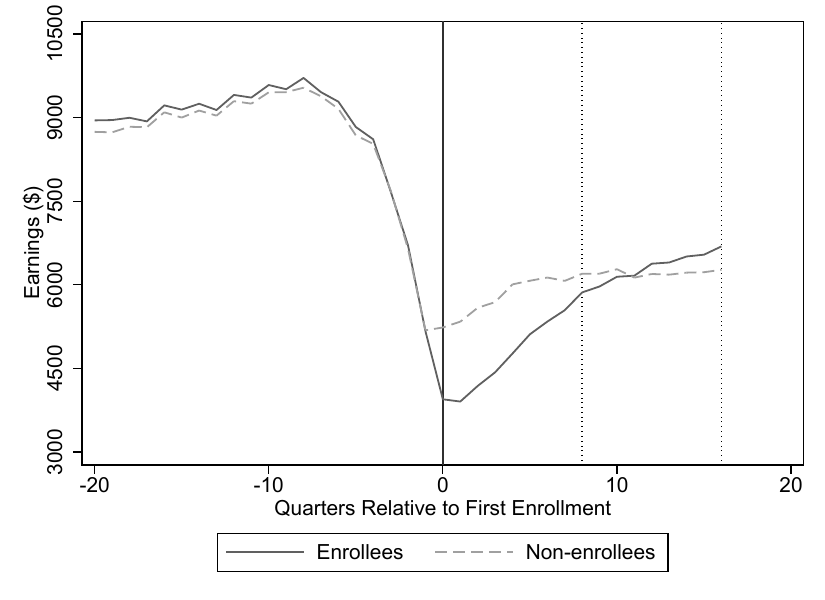}
\par\end{centering}
\begin{centering}
(B) Non-manufacturing
\par\end{centering}
\begin{centering}
\includegraphics[scale=0.85]{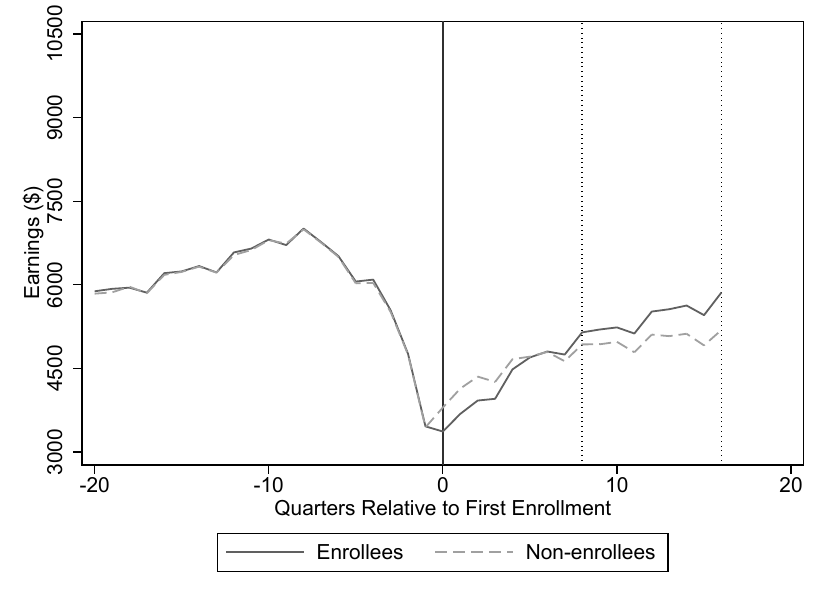}
\par\end{centering}
\centering{}%
\begin{minipage}[t]{0.8\columnwidth}%
Notes: The upper (lower) graph shows the average quarterly earnings of enrollees who previously worked in a manufacturing (non-manufacturing) sector and matched non-enrollee UI claimants. The solid vertical lines denote the quarter of first enrollment, and the vertical dashed lines denote eight and 16 quarters after first enrollment. There are $N=$41,252 (100,506) UI claims in the upper (lower) graph, corresponding to 39,514 (96,902) unique individuals.
\end{minipage}
\end{figure}
\clearpage{}
\begin{figure}[H]
\caption{Earnings of Enrollees and Matched Non-enrollees, By Age}
\label{fig:subgroups_age}
\begin{centering}
\smallskip{}
\par\end{centering}
\begin{centering}
(A) Under Age 40
\par\end{centering}
\begin{centering}
\includegraphics[scale=0.85]{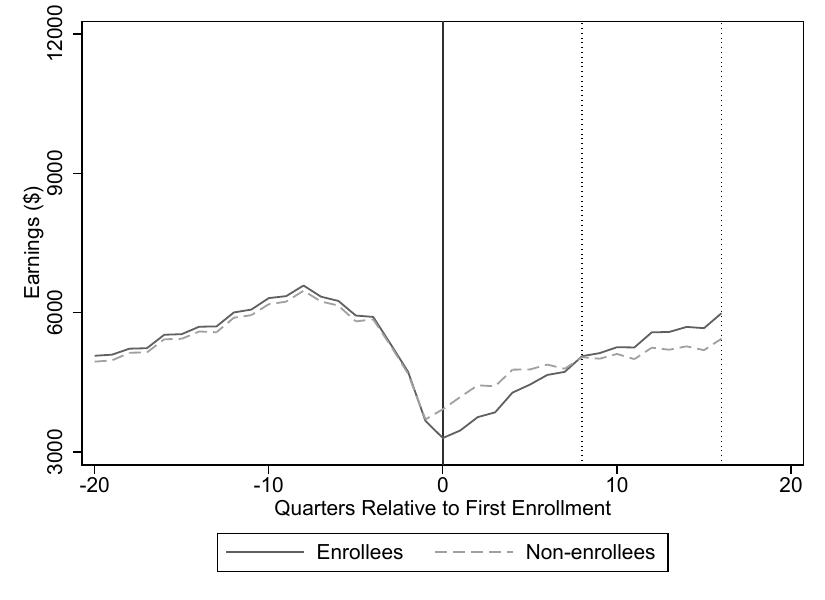}
\par\end{centering}
\begin{centering}
(B) Age 40 or Over
\par\end{centering}
\begin{centering}
\includegraphics[scale=0.85]{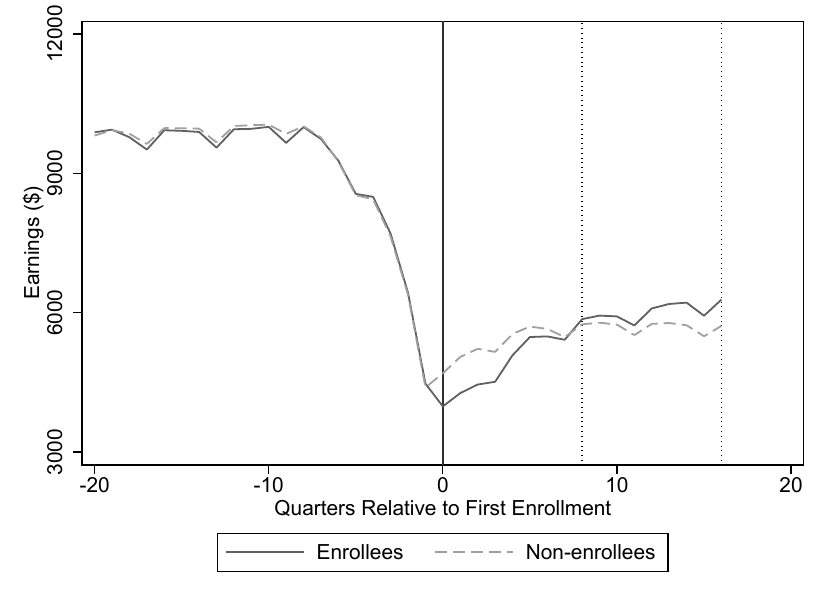}
\par\end{centering}
\centering{}%
\begin{minipage}[t]{0.8\columnwidth}%
Notes: The upper (lower) graph shows the average quarterly earnings of enrollees under age 40 (age 40 or over) and matched non-enrollee UI claimants. The solid vertical lines denote the quarter of first enrollment, and the vertical dashed lines denote eight and 16 quarters after first enrollment. There are $N=$92,712 (50,728) UI claims in the upper (lower) graph, corresponding to 88,477 (49,357) unique individuals.
\end{minipage}
\end{figure}

\clearpage{}
\begin{figure}[H]
\caption{Earnings of Enrollees and Matched Non-enrollees, By Tenure}
\label{fig:subgroups_tenure}
\begin{centering}
\smallskip{}
\par\end{centering}
\begin{centering}
(A) Job Tenure of One Year or Less
\par\end{centering}
\begin{centering}
\includegraphics[scale=0.5]{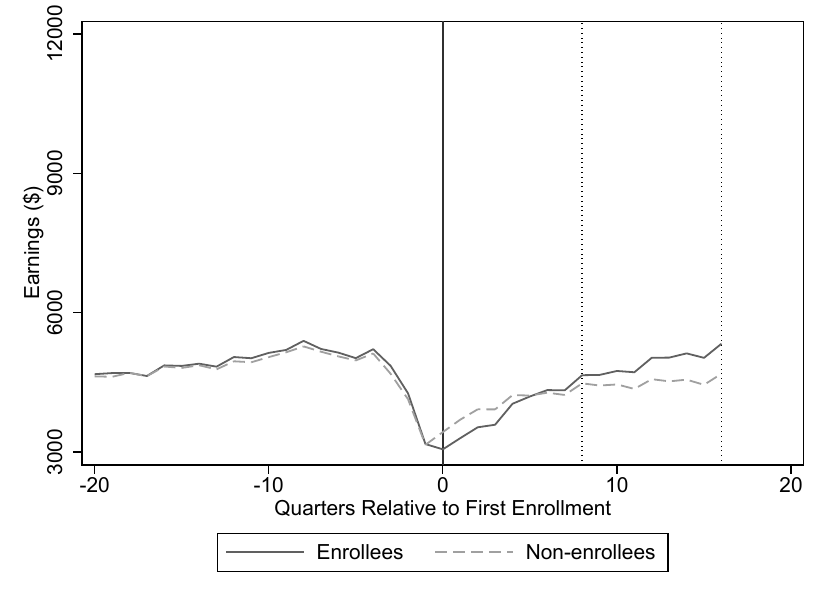}
\par\end{centering}
\begin{centering}
(B) Job Tenure More Than One Year, Less Than Or Equal to Six Years
\par\end{centering}
\begin{centering}
\includegraphics[scale=0.5]{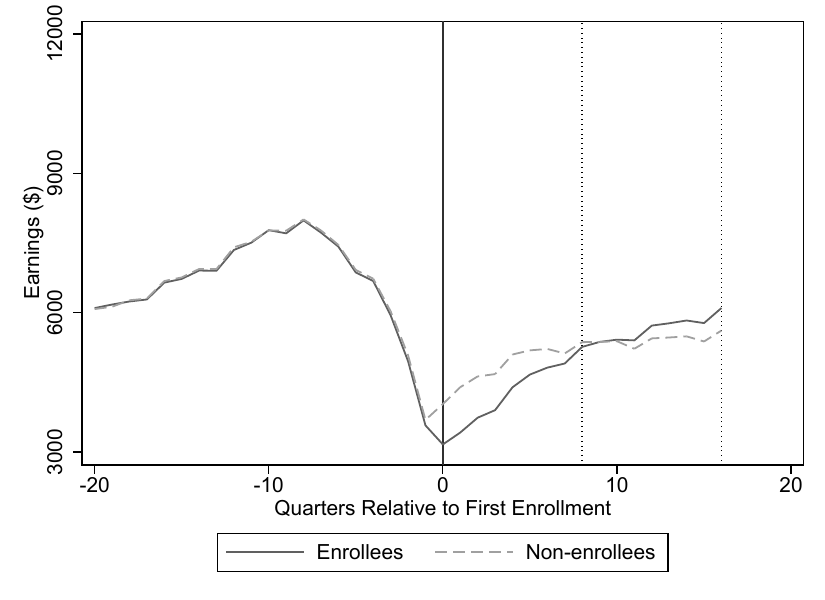}
\par\end{centering}
\begin{centering}
(C) Job Tenure More Than Six Years
\par\end{centering}
\begin{centering}
\includegraphics[scale=0.5]{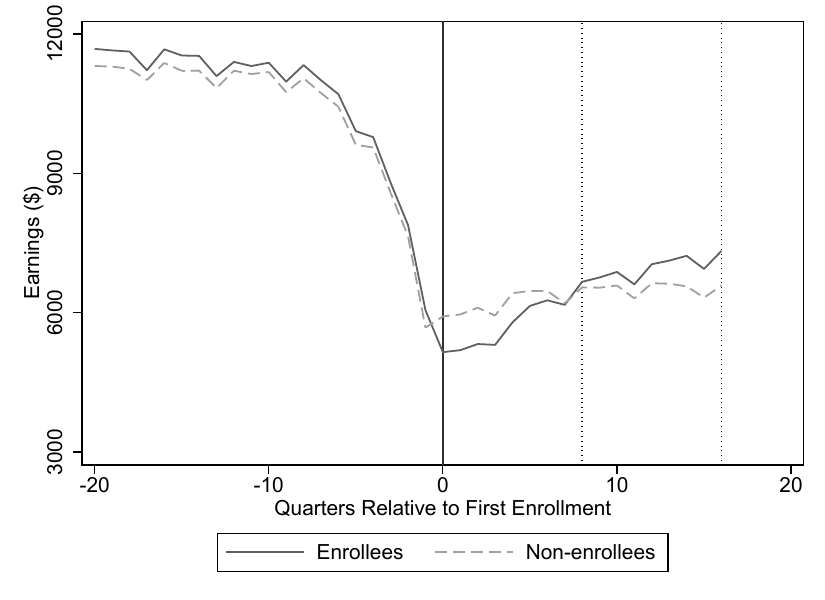}
\par\end{centering}
\centering{}%
\begin{minipage}[t]{0.7\columnwidth}%
Notes: These graphs show the average quarterly earnings of enrollees in different job tenure categories and matched non-enrollee UI claimants. The solid vertical lines denote the quarter of first enrollment, and the vertical dashed lines denote eight and 16 quarters after first enrollment. There are $N=$49,554; 63,902; and 29,984 UI claims in each graph, respectively, corresponding to 47,942; 61,865; and 29,026 unique individuals.
\end{minipage}
\end{figure}

\clearpage{}
\begin{figure}[H]
\caption{Earnings of Enrollees and Matched Non-enrollees, By Race}
\label{fig:subgroups_race}
\begin{centering}
\smallskip{}
\par\end{centering}
\begin{centering}
(A) White
\par\end{centering}
\begin{centering}
\includegraphics[scale=0.85]{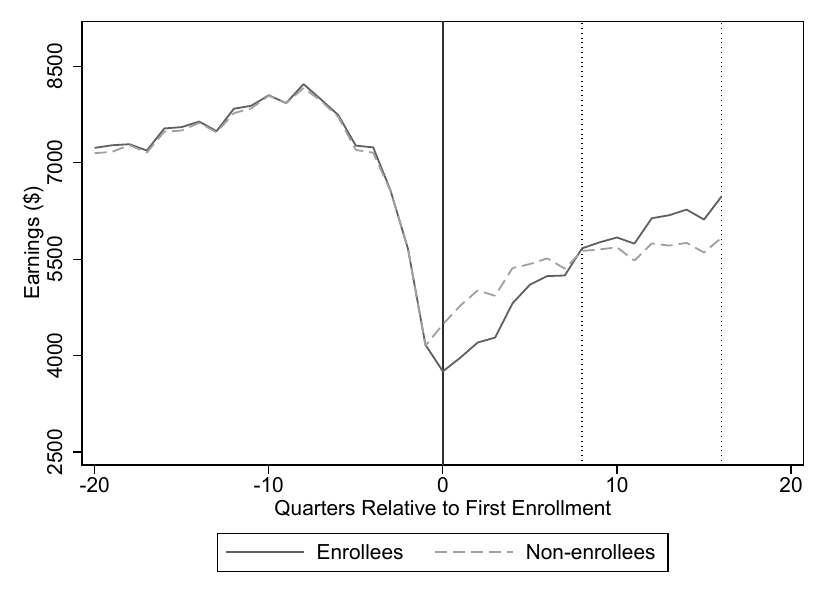}
\par\end{centering}
\begin{centering}
(B) Black
\par\end{centering}
\begin{centering}
\includegraphics[scale=0.85]{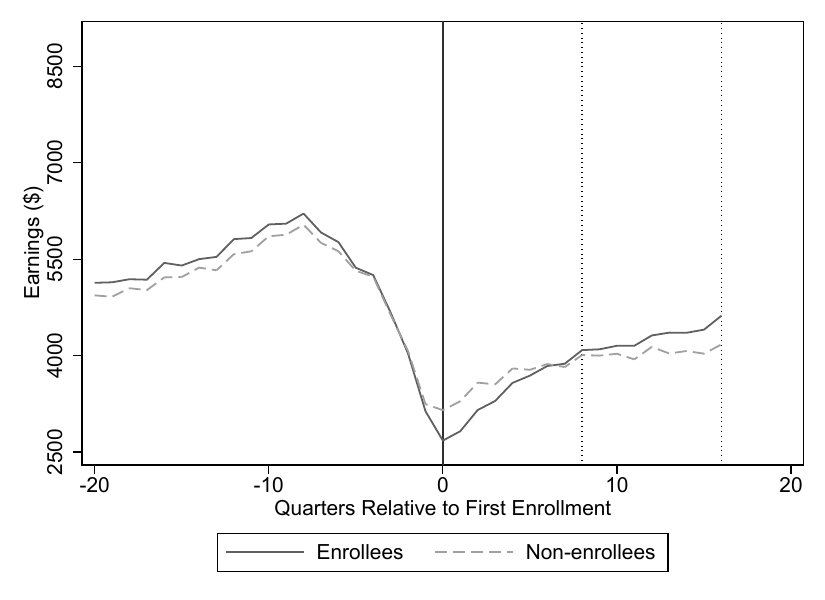}
\par\end{centering}
\centering{}%
\begin{minipage}[t]{0.8\columnwidth}%
Notes: The upper (lower) graph shows the average quarterly earnings of White (Black) enrollees and matched non-enrollee UI claimants. The solid vertical lines denote the quarter of first enrollment, and the vertical dashed lines denote eight and 16 quarters after first enrollment. There are $N=$107,810 (25,292) UI claims in the upper (lower) graph, corresponding to 103,514 (24,037) unique individuals.
\end{minipage}
\end{figure}
\clearpage\clearpage{}
\begin{figure}[H]
\caption{Earnings of Enrollees and Matched Non-enrollees, By Prior College Experience}
\label{fig:subgroups_prioredu}
\begin{centering}
\smallskip{}
\par\end{centering}
\begin{centering}
(A) No Prior College
\par\end{centering}
\begin{centering}
\includegraphics[scale=0.85]{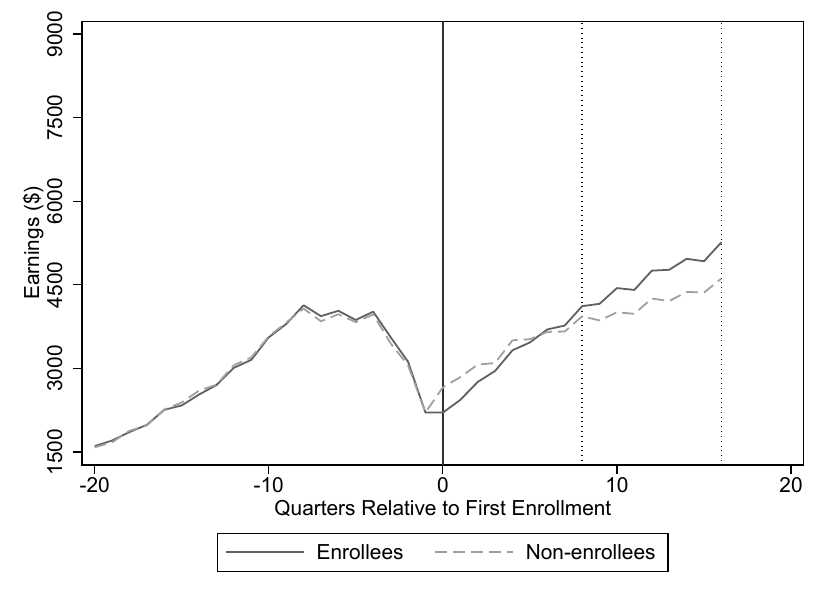}
\par\end{centering}
\begin{centering}
(B) With Prior College
\par\end{centering}
\begin{centering}
\includegraphics[scale=0.85]{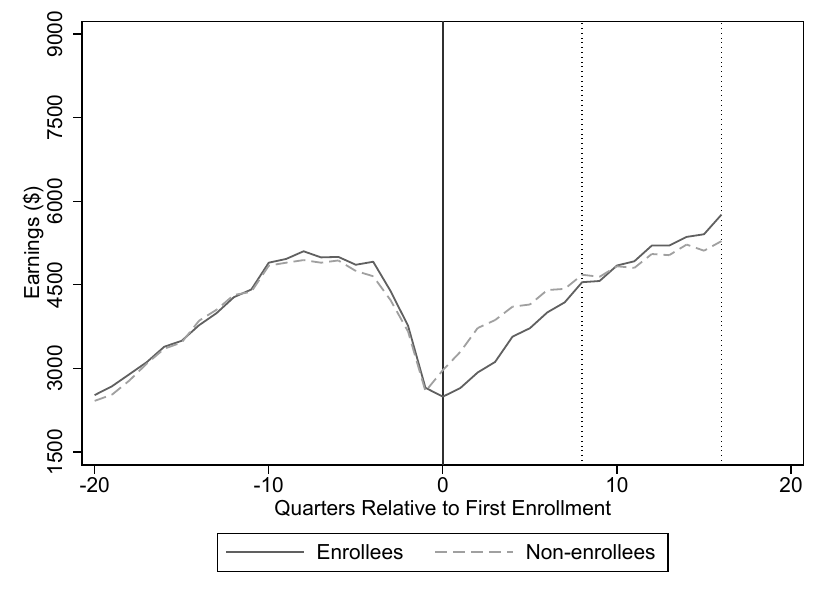}
\par\end{centering}
\centering{}%
\begin{minipage}[t]{0.8\columnwidth}%
Notes: The upper (lower) graph shows the average quarterly earnings of enrollees without (with) prior college experience, and their matched non-enrollee UI claimants. The solid vertical lines denote the quarter of first enrollment, and the vertical dashed lines denote eight and 16 quarters after first enrollment. There are $N=$12,966 (5,212) UI claims in the upper (lower) graph, corresponding to 12,421 (4,862) unique individuals.
\end{minipage}
\end{figure}
\clearpage{}
\begin{figure}[H]
\caption{Earnings of Enrollees and Matched Non-enrollees, By School Type}
\label{fig:schooltypes}
\begin{centering}
\smallskip{}
\par\end{centering}
\begin{centering}
(A) Community College
\par\end{centering}
\begin{centering}
\includegraphics[scale=0.85]{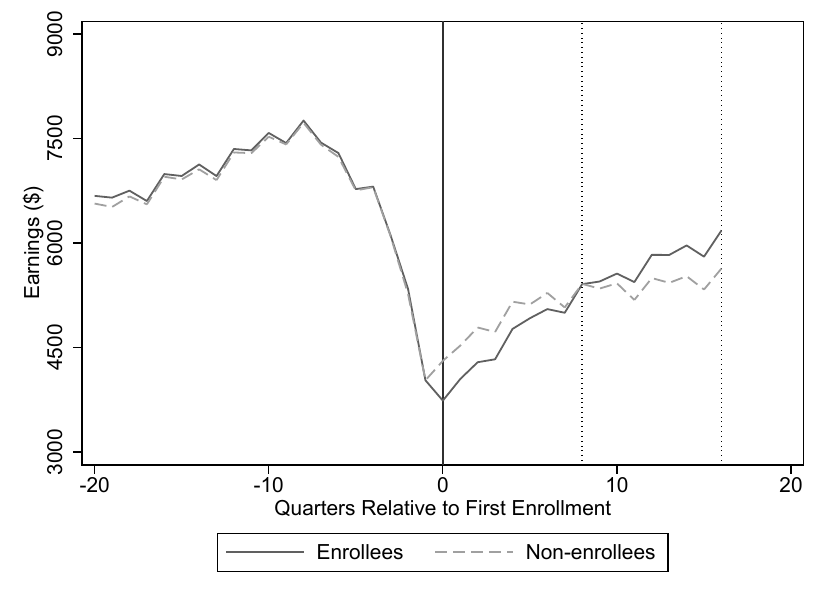}
\par\end{centering}
\begin{centering}
(B) Technical Center
\par\end{centering}
\begin{centering}
\includegraphics[scale=0.85]{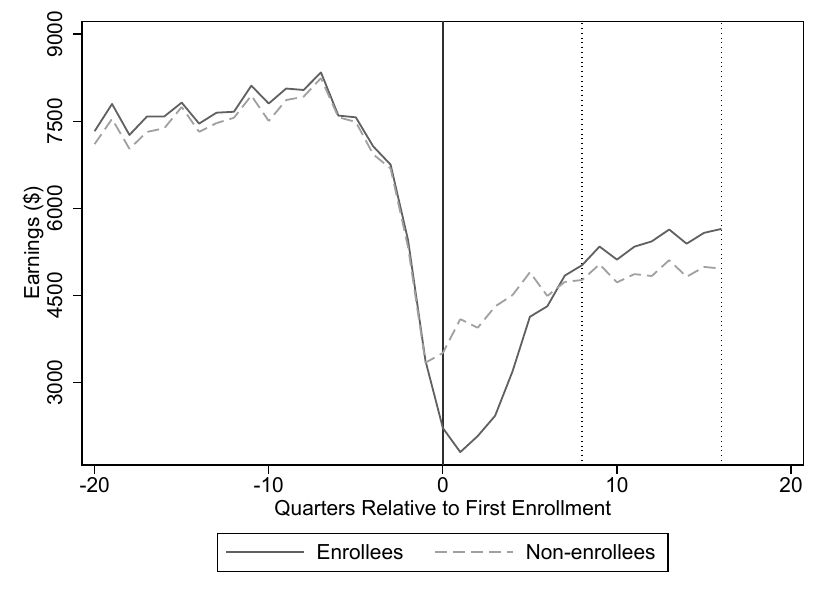}
\par\end{centering}
\centering{}%
\begin{minipage}[t]{0.8\columnwidth}%
Notes: The upper (lower) graph shows the average quarterly earnings of community college (technical center) enrollees and matched non-enrollee UI claimants. The solid vertical lines denote the quarter of first enrollment, and the vertical dashed lines denote eight and 16 quarters after first enrollment. There are $N=$121,078 (15,932) UI claims in the upper (lower) graph, corresponding to 116,691 (15,703) unique individuals.
\end{minipage}
\end{figure}
\clearpage{}
\begin{figure}[H]
\caption{Earnings of Enrollees and Matched Non-enrollees, By Institution Quality}
\label{fig:schoolquality}
\begin{centering}
\smallskip{}
\par\end{centering}
\begin{centering}
(A) Instructional Expenditures Per Student
\par\end{centering}
\begin{centering}
\includegraphics[scale=0.43]{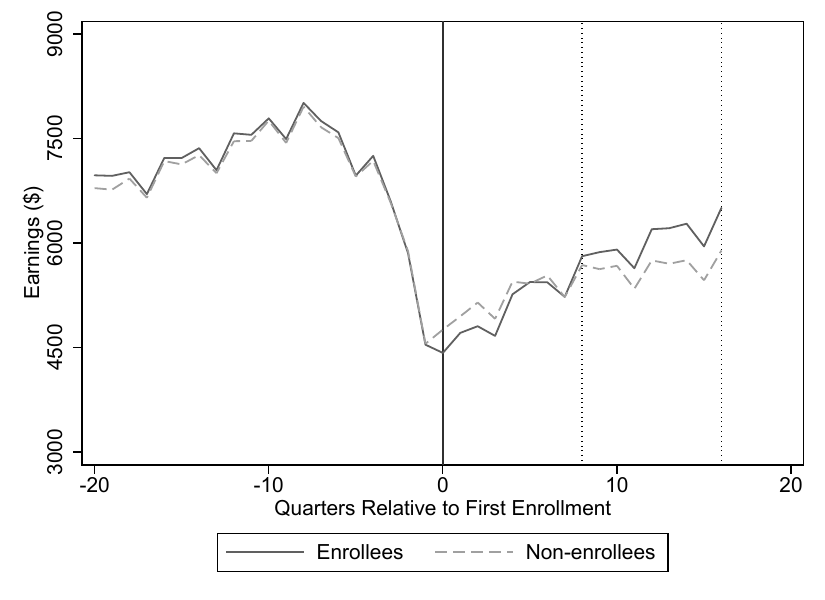} \includegraphics[scale=0.43]{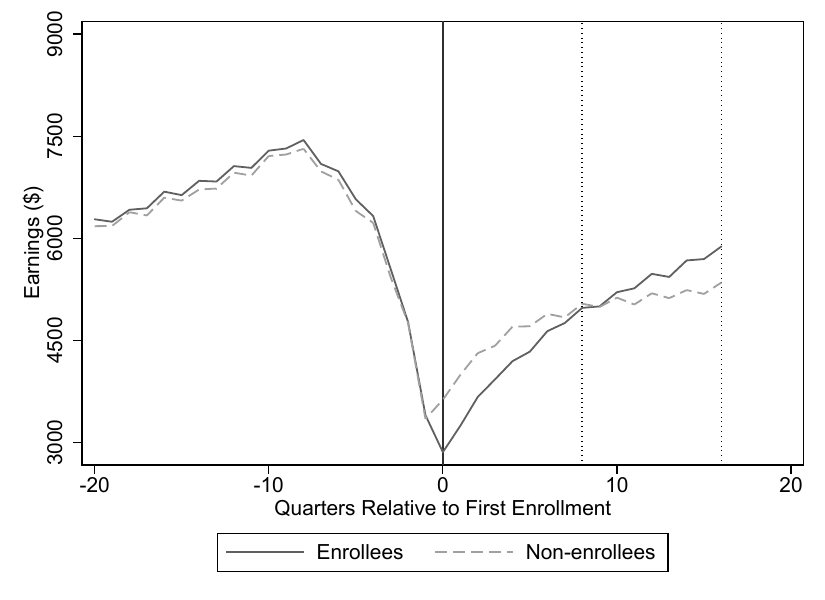}
\par\end{centering}
\begin{centering}
(B) Institutional Completion Rates
\par\end{centering}
\begin{centering}
\includegraphics[scale=0.43]{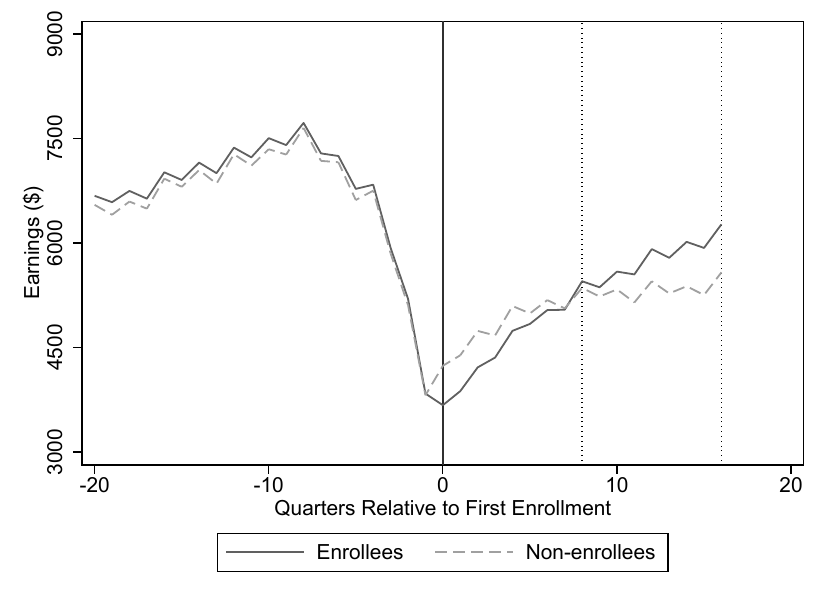} \includegraphics[scale=0.43]{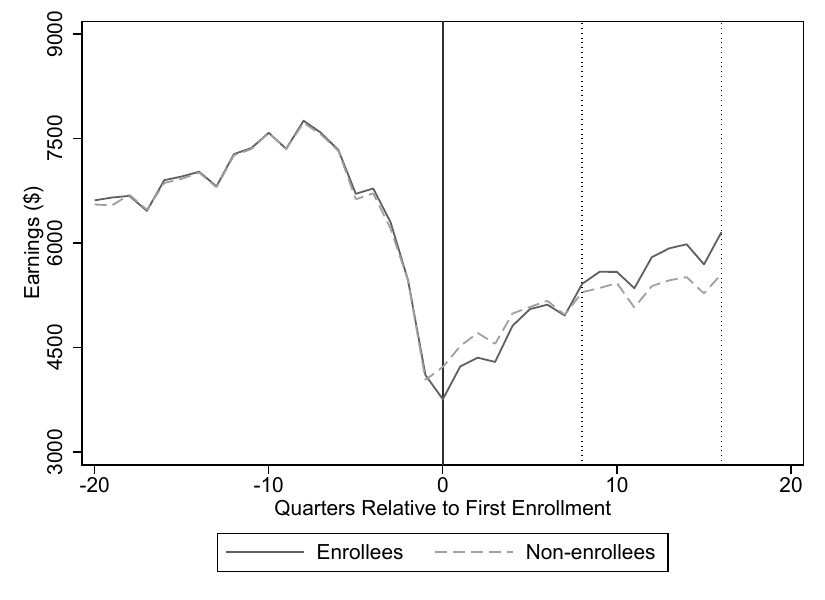}
\par\end{centering}
\begin{centering}
(C) Institutional Earnings
\par\end{centering}
\begin{centering}
\includegraphics[scale=0.43]{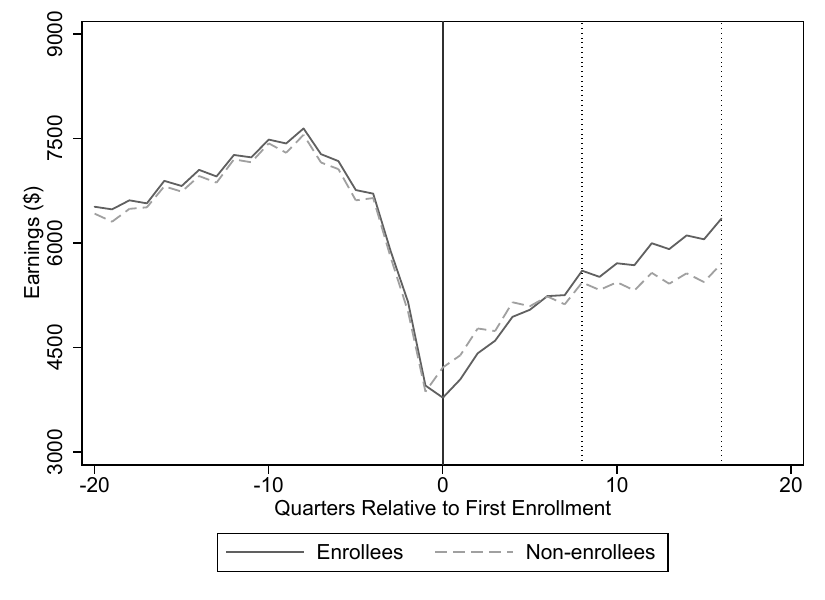}\includegraphics[scale=0.43]{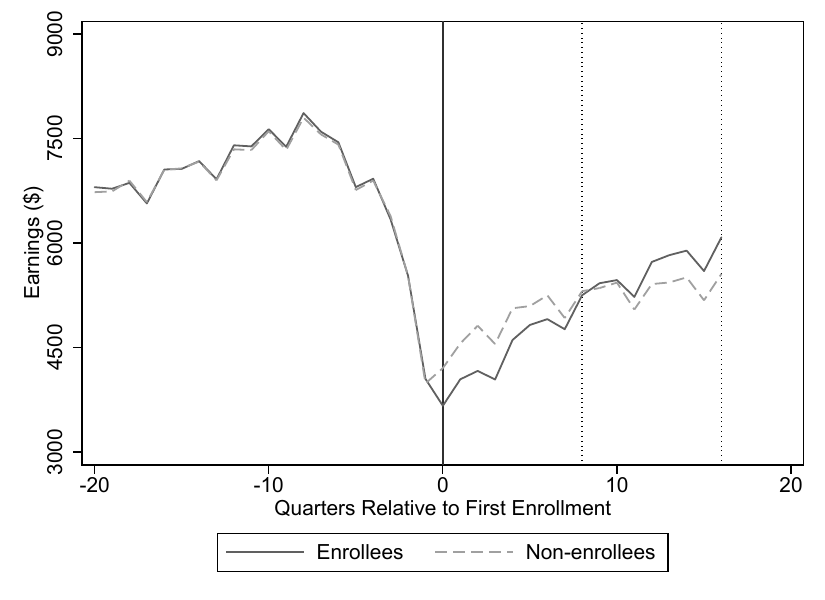}
\par\end{centering}
\centering{}%
\begin{minipage}[t]{0.8\columnwidth}%
{\small{}Notes: The left (right) graph of Panel A shows the average quarterly earnings of enrollees who attend community colleges that have above (below) median instructional expenditures per student, and matched non-enrollee UI claimants. The left (right) graph of Panel B shows analogous results for those who attend community colleges that have above (below) median institutional completion rates, and matched non-enrollee UI claimants. Institutional completion rate is defined as the proportion of entering students who graduated within eight years of entry. The left (right) graph of Panel C shows the analogous results for those who attend community colleges that have above (below) median institutional earnings. Institutional earnings is defined as the median earnings of students working and not enrolled 10 years after entry. The solid vertical lines denote the quarter of first enrollment, and the vertical dashed lines denote eight and 16 quarters after first enrollment. There are $N=$63,348, 50,754, 60,488, 54,914, 57,504, 58,238 UI claims in the upper left, upper right, middle left, middle right, lower left, and lower right graphs, respectively, corresponding to 61,742, 49,789, 59,129, 53,573, 56,261, 57,004 unique individuals.}
\end{minipage}
\end{figure}

\begin{figure}[H]
\caption{Joining Higher AKM Premium Firms Among Enrollees and Matched Non-enrollees}
\label{fig:firmakm-prob}\label{fig:firmakm}
\begin{centering}
\smallskip{}
\par\end{centering}
\begin{centering}
(A) Probability of Switching to a Higher Premium Firm Over Time
\par\end{centering}
\begin{centering}
\includegraphics[scale=0.85]{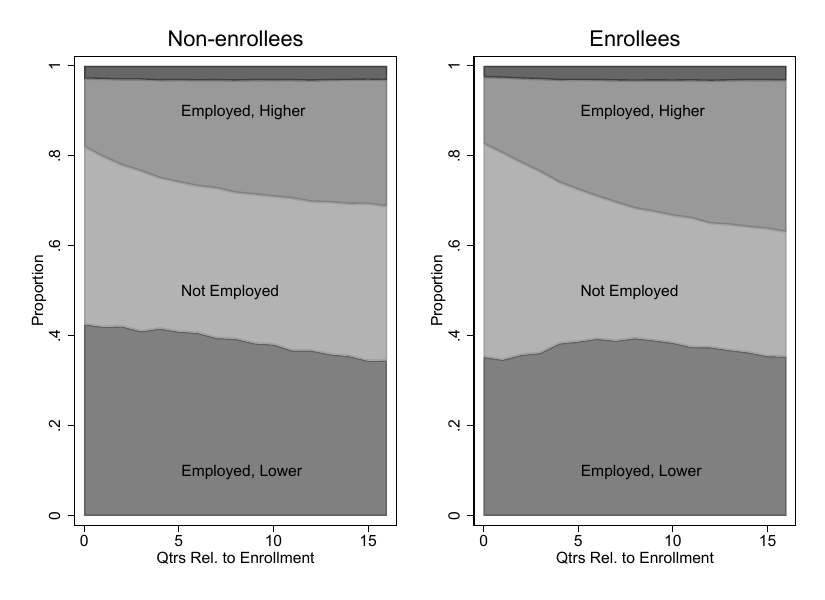}
\par\end{centering}
\begin{centering}
(B) Decomposition of Earnings: the Role of Firm Wage Premium
\par\end{centering}
\begin{centering}
\includegraphics[scale=0.85]{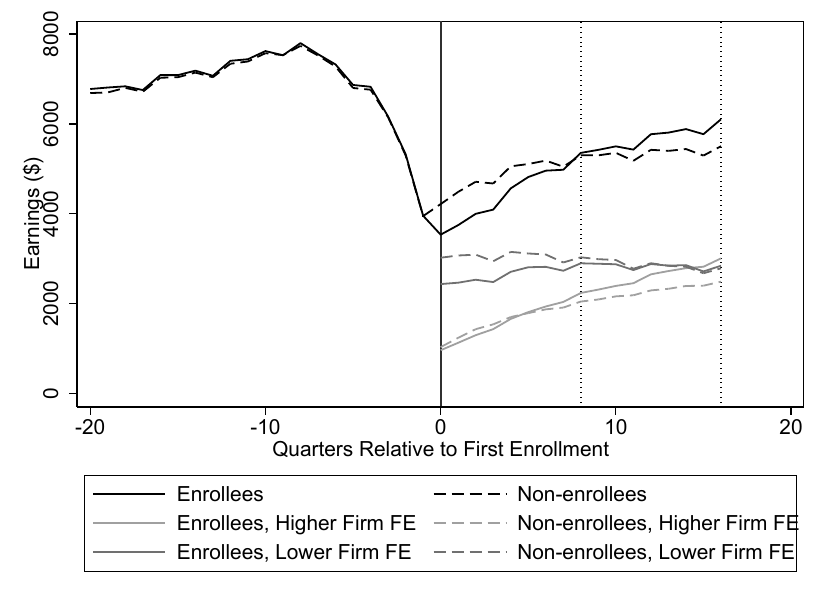}
\par\end{centering}
\centering{}%
\begin{minipage}[t]{0.9\columnwidth}%
Notes: The figures in Panel A plot the probability of employment in a firm with a higher firm AKM premium, employment in a firm with a lower (including same) premium, and non-employment over time for enrollees and matched non-enrollee UI claimants (the dark unlabeled group at the top corresponds to those who are employed but whose prior or current firm AKM effects cannot be estimated). Panel B plots the average quarterly earnings of enrollee and matched non-enrollee UI claimants (black solid and dashed lines). The gray lines disaggregate the post-enrollment earnings into two components: average quarterly earnings from joining firms with higher and lower AKM premia, each scaled by the probability of employment in those two groups, respectively.  $N=141,758$, corresponding to 136,074 unique individuals.
\end{minipage}
\end{figure}

\clearpage{}
\begin{figure}[H]
\caption{Industries of Enrollees and Matched Non-enrollees, 16th Quarter Post-Enrollment}
\label{fig:destination_ind}
\begin{centering}
\smallskip{}
\par\end{centering}
\begin{centering}
(A) Women
\par\end{centering}
\begin{centering}
\includegraphics[scale=0.85]{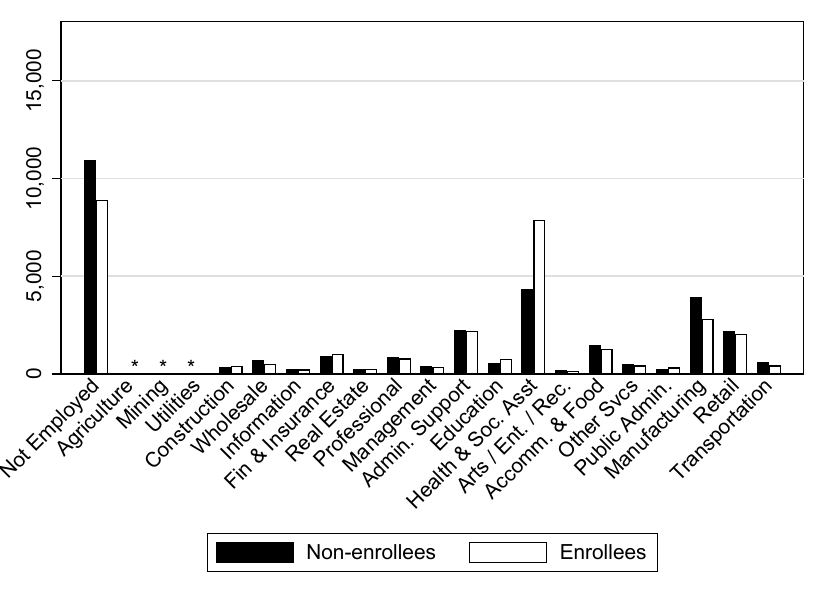}
\par\end{centering}
\begin{centering}
(B) Men
\par\end{centering}
\begin{centering}
\includegraphics[scale=0.85]{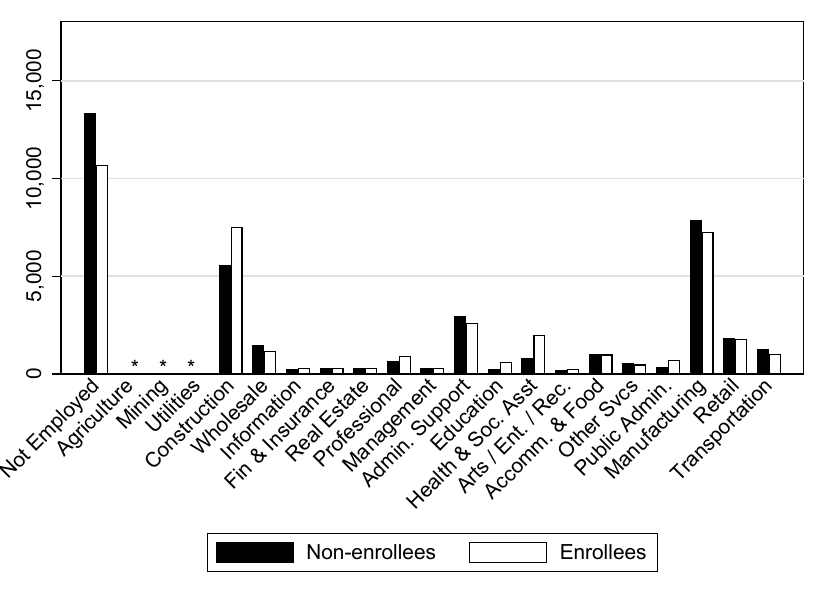}
\par\end{centering}
\centering{}%
\begin{minipage}[t]{0.8\columnwidth}%
Notes: These figures plot the number of enrollee and matched non-enrollee UI claimants employed in each sector (two-digit NAICS) or not employed. Agriculture, Mining, and Utilities sectors have fewer than 200 workers in each enrollee/non-enrollee cell and are not plotted.
\end{minipage}
\end{figure}

\clearpage{}
\begin{figure}[H]
\caption{Enrollees and Matched Non-enrollees Who Switched Industry, 16th Quarter Post-Enrollment}
\label{fig:destination_ind_switching}\smallskip{}

\begin{centering}
(A) Women
\par\end{centering}
\begin{centering}
\includegraphics[scale=0.85]{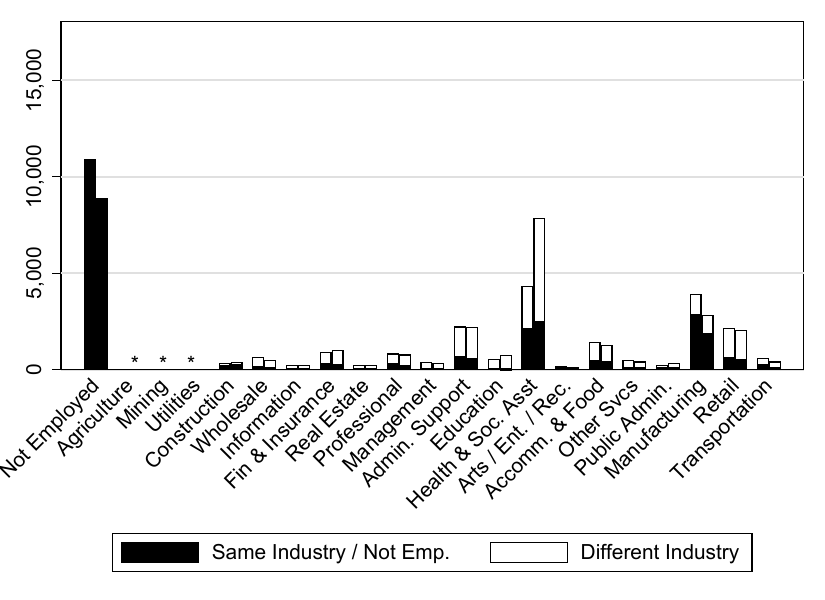}
\par\end{centering}
\begin{centering}
(B) Men
\par\end{centering}
\begin{centering}
\includegraphics[scale=0.85]{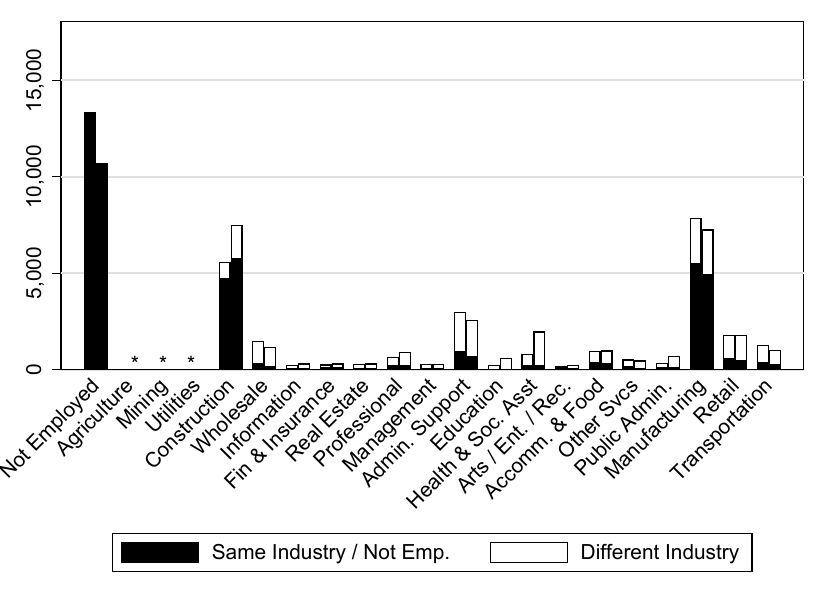}
\par\end{centering}
\centering{}%
\begin{minipage}[t]{0.8\columnwidth}%
Notes: These figures plot the number of enrollees (right bar of each pair) and matched non-enrollees (left bar) employed in each sector (two-digit NAICS) or not employed; the white portion of each bar counts those employed in a different sector than their pre-layoff sector. Agriculture, Mining, and Utilities sectors have fewer than 200 workers in each enrollee/non-enrollee cell and are not plotted.
\end{minipage}
\end{figure}

\clearpage{}
\begin{figure}[H]
\caption{Decomposition of Enrollee Earnings by Completers and Non-completers}
\label{fig:degree_decomp}
\begin{centering}
\smallskip{}
\par\end{centering}
\begin{centering}
\includegraphics[scale=0.85]{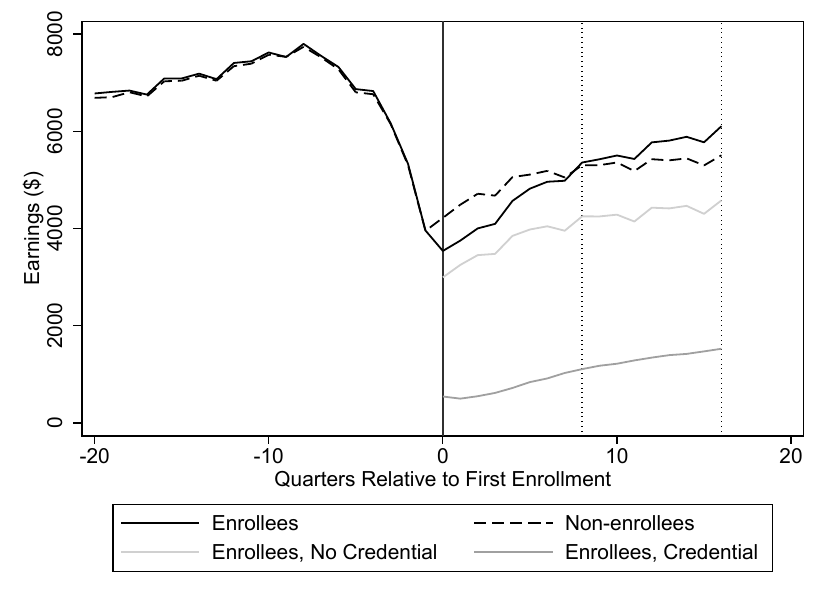}
\par\end{centering}
\centering{}%
\begin{minipage}[t]{0.8\columnwidth}%
{\small{}Notes: This figure plots the average quarterly earnings of enrollee and matched non-enrollee UI claimants (black solid and dashed lines). The gray lines disaggregate the post-enrollment earnings of enrollees into two components: average quarterly earnings of those who eventually obtain a credential and those who do not, each scaled by the proportion obtaining a credential or not. The gray lines sum up to the solid black line. $N=141,758$, corresponding to 136,074 unique individuals.}
\end{minipage}
\end{figure}

\clearpage{}
\begin{figure}[H]
\caption{``Event Study'' Graphs for Fixed Effect Models (Ohio Sample)}
\label{fig:jls_event_ohio}
\begin{centering}
\smallskip{}
\par\end{centering}
\begin{centering}
(A) With Individual Fixed Effects
\par\end{centering}
\begin{centering}
\includegraphics[scale=0.5]{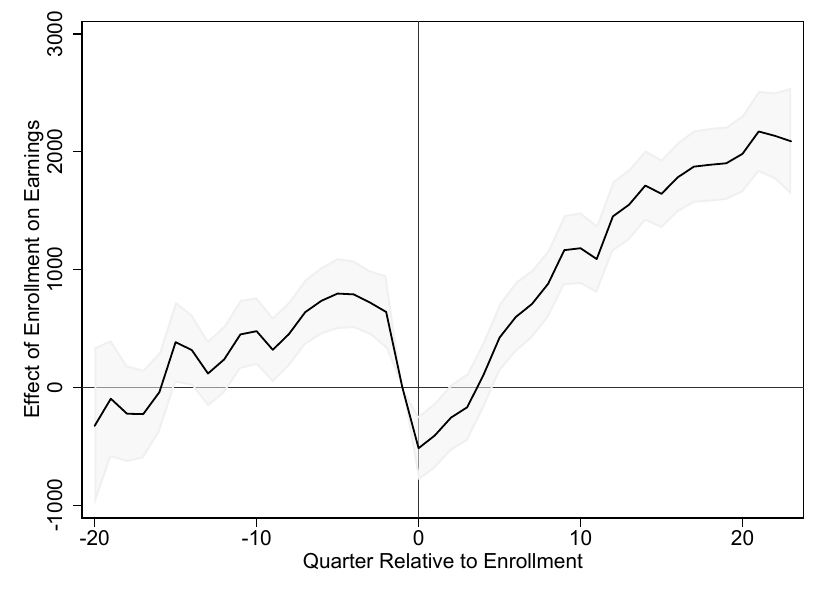}
\par\end{centering}
\begin{centering}
(B) With Individual Trends
\par\end{centering}
\begin{centering}
\includegraphics[scale=0.5]{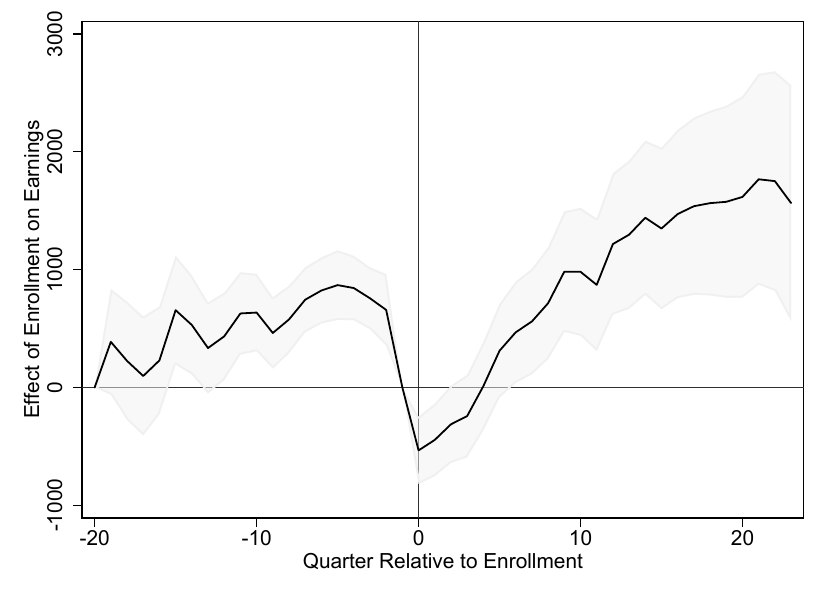}
\par\end{centering}
\begin{centering}
(C) With Individual Trends and Heterogeneous Layoff Effects
\par\end{centering}
\begin{centering}
\includegraphics[scale=0.5]{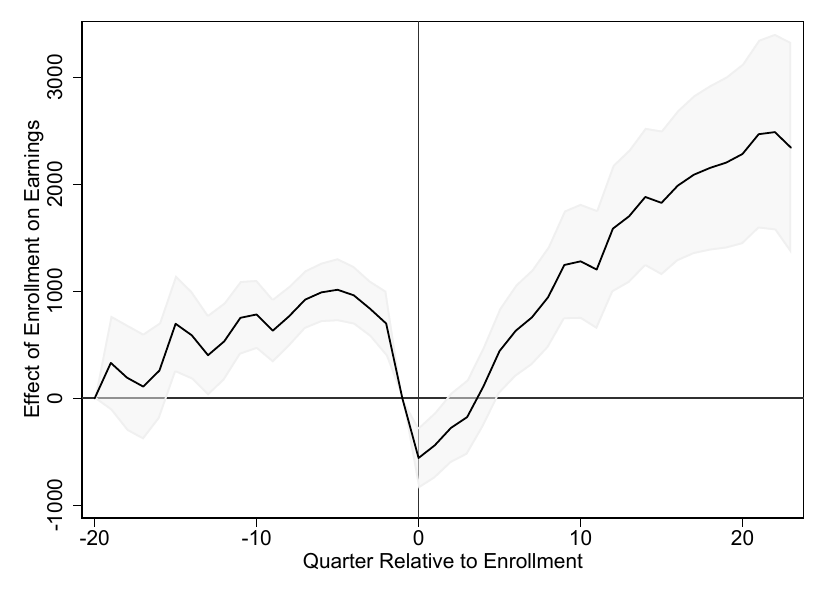}
\par\end{centering}
\centering{}%
\begin{minipage}[t]{0.7\columnwidth}%
Notes: These figures show specification checks of fixed effect models using five percent of our main analysis sample. Each graph plots the estimated ``effect'' of enrollment on earnings. The model in Panel A includes individual fixed effects, quarter fixed effects, and indicators for time relative to layoff. The model in Panel B adds individual time trends. The model in Panel C adds heterogeneous layoff effects. Shaded regions are 95 percent confidence intervals, where standard errors are clustered at the UI claim level. $N=93,528$ UI claims (corresponding to 91,043 unique individuals).
\end{minipage}
\end{figure}

\clearpage{}
\begin{figure}[H]
\caption{Enrollment by UI Claim Date}

\label{fig:enroll_by_claimdate}
\begin{centering}
\smallskip{}
\par\end{centering}
\begin{centering}
\includegraphics{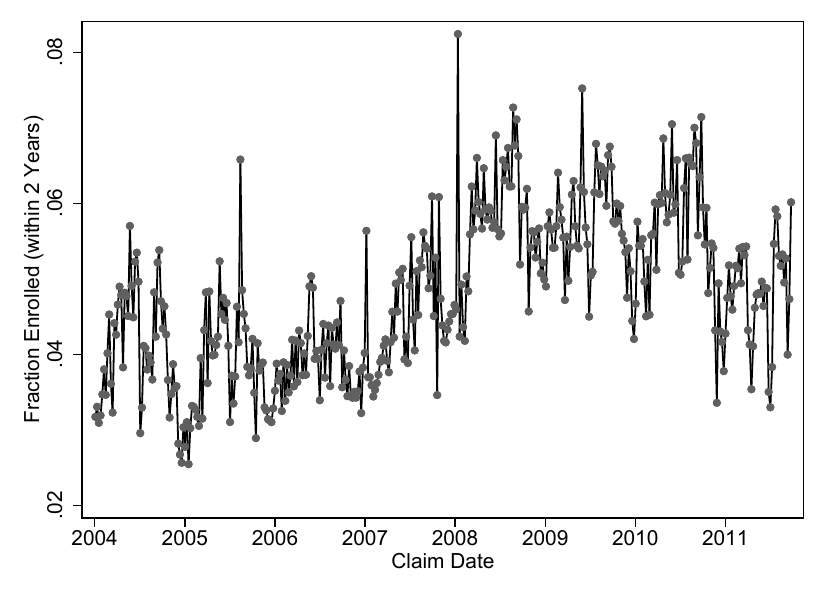}
\par\end{centering}
\centering{}%
\begin{minipage}[t]{0.8\columnwidth}%
Notes: This figure plots the fraction of UI claimants who enroll within two years of claiming UI, by the date of the UI claim. $N=1,994,777$, corresponding to 1,335,958 unique individuals.
\end{minipage}
\end{figure}
\clearpage{}
\begin{figure}[H]
\caption{Ohio UI Extensions, 2004-2012}
\label{fig:potdur_overtime}
\begin{centering}
\includegraphics{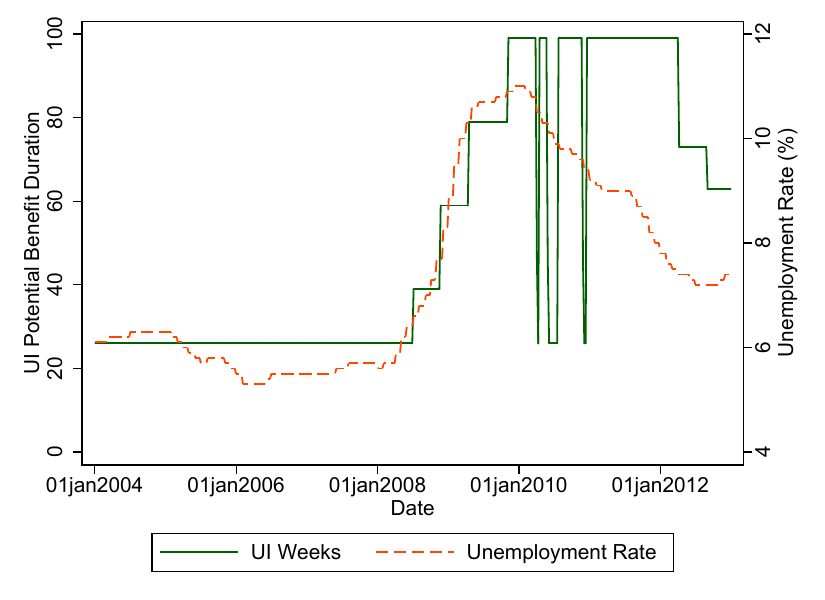}
\par\end{centering}
\centering{}%
\begin{minipage}[t]{0.8\columnwidth}%
Notes: This figure shows the statutory UI benefit duration in weeks (left axis) and the unemployment rate (right axis) in Ohio.
\end{minipage}
\end{figure}

\clearpage{}
\begin{figure}[H]
\caption{Simulated Weeks of UI Remaining to Workers of Different UI Claim Cohorts}
\label{fig:wksleft_sim}
\begin{centering}
\includegraphics{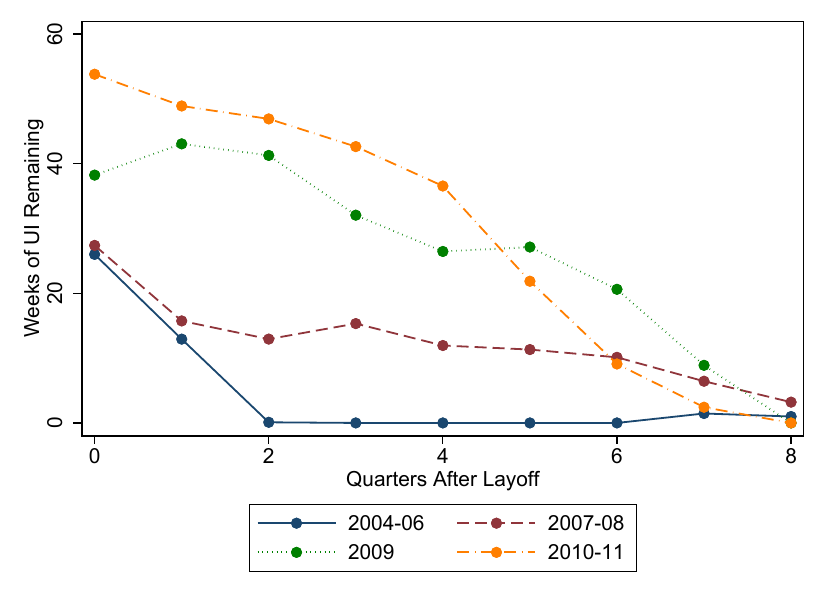}
\par\end{centering}
\centering{}%
\begin{minipage}[t]{0.8\columnwidth}%
Notes: This figure shows the simulated number of weeks of UI benefits remaining to various cohorts of UI claimants in five percent of our analysis sample. $N=99,478$ UI claims (corresponding to 96,874 unique individuals).
\end{minipage}
\end{figure}

\newpage{}
\begin{table}[H]
\caption{Earnings Differences between Later-Enrollees and Matched Non-enrollees}
\label{tab:lower_bound_test}\vspace{10bp}

\noindent \begin{centering}
\includegraphics[viewport=20bp 580bp 612bp 742bp,scale=0.85]{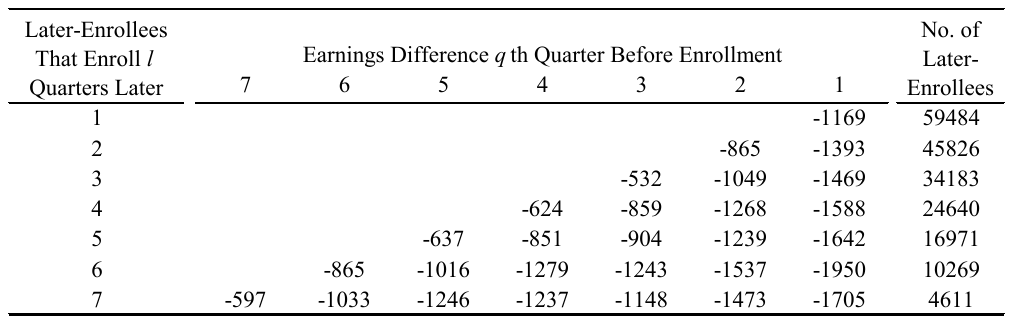}
\par\end{centering}
\centering{}%
\begin{minipage}[t]{0.82\columnwidth}%
Notes: This table presents a test of the assumption that later-enrollees have lower potential earnings than similar non-enrollees (Assumption \ref{assu:Sgeneral_potential_outcome}, which generalizes Assumption \ref{assu:S2_potential_outcome}). Each cell shows the difference in earnings between later-enrollees and matched never-enrollees. Since we define enrollment as enrolling in one of the eight quarters after layoff, later-enrollees can be categorized as enrolling $l=1,\dots,7$ quarters later. That is, later-enrollees among period-1 non-enrollees ($D_{1}=0$) can enroll up to seven quarters later, later-enrollees among period-2 non-enrollees ($D_{1}=D_{2}=0$) can enroll up to six quarters later, and so on. For each of these later-enrollee cohorts (denoted by rows), earnings differences between the later-enrollee and matched never-enrollees are presented. In particular, for later-enrollee cohort $l$, we report earnings difference for each of the $l$ quarters before enrollment, and the columns denote the quarter $q=1,\dots,l$ in which the comparison is made. For example, the upper right number shows that for later-enrollees that enroll one quarter later, earnings are \$1,169 lower than similar never-enrollees in the period before they enroll; the last row shows that average earnings for workers who enroll in 7 quarters (i.e., they are in the subset of period-1 non-enrollees who end up enrolling in quarter 8) are lower than their matched never-enrollees for each of the seven quarters before enrollment.
\end{minipage}
\end{table}
\newpage{}
\begin{table}[H]
\caption{Pre-enrollment Earnings Differences, By Subgroup}
\label{tab:subgroup_balance}\vspace{10bp}

\begin{centering}
\includegraphics[viewport=-1bp 154bp 612bp 742bp,scale=0.85]{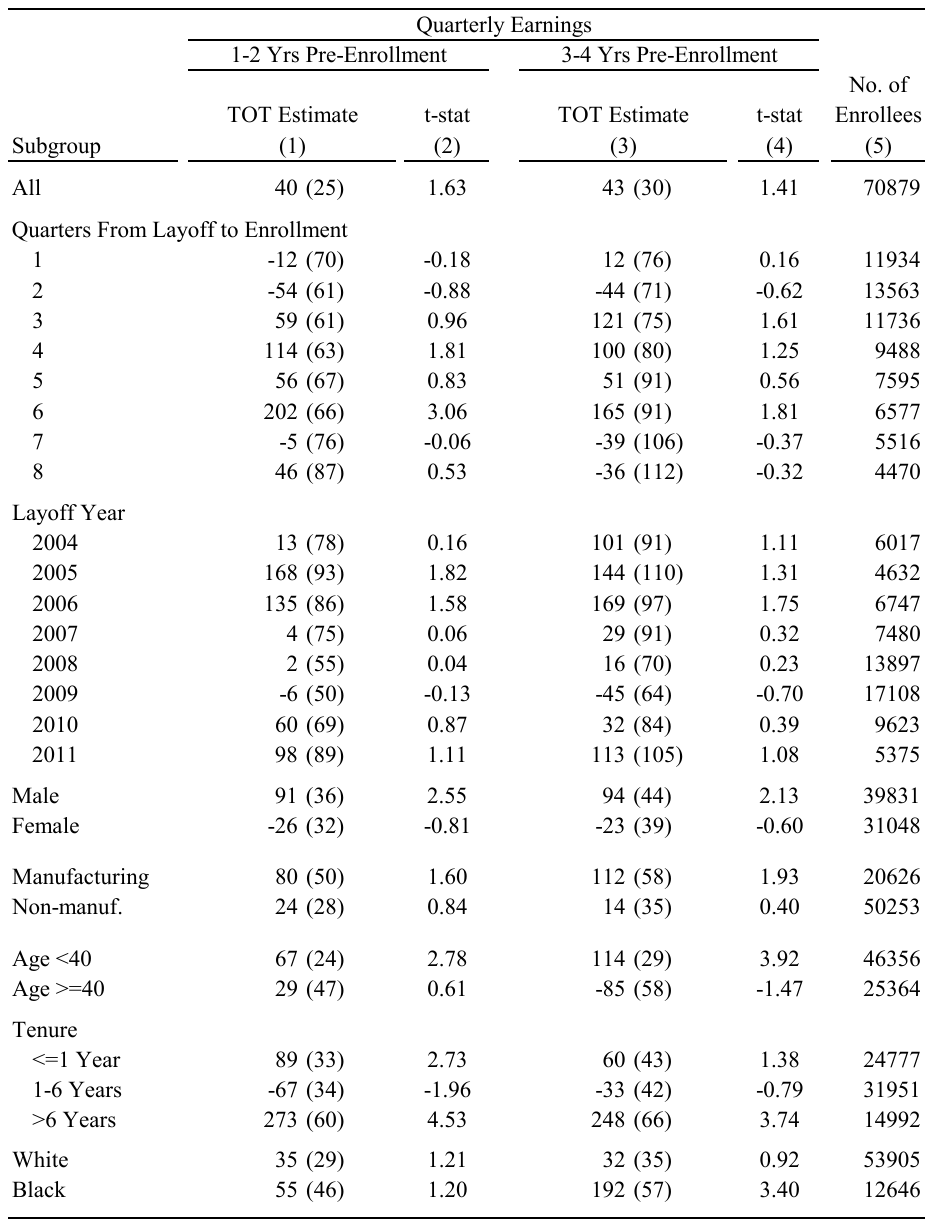}
\par\end{centering}
\centering{}%
\begin{minipage}[t]{0.82\columnwidth}%
Notes: This table presents balance tests for enrollees and matched non-enrollees within subgroups. Columns (1) and (3) show the difference between enrollees and matched non-enrollees in the two years prior to enrollment and three to four years prior to enrollment, respectively. Standard errors (in parentheses) and t-statistics for the mean pairwise difference between enrollees and matched non-enrollees are reported.
\end{minipage}
\end{table}
\newpage{}
\begin{sidewaystable}[H]
\caption{Enrollment Effects by Year of Layoff: Original versus Reweighted Estimates}
\label{tab:recession_reweight}\vspace{-146bp}

\noindent \begin{centering}
\includegraphics[viewport=60bp 220bp 612bp 742bp,scale=0.85]{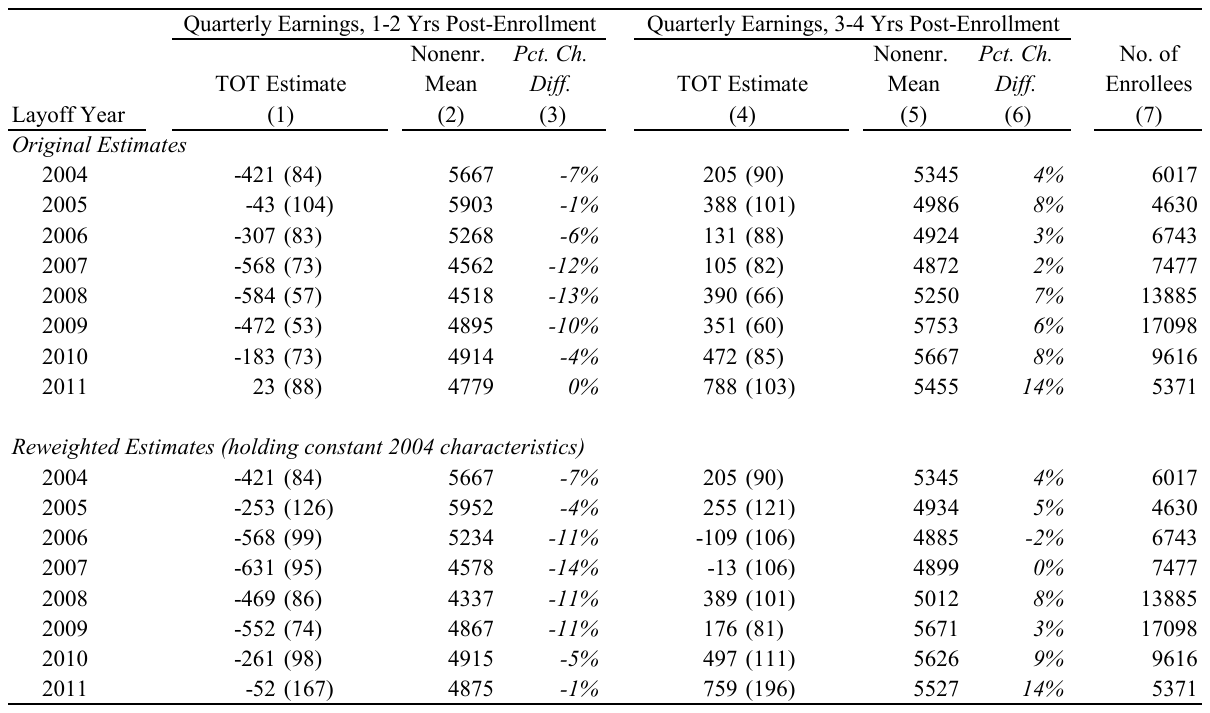}
\par\end{centering}
\centering{}%
\begin{minipage}[t]{0.72\columnwidth}%
Notes: This table shows the enrollment effects by year of layoff. The lower part of the table shows the enrollment effects where observations are reweighted to match the demographic composition of enrollees in the year 2004. Standard errors (in parentheses) for the mean pairwise difference between enrollees and matched non-enrollees are reported.%
\end{minipage}
\end{sidewaystable}
\newpage{}
\begin{table}[H]
\caption{Characteristics of Community College and Technical Center Enrollee Subsamples}
\label{tab:cc_otc_chars}\vspace{10bp}

\begin{centering}
\includegraphics[viewport=160bp 120bp 430bp 742bp,scale=0.75]{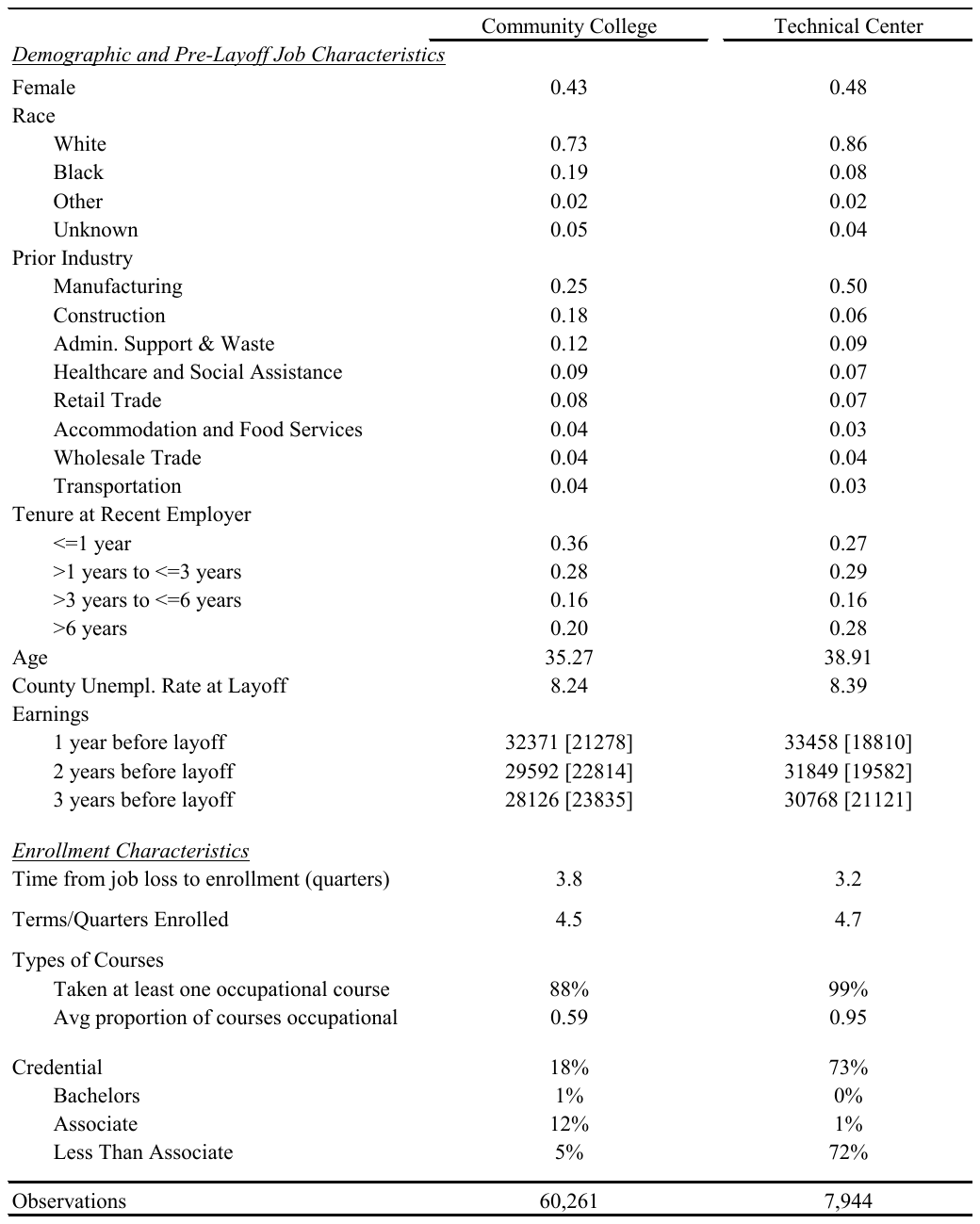}
\par\end{centering}
\centering{}%
\begin{minipage}[t]{0.77\columnwidth}%
Notes: This table presents descriptive characteristics for the subgroups of enrollees whose first institution is either a community college or a technical center. Terms/Quarters Enrolled, Types of Courses, and Credential are calculated within four years of first enrollment. ``Less than Associate'' credentials include less than two-year awards from HEI and any credential from OTC. Standard deviations are in brackets.
\end{minipage}
\end{table}
\newpage{}
\begin{table}[H]
\caption{Characteristics of WIA Enrollee Subsample}
\label{tab:wia_chars}\vspace{10bp}

\begin{centering}
\includegraphics[viewport=180bp 80bp 430bp 742bp,scale=0.75]{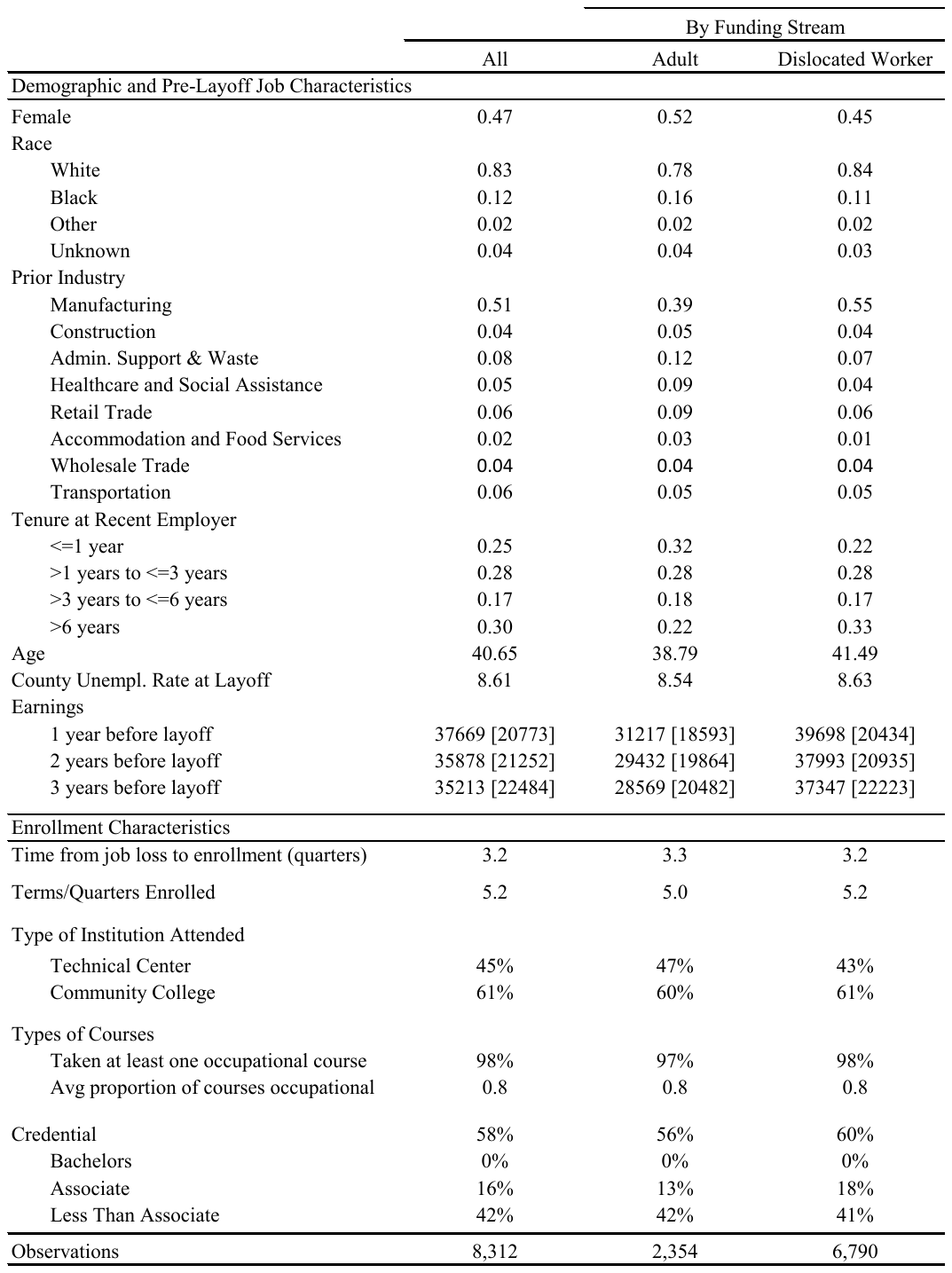}
\par\end{centering}
\centering{}%
\begin{minipage}[t]{0.8\columnwidth}%
\vspace{20bp}
Notes: This table presents descriptive characteristics for the subgroup of enrollees who received WIA training services. Type of Institution Attended, Terms/Quarters Enrolled, Types of Courses, and Credential are calculated within four years of first enrollment. Enrollees may attend more than one type of institution over the four-year period. ``Less than Associate'' credentials include less than two-year awards from HEI and any credential from OTC. Standard deviations are in brackets.
\end{minipage}
\end{table}
\newpage{}
\begin{table}[H]
\caption{Industry of Employment of CPS Enrollees by Education Institution Ownership Type}
\label{tab:cps_private_enrollment}\medskip{}

\begin{centering}
\includegraphics[viewport=170bp 475bp 430bp 742bp,scale=0.75]{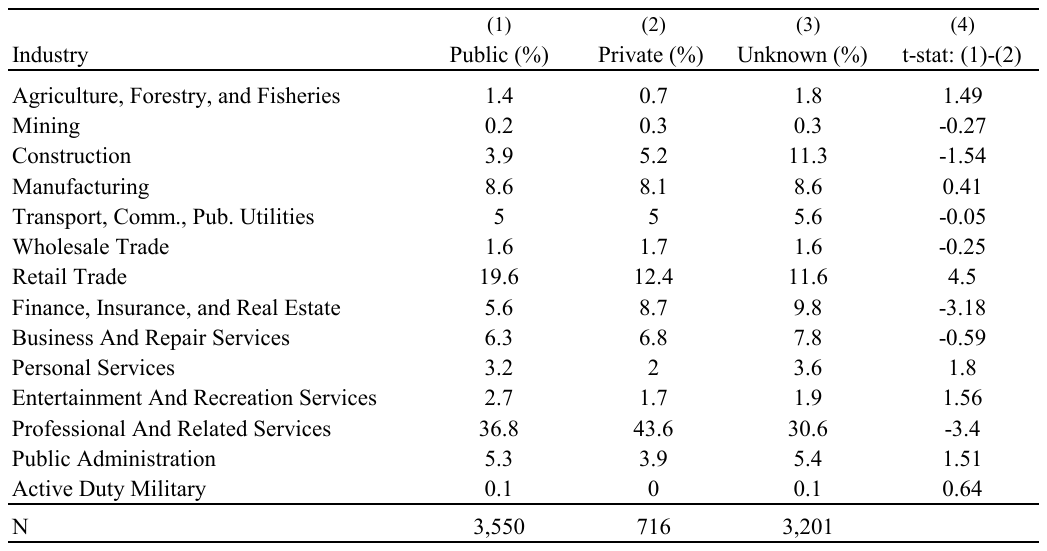}
\par\end{centering}
\centering{}%
\begin{minipage}[t]{0.8\columnwidth}%
Notes: This table presents the industry breakdown of CPS respondents who are in school or taking courses. Specifically, columns (1)-(3) report the past-employment industry shares of CPS respondents who are currently enrolled in institutions of public, private, and unknown ownership. Column (4) reports the t-statistics in the difference between the shares in columns (1) and (2).
\end{minipage}
\end{table}

\newpage{}
\begin{sidewaystable}[H]
\caption{Decomposition of Quarter 16 Earnings Effect}
\label{tab:decompqtr16}\vspace{-146bp}

\noindent \begin{centering}
\includegraphics[viewport=160bp 410bp 612bp 742bp,scale=0.85]{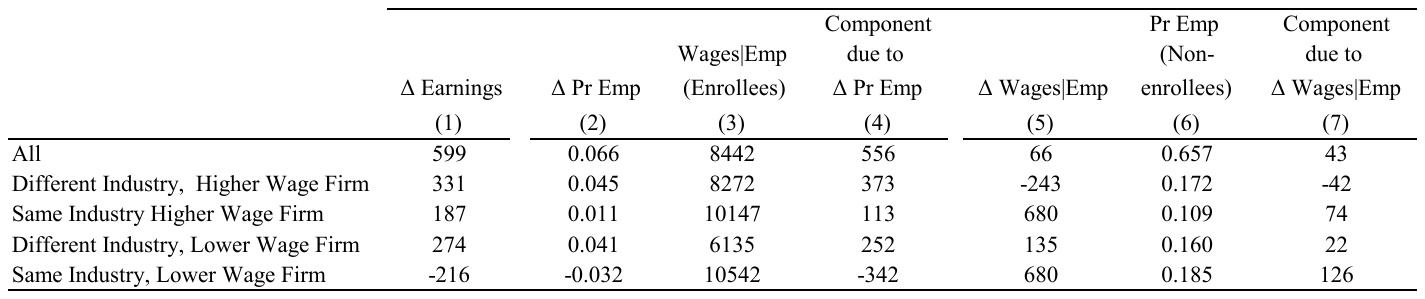}
\par\end{centering}
\centering{}%
\begin{minipage}[t]{0.85\columnwidth}%
Notes: This table decomposes the quarterly earnings difference between enrollees and matched non-enrollees at 16 quarters post-enrollment. The rows decompose the overall earnings effects into components due to different combinations of being employed in a different versus same industry as pre-layoff, and in a firm with a higher versus lower (including same) firm wage premium. Note that rows 2-5 do not add up to row 1 because of missing firm effect estimates. The columns decompose each row into components due to differences in employment and differences in wages conditional on employment. Column 4 (7) is equal to the product of columns 2 (5) and 3 (6). Columns (4) and (7) sum up to column (1). 
\end{minipage}
\end{sidewaystable}

\newpage{}
\begin{table}[H]
\caption{Effect of UI Potential Duration on Enrollment}
\label{tab:ui_extensions}\medskip{}

\begin{centering}
\includegraphics[viewport=30bp 545bp 612bp 740bp]{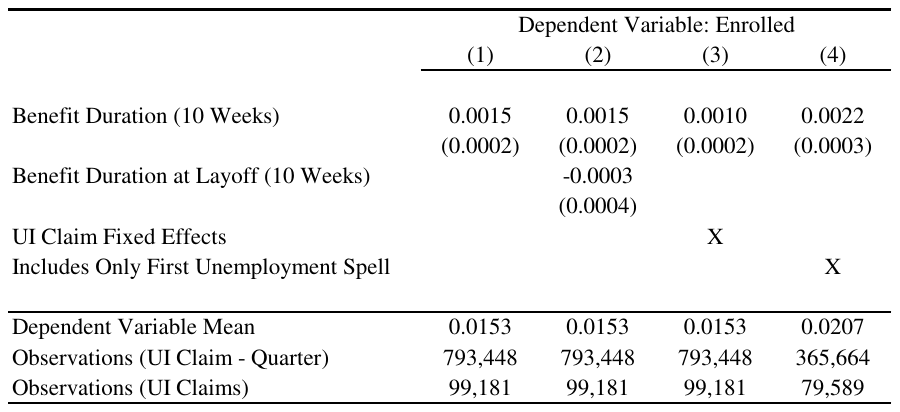}
\par\end{centering}
\centering{}%
\begin{minipage}[t]{0.9\columnwidth}%
Notes: This table shows the estimated effect of UI benefit durations on enrollment over the eight quarters after filing a UI claim. Additional controls include: (all columns) indicators for quarters post-layoff, year indicators, quarter-in-year indicators, quarterly state unemployment rate (quadratic); (all columns except (3)) female, 10-year age category, Black, Hispanic, indicator for having dependents, prior year wage quintile, tenure category, 2-digit prior industry, 2-digit prior occupation. Standard errors are clustered by individuals and in parentheses. There are 96,583 unique individuals in the regressions in columns (1)-(3).
\end{minipage}
\end{table}

\end{document}